\documentclass[11pt,reqno]{amsart}

\usepackage{etex}
\usepackage{amsmath}
\usepackage{amssymb}
\usepackage{amsfonts}
\usepackage{amsthm}
\usepackage{a4wide}
\usepackage{bm}
\usepackage{graphics}
\usepackage{color}
\usepackage[sans]{dsfont}
\usepackage[toc,page]{appendix}
\usepackage{anyfontsize}
\usepackage{hyperref}
\usepackage{protosem}
\usepackage{linearA}
\usepackage{pgfkeys}
\usepackage{longtable}
\usepackage{pdflscape}
\usepackage{float}
\usepackage{cancel}
\usepackage{dsfont}
\usepackage{setspace}
\usepackage{array}
\usepackage{hyperref}
\usepackage{appendix}
\usepackage{amssymb}
\usepackage{epstopdf}
 \usepackage{multirow}
\usepackage{epsfig}
\usepackage{lscape}
\usepackage{yhmath}
\usepackage[english]{babel}
\usepackage[latin1]{inputenc}
\usepackage[T1]{fontenc}
\usepackage{wasysym}
\usepackage{tikz,pgflibraryshapes}
\usepackage{pstricks,pst-plot,pstricks-add}
\usepackage{enumerate}
\usetikzlibrary{arrows}
\usetikzlibrary{snakes}
\usetikzlibrary {shapes}
\usepackage{textcomp}
\usepackage{mathrsfs}
\usepackage[toc,page]{appendix}

\newtheorem{theorem}{Theorem}[section]
\newtheorem{lemma}[theorem]{Lemma}
\newtheorem{proposition}[theorem]{Proposition}

\newtheorem{definition}[theorem]{Definition}

\newtheorem*{remark}{Remarks}

\numberwithin{equation}{section}

\numberwithin{equation}{section}

\newcommand{\cH}{\mathcal{H}}

\title[Quantum Electrodynamics of Atomic Resonances]{Quantum Electrodynamics of Atomic Resonances} 
\begin{document}

\author[M. Ballesteros]{ Miguel Ballesteros}
\address[M.Ballesteros]{Department of Mathematical Physics, Applied Mathematics and Systems Research Institute (IIMAS), National Autonomous University of Mexico (UNAM), Campus CU (University City), 01000 Mexico City, Mexico}
\email{miguel.ballesteros@iimas.unam.mx}
\author[J. Faupin]{J{\'e}r{\'e}my Faupin}
\address[J. Faupin]{Institut Elie Cartan de Lorraine \\
Universit{\'e} de Lorraine, 
57045 Metz Cedex 1, France}
\email{jeremy.faupin@univ-lorraine.fr}
\author[J. Fr\"ohlich]{J\"urg Fr\"ohlich}
\address[J. Fr{\"o}hlich]{Institut f{\"u}r Theoretische Physik, ETH H{\"o}nggerberg, CH-8093 Z{\"u}rich, Switzerland}
\email{juerg@phys.ethz.ch}
\author[B. Schubnel]{Baptiste Schubnel}
\address[J. Faupin]{Institut Elie Cartan de Lorraine \\
Universit{\'e} de Lorraine, 
57045 Metz Cedex 1, France}
\email{baptiste.schubnel@univ-lorraine.fr}
 
\date \today

\maketitle

\begin{abstract}
 A simple model of an atom interacting with the quantized electromagnetic field is studied. The atom has a finite mass $m$, finitely many excited states and an electric dipole moment, $\vec{d}_0 = -\lambda_{0} \vec{d}$, where $\| d^{i}\| = 1,$  $ i=1,2,3,$ and $\lambda_0$ is proportional to the elementary electric charge. The interaction of the atom with the radiation field is described with the help of the Ritz Hamiltonian, 
$-\vec{d}_0\cdot \vec{E}$, where $\vec{E}$ is the electric field, cut off at large frequencies. A mathematical study of the Lamb shift, the decay channels and the life times of the excited states of the atom is presented. It is rigorously proven that these quantities are analytic functions of the momentum $\vec{p}$ of the atom and of the coupling constant $\lambda_0$, provided $\vert\vec{p}\vert < mc$ and $\vert \Im\vec{p} \vert$ and 
$\vert \lambda_{0} \vert$ are sufficiently small. The proof relies on a somewhat novel inductive construction involving a sequence of `smooth Feshbach-Schur maps' applied to a complex dilatation of the original Hamiltonian, which yields an algorithm for the calculation of resonance energies that converges super-exponentially fast. 
\end{abstract}

\section{Introduction}
\label{intro}

This paper is devoted to a study of atomic resonances, in particular of the Lamb shift, the decay channels and the life times of excited states of atoms, in quantum electrodynamics. Our analysis is based on a variant of the so-called Pauli-Fierz model of quantum electrodynamics. The atomic degrees of freedom are treated non-relativistically, but photons are massless, and no infrared cutoff is imposed on the interactions between atoms and the quantized radiation field. In order to avoid technical complications that might hide the basic simplicity and elegance of this work, we focus on a somewhat mutilated model of an atom: The mass, $m$, of an atom is positive and finite, its kinematics is non-relativistic, but it has only finitely many excited states and cannot be ionized. The total electric charge of every atom vanishes and its interaction with the quantized radiation field arises by coupling its electric dipole moment to the quantized electric field, (i.e., the interaction Hamiltonian is given by  $-\vec{d}_0 \cdot \vec{E}$, where $\vec{d}_0$ is the dipole moment operator of the atom and $\vec{E}$ is the quantized electric field).  While the masslessness of photons makes a straightforward application of perturbation theory impossible, this model does not exhibit a genuine ``infrared catastrophe''. 

Our main aim, in this paper, is to determine the radiative corrections to the ground-state dispersion law of an atom and to calculate atomic resonance energies, decay channels and life times of excited states. Among our new results are proofs of real analyticity of these quantities as functions of the total momentum, $\vec{p}$, of the dressed atom,  for $\vert\vec{p}\vert < mc$ (where $c$ is the speed of light) and of analyticity in the elementary electric charge near the origin. 

During the past twenty years, there has been very impressive progress in the mathematical analysis of atomic ground-states and resonances in the realm of the Pauli-Fierz model, (and of Rayleigh scattering, ionization of atoms, etc.); see \cite{1995}, \cite{BaFrSi98_02}, \cite{BaFrSi98_01}, \cite{BFS-1999}, \cite{GLL}, \cite{Si09_01}, \cite{BaBaPi}, and references given there. However, the atomic nucleus has usually been treated as static (infinitely heavy). Our goal, in this paper, is to remove this shortcoming.

The main mathematical tools we will employ to prove our main results are based on a combination of dilatation analyticity with a novel method of ``spectral renormalization'' (in the guise of an inductive construction based on a sequence of smooth iso-spectral Feshbach-Schur maps). In the form needed in the analysis of quantum systems with infinitely many degrees of freedom, these tools were first introduced in \cite{1995} and systematically developed in
\cite{BaFrSi98_02}, \cite{BFS-1999} and \cite{BaChFrSi03_01}. Important refinements of these methods have appeared in \cite{GrHa08_01}, \cite{GrHa09_01}, \cite{HaHe10_01}, \cite{Si09_01}, \cite{AmGrGu06_01}, \cite{FFS}; and references given there. Some alternative methods have been introduced in \cite{AbHa12_01} and \cite{BaBaPi}.

In previous work, as quoted above, spectral renormalization is cast in the form of a renormalization group construction involving iteration of a renormalization map (constructed from a Feshbach-Schur map that lowers an energy scale by a fixed factor $\rho<1$), which maps a suitably chosen Banach space of effective Hamiltonians on Fock space into itself. One then attempts to determine the fixed points and the stable and unstable manifolds of the renormalization map -- in accordance with the general philosophy of the renormalization group. While this approach is conceptually transparent and yields very detailed information on the spectral problem under consideration, it leads to certain somewhat artificial technical complications. In this paper, we present an \textit{inductive construction} involving a sequence of smooth iso-spectral Feshbach-Schur maps indexed by a sequence of energy scales that converges to $0$ in a ``super-exponential'' fashion. One of our main aims in this paper is to describe this method and to demonstrate its basic simplicity and efficiency on an example that is of interest to physicists.

\subsection{The Model} 
In this section we describe the physical model studied in this paper in precise mathematical terms. 

\subsubsection{A Simple Model of an ``Atom''} 
Our model of an atom is non-relativistic. For simplicity, the atom is assumed to have only finitely many excited states.
We describe its internal degrees of freedom by an $ N $-level system: The Hilbert space of state vectors of the internal degrees of freedom is given by $\mathbb{C}^N $,
their Hamiltonian by an $N \times N $ matrix  
\begin{equation*}
\label{H1}
H_{is}:= \begin{pmatrix}  E_N & \cdots & 0   \\ \vdots & \ddots  &  \vdots \\ 0 & \cdots & E_1  \end{pmatrix},
\end{equation*}
with $E_N > \cdots > E_1$. The energy scale of transitions between internal states of the atom is measured by the quantity
\begin{equation}\label{delta0}
\delta_0 := \min_{i \ne j }{|E_i - E_j|}.
\end{equation}
By $\vec x \in \mathbb{R}^3$ we denote the position of the (center of mass of the) atom in physical space. The center of mass momentum corresponds to the operator $-i \vec{\nabla}$, the kinetic energy of the free center of mass motion is given by $-\frac{1}{2m}\Delta$. These operators act on the usual Hilbert space 
$L^{2}(\mathbb{R}^3)$ of orbital wave functions. The total Hilbert space of state vectors of the 
atom corresponds to the tensor product 
\begin{equation*}
\mathcal{H}_{at} : = L^{2}(\mathbb{R}^{3}) \otimes \mathbb{C}^N.
\end{equation*}
The total Hamiltonian of the atom is given by the self-adjoint operator 
\begin{equation}\label{hatsn}
H_{at} : =  - \frac{1}{2m}\Delta  +  H_{is},  
\end{equation}
with domain $D(H_{at}) =H^2(\mathbb{R}^3)  \otimes \mathbb{C}^N$, where  $H^2(\mathbb{R}^3)$ denotes the Sobolev space of wave functions with square-integrable derivatives up to order $2$. 
To simplify our notations, we henceforth put the mass of the atom to $1$, which merely amounts to choosing a suitable system of units.  

The electric dipole moment of an atom is represented by the vector operator
\begin{equation}\label{dip}
\vec{d}_0 = (d_{0}^{1}, d_{0}^{2}, d_{0}^{3}),
\end{equation} 
where, for $j=1,2,3, d_{0}^{j} \equiv \mathds{1} \otimes d_{0}^{j}$ is an $N \times N$ hermitian matrix. 

\subsubsection{The Quantized Electromagnetic Field} 
In the following $\vec{k}$ denotes the wave vector of a photon and $\lambda$ its helicity.
To simplify our formulae we define 
\begin{equation}\label{notation_R_souligne}
\underline{\mathbb{R}}^3:= \mathbb{R}^3\times\{1,2\} = \left\{\underline{k} := (\vec{k},\lambda) \: | \: 
\vec k \in\mathbb{R}^3, \lambda  \in \{1,2\} \right\}.
\end{equation}
We set $\underline{\mathbb{R}}^{3n}:=(\underline{\mathbb{R}}^3)^{\times n}$, and, for $B\subset\mathbb{R}^3$,   
\begin{equation}
\label{short2}
\underline{B}:=B\times\{1,2\}, \qquad \int_{\underline{B}} d \underline{k}(\cdot):=
\sum_{\lambda=1,2}\int_{B} d \vec{k}(\cdot).
\end{equation}
The Hilbert space of states of photons is given by
\begin{equation} \label{hsfhotons}
\mathcal{H}_f : = \mathcal{F}_{+}(L^{2}( \underline{\mathbb{R}}^3 )),
\end{equation}
where $ \mathcal{F}_{+}(L^{2}( \underline{\mathbb{R}}^3 ))$ is the symmetric Fock space over 
the space $L^{2}( \underline{\mathbb{R}}^3 )$ of one-photon states. \\
The usual photon creation- and annihilation operators are denoted by
\begin{equation} \label{aastar}
a^*(\underline{k}) \equiv  a^*_\lambda(\vec k), \hspace{2cm} a(\underline{k}) \equiv a_\lambda(\vec k),
\end{equation}
$\underline{k} \equiv (\vec{k},\lambda) \in \underline{\mathbb{R}}^{3}$, which are operator-valued distributions acting on $\mathcal{H}_f$.  The Fock space $\mathcal{H}_f$ contains a unit vector, 
$\Omega$, called ``vacuum (vector)'' and unique up to a phase, with the property that 
\begin{equation*}
a(\underline{k})\Omega = 0, 
\end{equation*}
for all $\underline{k}$.
 Readers  not familiar with 
these objects may wish to consult, e.g., Section X.7 of \cite{RS2}.

The Hamiltonian of the free electromagnetic field is given by
\begin{equation}
\label{H2}
H_{f}= \int_{ \underline{\mathbb{R}}^3 } \vert \vec{k} \vert \text{ } a^{*}( \underline{k} ) a ( \underline{k} ) 
d \underline{k},
\end{equation}
which is a densely defined, positive operator on $ \mathcal{H}_f$. 

\subsubsection{The Physical System} 
 Our goal is to study an atom interacting with the quantized electromagnetic field. The Hilbert space of states of this system (atom$\vee$photons) is the tensor product space
\begin{equation*}
\mathcal{H}=\mathcal{H}_{at} \otimes \mathcal{H}_f.
\end{equation*}
We choose the interaction of the atom with the quantized electromagnetic field to be given by the Ritz Hamiltonian  
\begin{equation}\label{becio}
\lambda_0 \mathbf{H}_{I}:=- \vec{d}_{0} \cdot \vec{E}(\vec{x}),  
\end{equation}
where 
\begin{equation} \label{d00}
\vec{d}_0 = -\lambda_{0}\vec{d}
\end{equation}
 is the atomic dipole moment, $\lambda_0 > 0$ is a coupling constant proportional to the elementary electric charge, and $\|d^{i} \|=1$, $i=1,...,3$.  Furthermore, $\vec{x}$ is the position of the (center of mass of the) atom and $\vec E$ denotes the quantized electric field, cut off at large photon frequencies. It is given by the operator
\begin{equation}\label{electricfield} 
\vec E (\vec{x}):=  i  \int_{\underline{\mathbb{R}}^3} \Lambda(\vec k) \vert \vec{k} \vert^{\frac12} \vec{\epsilon} (\underline{k} ) \left(  e^{i \vec{ k} \cdot \vec{x}}  \otimes  \mathds{1}_{ \mathbb{C}^N }  \otimes  a ( \underline{k} ) - 
  e^{-i \vec{ k} \cdot \vec{x}}  \otimes  \mathds{1}_{ \mathbb{C}^N } \otimes  a^* ( \underline{k} ) \right) d \underline{k},
\end{equation}
acting on $\mathcal{H}$.
In \eqref{electricfield}, 
$\underline k \mapsto \vec \epsilon(\underline k) \in \mathbb{R}^3$
represents the polarization vector. It is a measurable function with the properties
\begin{equation}\label{polariz}
 | \vec \epsilon(\underline k) | = 1, \hspace{1.5cm} \vec \epsilon(\underline k)\cdot \vec k = 0, \hspace{2cm} \vec \epsilon(r\vec k, \lambda) =\vec \epsilon( \vec k, \lambda ),
\: \forall r > 0, \: \forall \underline{k} \in \underline{\mathbb{R}}^3. 
\end{equation}
The function 
$\Lambda : \mathbb{R}^3 \mapsto \mathbb{R} $
is an ultraviolet cut-off. To be concrete, we take it to be the Gaussian
\begin{equation}\label{Lambda}
\Lambda(\vec k) = e^{ - |\vec k|^2 /(2 \sigma_\Lambda^2)}
\end{equation}
for some cut-off constant $ \sigma_\Lambda \geq 1$. (Obviously, one may consider a more general class of cut-off functions.) \\
The total Hamiltonian of the system is the sum of the Hamiltonians of the atom and the electromagnetic field, plus an interaction term. It is given by
\begin{equation} \label{H}
{\bf H } : = H_{at}  +  H_f
+ \lambda_0 \mathbf{H}_I.
\end{equation}    
Using the Kato-Rellich theorem, one shows that the Hamiltonian ${\bf H}$ is defined and self-adjoint on the dense domain $D(H_{at} \otimes \mathds{1}_{\mathcal{H}_f} +  \mathds{1}_{\mathcal{H}_{at}} \otimes H_f )$, 
where $ D(A) $ represents the domain of the linear operator $A$.
 \\      
\subsubsection{The Fibre Hamiltonian} 
 The photon momentum operator is the vector operator defined by 
\begin{equation} \label{pf}
\vec P_f : =   \int_{ \underline{\mathbb{R}}^3 }  \vec{k}  \text{ } a^{*}( \underline{k} ) a ( \underline{k} ) d \underline{k} . 
\end{equation}
Let $\mathcal F$ denote Fourier transformation in the electron position variable $\vec x \in \mathbb{R}^3$. 
We define the unitary operator 
\begin{equation} \label{u}
U : =  \mathcal{F}  e^{i \vec x \cdot  \vec P_f }
\end{equation}
 on $\mathcal{H}$.
We conjugate the Hamiltonian ${\bf H}$ in \eqref{H} by the unitary operator $U$ introduced in \eqref{u} and subtract the trivial term $\frac{\vec p^2}{2}$, to obtain the operator 
\begin{align} \label{bfHtilde}
 H : = U {\bf H}U^* - \frac{\vec p^2}{2} =  \frac{1}{2}
( \vec p - 
 \vec P_f )^2 - \frac{\vec p^2}{2}
 +   H_{is}
  +  H_f  +
  \lambda_0   H_{I}, 
\end{align}
where
\begin{equation}
\label{H31}
H_{I}  :   = i  \int_{\underline{\mathbb{R}}^3} \Lambda(\vec k) \vert \vec{k} \vert^{\frac12} \left( \vec{\epsilon} (\underline{k} ) \cdot   \vec{d} \otimes  a ( \underline{k} ) - 
  \vec{\epsilon} (\underline{k} )  \cdot   \vec{d} \otimes  a^* ( \underline{k} )
 \right) d \underline{k}
\end{equation}
and $\vec p= - i \vec{\nabla} + \vec{P}_f$ denotes the total momentum operator. The operator $H$ introduced in Eq. \eqref{bfHtilde} is the main object of study of this paper. 

We remark that 
\begin{equation}\label{iso}
L^2(\mathbb{R}^3)\otimes \mathbb{C}^N \otimes \mathcal{H}_f \cong L^2(\mathbb{R}^3;
 \mathbb{C}^N \otimes \mathcal{H}_f ).
\end{equation}
Using \eqref{iso} we see that, for an arbitrary 
 $\phi \in L^2(\mathbb{R}_{\vec{p}}^3;\mathbb{C}^N \otimes \mathcal{H}_f ) $, 
\begin{equation} \label{isoH}
( H \phi)(\vec p) = H(\vec p) \phi(\vec p), 
\end{equation}
where the fibre Hamiltonian, $H(\vec p)$, is the operator acting on the fibre space
$$ \mathcal{H}_{\vec p} : =  \mathbb{C}^N \otimes \mathcal{H}_f $$ 
given by
\begin{align}  \label{Hpsn}
 H(\vec p): = 
\frac{1}{2} \vec P_f^2 - \vec p\cdot \vec P_f 
 +    H_f  +   H_{is} + \lambda_0 H_{I}.
\end{align}
Using the fact that $H_{I}$ is relatively bounded with respect to $H_f^{1/2}$ and applying the Kato-Rellich theorem, one sees that, for all $\vec{p} \in \mathbb{R}^3$, $H(\vec{p})$ is a self-adjoint operator on its domain
\begin{equation}
D(  H( \vec{p} ) ) = D( H_f ) \cap D( \vec{P}_f^2 ).
\end{equation}
Eqs \eqref{isoH}-\eqref{Hpsn} can be reformulated in the formalism of direct integrals:     
\begin{equation}
\label{dec}
\mathcal{H} = \int_{\mathbb{R}^{3}}^{\oplus} \mathcal{H}_{\vec{p}}  d \vec{p}, \hspace{2cm}
 H = \int_{\mathbb{R}^{3}}^{\oplus}  H (\vec{p})  d \vec{p}.  
\end{equation}
This paper is devoted to studying properties of the fibre Hamiltonians $H(\vec p), \vec{p} \in \mathbb{R}^3$.

\subsubsection{Complex Dilatations} 
For $\theta \in \mathbb{R}$, we define the (unitary) dilatation operator 
$\gamma(\theta)$ by setting

\begin{equation} \label{Dil}
\gamma(\theta)(\phi)(\vec k, \lambda) := e^{ -3 \theta /2}\phi( e^{-\theta} \vec k, \lambda), 
\text{    } \text{for} \text{    } \phi \in L^2(\underline{\mathbb{R}}^3).
\end{equation}
By $\Gamma(\theta):= \Gamma( \gamma( \theta ) )$ we denote the operator on Fock space 
$\cH_f$ obtained by ``second quantization'' of $\gamma( \theta )$: For an operator $\omega$ acting on the one-photon Hilbert space $L^2( \underline{\mathbb{R}}^3 )$, $\Gamma(\omega)$ denotes the operator defined on $\mathcal{H}_f$ whose restriction to the $n$-photon subspace is given by
\begin{equation}
\Gamma(\omega) |_{ L^2( \underline{ \mathbb{R} }^3 ) ^{ \otimes_s^n } } := \otimes^n \omega .
\end{equation}
A straightforward computation shows that
\begin{align}
 \label{dilH}
 H_{\theta}(\vec p) : =  \Gamma(\theta) H(\vec p)\Gamma(\theta)^*   =  
\frac{1}{2}e^{-2 \theta} \vec P_f^2 - e^{-\theta} \vec p\cdot \vec P_f 
 +   
H_{is}
  + e^{-\theta}  H_f  + \lambda_0 H_{I, \theta},
\end{align}
where 
\begin{equation} \label{inttheta}
H_{I, \theta} : = i  e^{- 2\theta}\int_{\underline{\mathbb{R}}^3} \Lambda( e^{-  \theta} \vec k ) \vert \vec{k} \vert^{\frac12} \left( \vec{\epsilon} (\underline{k} ) \cdot   \vec{d} \otimes  a ( \underline{k} ) - 
\vec{\epsilon} (\underline{k} ) \cdot   \vec{d} \otimes  a^* ( \underline{k} ) \right) d \underline{k}.
\end{equation}
The operator $ H_\theta (\vec p)$ can be analytically extended to the complex domain
\begin{equation}
D(0, \pi/4) : = \{ \theta \in \mathbb{C} \: : \: |\theta| < \pi/4 \} .
\end{equation}
We will verify in Appendix \ref{typA}, below, that, for all $\vec{p} \in \mathbb{R}^3$, the map $\theta \mapsto H_\theta( \vec{p} )$ is an analytic family of type (A) on $D( 0 , \pi / 4 )$, in the sense that 
$H_\theta( \vec{p} )$ is closed on $ D( H_f ) \cap D( \vec{P}_f^2 )$, for all $\theta \in D( 0 , \pi / 4 )$, and the vector function
$\theta \mapsto  H_\theta( \vec{p} ) u$ is analytic in $\theta$ on $D( 0 , \pi / 4 )$, for all 
$u \in  D( H_f ) \cap D( \vec{P}_f^2 )$.   
The study of resonances of the operator $ H(\vec p)$ amounts to studying non-real eigenvalues of  $H_{\theta}(\vec p)$, for  $\theta$ belonging to a suitable open subset of $D(0, \pi/4) \setminus \mathbb{R} $.  

\subsubsection{Analyticity in the Total Momentum} \label{complex} 
We pick a vector $\vec{p}^*$ in $\mathbb{R}^3$  of length smaller than $1$ and a complex number 
$\theta =i \vartheta$ with $0 <\vartheta < \pi/4$. 
We set 
\begin{equation}\label{mu} \mu = \frac{1-\vert \vec{p}^* \vert }{2} \end{equation} 
and define an open set $U_{\theta}[\vec{p}^*]$ in complexified momentum space $\mathbb{C}^{3}$ by
\begin{equation}\label{UP}
U_{\theta}[\vec{p}^*]:= \lbrace \vec{p} \in \mathbb{C}^3  \mid  \vert \vec{p} - \vec{p}^* \vert < \mu  \rbrace \cap \lbrace \vec{p} \in \mathbb{C}^3  \mid  \vert \Im \vec{p}  \vert <\frac{\mu}{2}  \tan(\vartheta) \rbrace .
\end{equation}

Our main interest is to analyze the $\vec{p}$-dependence of the ground-state, the ground-state energy and the resonance energies of the Hamiltonian $H_{\theta}(\vec{p})$ defined in \eqref{dilH}, for $\vec{p} \in U_{\theta}[\vec{p}^*]$, and to verify that these quantities are analytic in $\vec{p} \in U_{\theta}[\vec{p}^*]$.
By $H_{\theta,0}(\vec{p})$ we denote the operator given by
\begin{equation}
H_{\theta,0}(\vec{p}) : =  e^{- 2 \theta} \frac{\vec{P}_f^2}{2 } - e^{-\theta} \vec{p}  \cdot \vec{P_f}+  
 H_{is}  + e^{-\theta} H_{f}
\end{equation}
corresponding to a vanishing coupling constant, $\lambda_{0}=0$.
It is easy to verify that, for $\delta_{0} > 0$, $E_1, \dots , E_N$ are simple eigenvalues of  $H_{\theta,0}(\vec{p})$. Moreover, it is easy to see that, for $|\vec{p}| < 1$ and $\vec{p} \in \mathbb{R}^3$, the spectrum of $H_{\theta,0}( \vec{p} )$ is included in a region of the form depicted in Figure 1.
\begin{figure}[H]
\begin{center}
\begin{tikzpicture}
 
       \draw[->] (0,0)--(15,0.0);

      \fill[ color=black] (2,0) circle (0.05) ;
     
          \fill[ color=black] (4,0) circle (0.05) ;
               \fill[ color=black] (7,0) circle (0.05) ;
                    \fill[ color=black] (11,0) circle (0.05) ;

         \draw[-,dashed] (2,0)--(4,-2.0);
           \draw[-,dashed] (4,0)--(6,-2.0);
              \draw[-,dashed] (7,0)--(9,-2.0);
                \draw[-,dashed] (11,0)--(13,-2.0);

                   \draw[-,dashed] (2,0)--(2.3,-2.0);
           \draw[-,dashed] (4,0)--(4.3,-2.0);
              \draw[-,dashed] (7,0)--(7.3,-2.0);
                \draw[-,dashed] (11,0)--(11.3,-2.0);
         
          \fill[color=black!25]     (4,0)to[bend left=0] (6,-2)  to[bend left=0] (4.6,-2.)   to[bend left=-10] (4,0);              
         \fill[color=black!25]     (7,0)to[bend left=0] (9,-2)  to[bend left=0] (7.6,-2.)   to[bend left=-10] (7,0);    
            \fill[color=black!25]     (11,0)to[bend left=0] (13,-2)  to[bend left=0] (11.6,-2.)   to[bend left=-10] (11,0);    
               \fill[color=black!25]     (2,0)to[bend left=0] (4,-2)  to[bend left=0] (2.6,-2.)   to[bend left=-10] (2,0);    
   
    \draw (2,-.06) node[above] {\small$ E_{1}$};
     \draw (4,-.06) node[above] {\small$ E_{2}$};
         \draw (7,-.06) node[above] {\small$ E_{3}$};
           \draw (9,-.06) node[above] { ... };
             \draw (11,-.06) node[above] {\small$ E_{N}$};

       \draw[->,color=black]  (2.4,0)to[bend left=20] (2.2,-0.2);
       \draw[->,color=black]  (2.8,0)to[bend left=20] (2.1,-0.5);
    \draw (2.4,0.1) node[below] { \footnotesize $\vartheta$};     
     \draw (2.7,-0.2) node[below] { \footnotesize $ 2 \vartheta$};

\end{tikzpicture}
\end{center}
\caption{ \small Shape of the spectrum of the unperturbed operator $H_{\theta,0}(\vec{p})$ for $\vec{p}\in \mathbb{R}^3$, $| \vec{p} | < 1$. $E_1$,...,$E_N$ are eigenvalues of  $H_{\theta,0}(\vec{p})$, the essential spectrum is located inside the cuspidal grey regions.}

\end{figure}
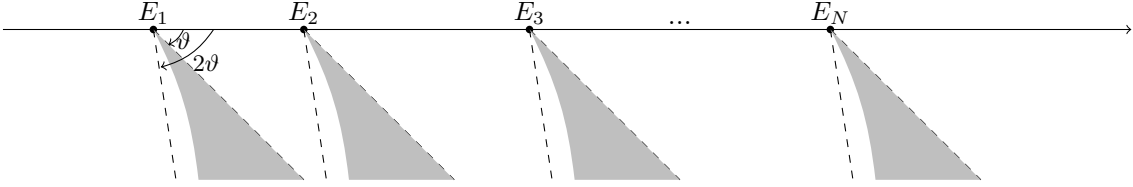

\subsection{Main Results} 
Theorem \ref{t1}, below, claims that, for $|\vec{p}|<1$, a ground-state and resonances exist and that the ground-state,  the ground-state energy and the resonance energies are analytic in $\vec{p} \in U_{\theta}[\vec{p}^{*}]$ (and in
 $\lambda_{0}$, for $|\lambda_{0}|$ small enough). If $|\vec{p}|>1$ one expects that the operator 
 $H(\vec{p})$ does \textit{not} have a stable ground-state, due to emission of Cherenkov radiation; see \cite{Ceren}.
Assuming that the so-called Fermi-Golden-Rule condition holds, the imaginary parts of the resonance energies are strictly negative, i.e., the life times of the excited states of an atom are strictly \textit{finite} due to radiative decay; see Proposition \ref{t2}.

\begin{theorem} \label{t1}
Let $0 < \nu < 1$. There exists $\lambda_c ( \nu )>0$ such that, for all $0 \le \lambda_0 < \lambda_c( \nu )$ and $\vec{p} \in \mathbb{R}^3$, $| \vec{p} | < \nu$, the following properties are satisfied:
\begin{itemize}
\item[a)] $E( \vec{p} ) := \inf \sigma( H( \vec{p} ) )$ is a non-degenerate eigenvalue of $H( \vec{p} )$. \vspace{0,1cm}
\item[b)] For every $i_0 \in \{1, \cdots, N\} $ and $\theta \in \mathbb{C}$ with $0 < \Im \theta < \pi / 4$ large enough, $H_\theta(\vec p)$ has an eigenvalue, $z_{i_0}^{(\infty)}(\vec{p})$, such that 
$z_{i_0}^{(\infty)}(\vec{p}) \to E_{i_0}$, as $\lambda_0 \to 0$. For $i_0=1$, $z_{1}^{(\infty)}( \vec{p} ) = E( \vec{p} )$.
\end{itemize}

\noindent Moreover, for $| \vec{p} | < \nu$, $\vert\lambda_0 \vert$ small enough and $0 < \Im \theta < \pi / 4$ large enough, the ground state energy, $E(\vec{p})$, the resonance energies $z_{i_0}^{(\infty)}(\vec{p})$, $i_0\geq 2$, and their respective eigenvectors (unique up to a phase), are analytic in $\vec{p}$, $\lambda_0$ and $\theta$. In particular, $E(\vec{p})$ and  $z_{i_0}^{(\infty)}(\vec{p})$, $i_0\geq 2$,  are independent of $\theta$.
\end{theorem}

\begin{remark} $\quad$
\begin{itemize}
\item Existence and analyticity  of a ground state, as well as  analyticity of the map $\vec{p} \mapsto E( \vec{p} )$, are proven in \cite{FFS}.
\item For simplicity of exposition, we only prove, in the present paper, the (existence and) analyticity of the resonance energies $z_{i_0}^{(\infty)}( \vec{p} )$ in $\vec{p}$, for $i_0\geq2$. In the following, we fix $i_0$ and write $z_{i_0}^{(\infty)}( \vec{p})=: z^{(\infty)}(\vec{p})$; (dependence on $i_0$ suppressed). Our proof can be adapted in a straightforward way to establish the statements concerning analyticity in $\lambda_0$ and $\theta$. For different models similar to the model of non-relativistic QED studied in this paper, analyticity in the coupling constant has been proven previously in \cite{GrHa09_01,HaHe11_01,HaHe10_01}.
\item For Pauli-Fierz models with static nuclei, resonances have been studied in \cite{BaFrSi98_02,Si09_01,AFFS,BaBaPi}.
\item The fact that $z^{(\infty)}( \vec{p} )$ is independent of $\theta$ is a direct consequence of the analyticity of $z^{(\infty)}( \vec{p} )$ in $\theta$, together with unitarity of the dilatation operator $\Gamma( \theta )$ for real $\theta$'s and with  the existence of a normalizable and analytic eigenstate of $H_{\theta}(\vec{p})$ associated to $z^{(\infty)}( \vec{p} )$.
\item Theorem 1.1 extends to more realistic models of atoms (with dynamical electrons) as considered for instance in \cite{AmGrGu06_01,FrGrSc07_01,LoMiSp07_01}. Such models are not treated in our paper in order not to hide the basic simplicity of our methods.
\end{itemize}
\end{remark}

\begin{proposition}\label{t2} 
Let $i_0 > 1$ and $\vec p \in \mathbb{R}^3$, $| \vec{p} | < 1$.
Suppose that
\begin{align*}
\sum_{j < i_0} \int_{\underline{\mathbb{R}}^3} d \underline{k}  \text{ } \big |  (\vec{d})_{N-j + 1,N- i_0 +1} \cdot \vec{\epsilon}(\underline{k})  \big |^2 |\vec k| | \Lambda (\vec k)|^2 \delta\big(E_j - E_{i_0} + |\vec k| 
-  \vec p \cdot \vec k + \frac{\vec k^2}{2} \big) > 0,
\end{align*}
\begin{align}\label{et2}
\text{(Fermi-Golden-Rule condition)}
\end{align}
where $ (d^{l})_{i,j} $, $l=1,...,3$, is the matrix element of the operator $d^{l}$ in the eigenbasis of $H_{is}$; see Eq. \eqref{dip} and \eqref{d00}.
Then, under the conditions of Theorem \ref{t1} and for $\lambda_0$ small enough, the imaginary part of $z^{(\infty)}(\vec p)$ is strictly negative. 
\end{proposition} 

The proof of Proposition \ref{t2}  does not rely on the inductive construction used to establish Theorem \ref{t1}. A single application of a suitably chosen Feshbach-Schur map, i.e., a single decimation step, is sufficient to prove this proposition, and our argumentation follows closely the one presented in \cite{BFS-1999,BaFrSi98_01}. To render this paper reasonably self-contained, the proof is given in Section \ref{imazinfty}. 

\subsection{Strategy of Proof and Sketch of Methods}  
Ultimately, our aim is to study spectral properties of the operators $H_\theta(\vec p)$ introduced in \eqref{dilH}. This spectral problem is difficult, because, among other things, it involves the study of eigenvalues imbedded in continuous spectrum and located at thresholds of the continuous spectrum of 
$H_\theta(\vec p)$.  Standard analytic perturbation theory is therefore not applicable. The key tool we will use to prove our results is the isospectral Feshbach-Schur map, which was originally developed to cope with problems of this kind in \cite{1995}. In this paper we will use the smooth Feshbach-Schur map introduced in \cite{BaChFrSi03_01} and further studied in \cite{GrHa08_01} and \cite{GrHa09_01}, which has major technical advantages (and, alas, some conceptual disadvantages), as compared to the original Feshbach-Schur map.

The Feshbach-Schur map is tailor-made for the analysis of small regions in the spectra of closed operators on Hilbert space, in particular regions of their spectra near thresholds. It enables one to construct ``effective operators'' that, on the part of the spectrum of interest, have the \textit{same} spectrum (with the same multiplicity) as the original operator, i.e., are \textit{iso-spectral} to the original operator. By iterating the Feshbach-Schur map one is able to zoom into tiny regions in the spectrum of an operator of interest and extract ever more accurate information on such parts of the spectrum. In particular, by constructing an infinite sequence of Feshbach-Schur maps, we will be able to determine the exact location of the ground-state- and the resonance energies and the corresponding eigenstates of the deformed Hamiltonians, 
$H_{\theta}(\vec{p})$, $\Im\theta > 0$, of atoms coupled to the radiation field. The Feshbach-Schur maps will be adapted to the particular resonance that one wishes to analyze. It is a novel aspect of our construction that it yields an algorithm that converges super-exponentially fast.

\subsubsection{Mathematical Tools}  \label{1.3.1}
\paragraph{\textbf{The Feshbach-Schur Map}} 
The fundamental tool used to prove our main results is the smooth Feshbach-Schur map; see \cite{BaChFrSi03_01,GrHa08_01}. A key property of this map is its iso-spectrality, which we now describe in more precise terms.

\begin{definition}[Feshbach-Schur Pairs] \label{feshbach} 
Let $ P $ be a positive operator on a separable Hilbert space $ \mathcal{H} $ whose norm is bounded by $1$, $0 \leq P \leq 1$. 
Assume that $P$ and $\overline P : = \sqrt{1 - P^2}$ are both non-zero. Let $H$ and $T$ be two closed operators on $\mathcal{H}$
with identical domains $D(H)$ and $D(T)$. Assume that $P$ and $\overline P$ commute with $T$. We set $W : = H - T$ and we define 
\begin{align*}
H_{P} : = & T + PW P, \hspace{2cm}  H_{\overline P} : = T + \overline{P} W \overline{P}.
\end{align*} 
 The pair $(H, T)$ is called a Feshbach-Schur pair associated with $P$ iff
\begin{itemize}
\item[(i)] $ H_{\overline P} $  and $ T $ are bounded invertible on $\overline P[\mathcal{H}] $
\item[(ii)]
$ H_{\overline P}^{-1} \overline P W P  $
can be extended to a bounded operator on $\mathcal{H}$ 
\end{itemize}
For an arbitrary Feshbach-Schur pair $(H,T)$ associated with $P$, we define the smooth Feshbach-Schur map by 
\begin{equation}\label{FB}
F_{P}(\cdot,T): H \mapsto F_P(H, T) : =  T + PWP - P W\overline P H_{\overline P}^{-1}\overline P W P.
\end{equation}
 \end{definition}
 
\begin{theorem} \label{feshbach-theorem}
Let $0 \leq P \leq 1$, and let $(H, T)$ be a Feshbach-Schur pair  associated with $ P $ (i.e., satisfying properties (i) and (ii) in Definition \ref{feshbach}).  Let $V$ be a closed subspace  with $P[\mathcal{H}] \subset V \subset \mathcal{H}$, and  such that
$$T: D(T) \cap V \rightarrow V, \qquad \overline{P}T^{-1} \overline{P} V \subset V.$$
 Define 
\begin{align*}
Q_P(H, T) : = P - \overline P H_{\overline P}^{-1}\overline P W P. 
\end{align*}
Then the following hold true:  
\begin{itemize}
\item[(i)] $H$ is bounded invertible on $\mathcal{H}$ if and only if $ F_P(H, T) $ is bounded invertible on $V$. 
\item[(ii)] $H$ is not injective if and only if $ F_P(H, T) $ is not injective as an operator on $V$: 
$$
H \psi = 0, \:\psi \ne 0 \Longrightarrow  F_P(H, T) P \psi = 0,\: P \psi \ne 0,
$$
$$
F_P(H, T)\phi = 0,\:  \phi \ne 0 \Longrightarrow  H Q_P(H, T)\phi = 0,\: Q_P(H, T)\phi \ne 0.
$$
\end{itemize}
\end{theorem}
\begin{remark} \label{isonoiso} 
\item
\begin{itemize}
\item  Items (i) and (ii) of Theorem   \ref{feshbach-theorem} describe what we call 
\underline{iso-spectrality}. This notion does not mean that the spectra of $H$ and of $ F_P(H, T) $ are identical. Rather, iso-spectrality is a {\text{local}} property: One uses the Feshbach-Schur map to explore 
spectral properties of an operator within specific, small regions in the complex plane.
\item  As  emphasized in \cite{GrHa08_01}, if $T$ is bounded invertible  in $ \overline{P}[\mathcal{H}]$, if $ T^{-1} \overline{P} W \overline{P}$ and $\overline{P} W T^{-1} \overline{P}$ are bounded operators with norm  strictly less  than one, and if $T^{-1} \overline{P} W P$ is bounded, then items (i) and (ii) of Definition \ref{feshbach} are satisfied. We  will often use these criteria  to show that a pair $(H,T)$ is a Feshbach-Schur pair associated with $P$.
\end{itemize}
\end{remark}

\paragraph{\textbf{Wick Monomials}} 
We now describe the general class of operators to which the methods developed in this paper, based on the smooth Feshbach-Schur map, can be applied.

Setting
$\mathbb{N}_0 : = \mathbb{N}\cup \{ 0\}$,
we denote by  
\begin{equation} \label{uw}
\underline w : = \{ w_{m,n} \}_{m, n \in \mathbb{N}_0}
\end{equation}
a sequence of bounded measurable functions, 
\begin{equation}
\forall m, n \: : \:  \: w_{m,n} : \mathbb{R} \times \mathbb{R}^3  \times \underline{\mathbb{R}}^{3m} \times 
 \underline{\mathbb{R}}^{3 n} \to \mathbb{C}, 
\end{equation}
that are continuously differentiable in the variables, $r \in \sigma( H_f ) \subset \mathbb{R}$, $\vec{l} \in 
\sigma( \vec{P}_f ) = \mathbb{R}^{3}$, respectively, appearing in the first and the second argument, and symmetric in the $m$ variables in
$\underline{\mathbb{R}}^{3m}$ and the $n$ variables in $\underline{\mathbb{R}}^{3n}$. We suppose furthermore that 
\begin{equation} \label{w00eq0}
w_{0,0}(0,\vec{0})=0.
\end{equation}
With a sequence, $\underline w$, of functions, as specified above,  and a positive number $1 \geq\rho>0$, we associate  a sequence of operators
\begin{equation}
\begin{split}
\label{Wmn}
W_{m,n}(\underline w) : = & \mathds{1}_{H_f \leq \rho} \int_{\underline{ \mathbb{R}}^{3 m} \times \underline{  \mathbb R  }^{3n}}
a^*(\underline{k}_1)\cdots a^*(\underline{k}_m)  w_{m,n}( H_f, \vec P_f, \underline{k}_1, \cdots, \underline{k}_m, \underline{\tilde{k}}_1, \cdots, \underline{\tilde{k}}_n)
\\  & a( \underline{\tilde{k}}_1)\cdots a(  \underline{\tilde{k}}_n) \text{ } \prod_{i=1}^{m} d \underline{k}_i   \text{ } \prod_{j=1}^{n} d \underline{\tilde{k}}_j \mathds{1}_{H_f \leq \rho}.
\end{split}
\end{equation}
It is easy to show that $W_{m,n}(\underline w)$ is actually a bounded operator on $\mathcal{H}_f$. The operators   $W_{m,n}(\underline w)$  defined in \eqref{Wmn} are called (generalized) Wick-monomials (at the energy scale $\rho$).  
For every sequence of functions $\underline{w}$ and every $\mathcal{E}\in \mathbb{C}$ we define 
\begin{align}\label{Hw}
H[\underline w,\mathcal{E}] = \sum_{m + n \geq 0} W_{m,n}(\underline w) + \mathcal{E}, \hspace{1cm} W_{\geq 1}(\underline w) : =  \sum_{m + n \geq 1} W_{m,n}(\underline w).
\end{align}
The complex number $\mathcal E$ is the vacuum expectation value of $ H[\underline w, \mathcal E] $: 
\begin{equation} \label{1.39}
\langle  \Omega  |  H[\underline w, \mathcal E] \Omega   \rangle = \mathcal E . 
\end{equation}
\subsubsection{The First Decimation Step of Spectral Renormalization} \label{fds} 
Recall that we wish to analyze the fate of an excited state of an atom after it is coupled to the radiation field. Let us consider the excited state indexed by $i_0 \in \{ 2, \cdots, N \}$, with unperturbed internal energy $E_{i_0}$. We expect that, after coupling the atom to the quantized radiation field, an excited state (corresponding to an index $i_0 > 1$) is \textit{unstable}, i.e., is turned into a resonance. Our goal is to determine its life time and the real part of the resonance energy (Lamb shift). For this purpose, we introduce a sequence of smooth Feshbach-Schur ``decimation'' maps that will be successively applied to the deformed Hamiltonians 
$H_{\theta}(\vec{p})$, with the goal of constructing a sequence of operators, which -- when applied to the vacuum $\Omega$ -- will converge to 
$z^{(\infty)}(\vec{p}) \Omega$, where $z^{(\infty)}(\vec{p})$ is the $i_{0}^{th}$ resonance energy; (as announced, we will omit reference to $i_{0}$ in our notation, since $i_{0}$ will be fixed). In this subsection, we sketch the construction of the first decimation map. 

We define a decreasing function $\chi \in C^{\infty}(\mathbb{R}) $ satisfying 
\begin{equation}\label{chi}
\chi(r) : = \begin{cases} 1, &  \text{if $ r \leq  3/4 $,} \\ 0 & \text{if $ r > 1 $,} \end{cases}
\end{equation}
and strictly decreasing in $(3/4, 1  )$. Furthermore, we choose
a constant $ \rho_0 \in (0, 1)$ and define
\begin{equation}\label{chirho}
\chi_{\rho_0 }(r) := \chi(r/\rho_0), \hspace{2cm} \overline \chi_{\rho_0}(r) : = \sqrt{1 - \chi_{\rho_0}^2(r)}.
\end{equation}
\noindent Let $\psi_{i_0}$ denote the normalized eigenvector (unique up to a phase) of the operator $H_{is}$ corresponding to the eigenvalue $E_{i_0}$. The orthogonal projection onto $\psi_{i_0}$ is denoted by
\begin{equation}\label{psi0}
P_{i_0} := \vert \psi_{i_0} \rangle  \langle \psi_{i_0} \vert. 
\end{equation}

\noindent Next, we define an operator ${\bm{\chi}}_{i_0}$ by
\begin{equation} \label{chi0}
\bm{\chi}_{i_0} : =  P_{i_0} \otimes \chi_{\rho_0}(H_f).
\end{equation}
In Section \ref{premiere} we will prove that, for $ |z - E_{i_0}| \ll \rho_0\mu \sin(\vartheta) $,
$(H_\theta(\vec p) - z, H_{\theta,0}(\vec p) - z) $ is a Feshbach-Schur pair associated to 
$\bm{\chi}_{i_0} $ and that, as a consequence, 
there is a sequence of functions $\underline w^{(0)} (\vec p, z)$ [see \eqref{uw}]  and a complex number 
$\mathcal{E}^{(0)}(\vec p, z)$ such that an application of the Feshbach-Schur map, 
$F_{\bm{\chi}_{i_0}}(\cdot, H_{\theta,0}(\vec{p})-z)$, to the operator $H_{\theta}(\vec{p})-z$ yields an operator of the form specified in Eq. (\ref{Hw}). More precisely,
\begin{equation} \label{Wo}
 F_{\bm{\chi}_{i_0}} (H_\theta(\vec p) -z, H_{\theta,0}(\vec p) - z)_{\mid   P_{i_0} \otimes  \mathds{1}_{H_f \leq \rho_0} [\mathcal{H}_{\vec{p}}]} =  P_{i_0}\otimes  H[\underline w^{(0)}(\vec p, z),\mathcal{E}^{(0)}(\vec{p},z)].
\end{equation}
\noindent We simplify our notation by writing [see \eqref{Hw}]
\begin{align}\label{sim}
 H^{(0)}(\vec p, z) : =  H[\underline w^{(0)}(\vec p, z),\mathcal{E}^{(0)}(\vec{p},z)]  = 
 W^{(0)}_{\geq 1}(\vec{p},z) + w_{0,0}^{(0)}(\vec{p},z,H_f,\vec{P}_f)  + \mathcal{E}^{0}(\vec{p},z),
 \end{align}
 where
\begin{equation*}
W^{(0)}_{\geq 1}(\vec p, z) : = \sum_{m+n \geq 1} W_{m,n}(\underline w^{(0)}(\vec p, z)).  
\end{equation*}
One expects that it is easier to analyze the operator $  H^{(0)}(\vec p, z)  $, rather than the original operator 
$ H_\theta(\vec p) - z $, because the former acts on a subspace, $  P_{i_0} \otimes   \mathds{1}_{H_f \leq \rho_0} [\mathcal{H}_{\vec{p}}] \subset \mathcal{H}_{\vec{p}}$ (with all internal states corresponding to indices $i \not= i_0$ eliminated), and the operator
$  W^{(0)}_{\geq 1 }(\vec p, z) $ is bounded in norm by some power of $\rho_0$.



\subsubsection{Inductive Construction of Effective Hamiltonians}\label{secic} 
The accuracy of the information on the spectrum of the operator $H^{(0)}(\vec{p},z)$ near $0$, and hence on the spectrum of the operator $H_{\theta}(\vec{p})$ near $E_{i_0}$, that can be achieved (after one application of the Feshbach-Schur map) is limited by the circumstance that $\rho_0$  cannot be taken to be very small. Luckily, it turns out that this limitation can be removed by successive applications of Feshbach-Schur maps that lower the energy range of the states in the subspaces on which the Feshbach-Schur operators act further and further towards $0$ and, hence, determine the location of the spectrum of $H_{\theta}(\vec{p})$ near $E_{i_0}$ ever more accurately. Successive applications of Feshbach-Schur maps yield Hamiltonians  
\begin{equation}
 H^{(j)}(\vec p, z) = H[\underline{w}^{(j)}(\vec{p},z), \mathcal{E}^{(j)}(\vec{p},z)], \qquad j \in \mathbb{N}_0,
 \end{equation}
as in Eq. \eqref{Hw}, with the following properties: 
 \begin{equation}\label{iso1}
H_\theta(\vec p) - z \: \text{is bounded invertible}  \Longleftrightarrow  H^{(j)}(\vec p, z)\:  \text{is bounded invertible}.
\end{equation}
\begin{equation}\label{iso2}
H_\theta(\vec p) - z \: \text{is not injective}  \Longleftrightarrow  H^{(j)}(\vec p, z)\:  \text{is not injective}.
\end{equation}
These ``iso-spectrality properties'' permit us to trade the analysis of the spectrum of $ H_\theta(\vec p) $
near the energy $E_{i_0}$ of an excited state of the atom for the analysis of the spectrum of the operators 
$ H^{(j)}(\vec p, z) $ near the origin. This turns out to simplify matters considerably: The study of the operators 
$ H^{(j)}(\vec p, z) $ is much easier than the study of the original Hamiltonian, because $ H^{(j)}(\vec p, z) $ is the sum of a diagonal operator, whose spectrum is known explicitly, and a perturbation term whose operator norm will turn out to decrease to zero  super-exponentially, as $j\rightarrow \infty$.
Below, we describe in somewhat more detail how this idea, which was originally developed in \cite{1995}, \cite{BaFrSi98_01}, can be implemented, technically; (details will be presented in Section \ref{itere}). 

Let $ \vec p^* \in \mathbb{R}^3 $, with $ |p^*| < 1 $, and let $\vec p \in U_{\theta}[\vec p^*]$. 
We define two sequences of numbers $ ( \rho_j)_{j \in \mathbb{N}_0} $,  $ ( r_j)_{j \in \mathbb{N}_0} $ by
\begin{equation} \label{sequence}
\rho_j = \rho_0^{(2 - \varepsilon)^j}, \text{   }\text{with} \text{   }\varepsilon \in (0,1), \hspace{1cm}  r_j : = 
\frac{ \mu \sin(\vartheta)}{32}\rho_j, 
\end{equation}
where $0<\rho_0<1$ is a suitably chosen parameter; (see Section \ref{itere}).  The  rate of convergence of the sequence $\rho_j$ depends on the infrared behavior of  the interacting Hamiltonian $H_{I}$. In general, if  $H_I$   behaves like $\vert \vec{k} \vert^{\alpha-1/2}$ in the infrared and   $\alpha>0$,  the sequence $\rho_j$  can be chosen to be equal to  $\rho_0^{(1 + \epsilon)^j}$ for any $0<\epsilon<\alpha$. A filtration of Hilbert spaces 
$(\mathcal{H}^{(j)})_{j \in \mathbb{N}_0}$ is given by setting
\begin{equation}\label{hilbertj}
\cH^{(j)} = \mathds{1}_{H_f \leq \rho_{j}}[\cH_f].
\end{equation}
We construct inductively a sequence of complex numbers 
$\lbrace z^{(j- 1)}(\vec p)\rbrace_{j \in \mathbb{N}_0}$, $z^{(-1)}(\vec p) : = E_{i_0}$, and, for every $z \in D(z^{(j - 1)}(\vec p), r_{j})$, a sequence of  functions 
$\underline w^{(j)}(\vec p, z) $ and a complex number $\mathcal{E}^{(j)}(\vec p, z)$  [see \eqref{uw}-\eqref{1.39}] , with the following properties: 
\begin{itemize} 
\item[(a)] Let 
\begin{equation}\label{Wmnj}
W^{(j)}_{m,n}(\vec p, z) : = W_{m,n}(\underline w^{(j)}(\vec p, z)), \text{   }  \quad 
H^{(j)}(\vec p, z) : = H[\underline w^{(j)}(\vec p, z), \mathcal{E}^{(j)}(\vec p, z)],
\end{equation}
acting on $\mathcal{H}^{(j)} $, (with $m,n \in \mathbb{N}_0$); see \eqref{Wmn} and \eqref{Hw}.
Then we have that 
\begin{equation}\label{LowerBound}
\| W^{(j)}_{0,0}(\vec{p},z)  \Psi \|= \| w^{(j)}_{0,0}(\vec{p},z,H_f, \vec{P}_f) \Psi\| \geq \epsilon  \| H_f \Psi\|, \qquad \forall \Psi \in \mathcal{H}^{(j)}
\end{equation}
for some constant $\epsilon > 0$ depending on $\vec{p}$, but \textit{independent} of $j$.
The pair of operators $\big(H^{(j)}(\vec p, z), W^{(j)}_{0,0}(\vec{p},z) + \mathcal{E}^{(j)}(\vec{p},z)\big )$
is a Feshbach-Schur pair associated to $\chi_{\rho_{j}}(H_f)$. Thus 
\begin{equation} \label{ind}
H^{(j+1)}(\vec p, z) = F_{\chi_{  \rho_{j+1} }(H_f)}[H^{(j)}(\vec p, z),W^{(j)}_{0,0}(\vec{p},z) + \mathcal{E}^{(j)}(\vec{p},z)]|_{   \mathcal{H}^{(j+1)}}
\end{equation}
is well defined. Note that the vacuum vector $\Omega \in \mathcal{H}_f $ is an eigenvector of $W_{0,0}^{(j)}(\vec{p},z)$ with associated eigenvalue $0$, for all $j \in \mathbb{N}_0$.
\item[(b)]  The complex number $z^{(j)}(\vec p)$ is defined as the only zero of the function 
\begin{equation}
D\Big (z^{(j - 1)}(\vec p), \frac{2}{3}r_{j}  \Big ) \ni z \longrightarrow    \mathcal{E}^{(j)}(\vec p, z) = \langle \Omega |\; H^{(j)}(\vec p, z)   \Omega\rangle, 
\end{equation}
and the following inequalities hold: 
\begin{equation}\label{fore}
|z^{(j )}(\vec p) - z^{(j-1)}(\vec p)|< \frac{r_{j}}{2}, \text{   } \qquad
\big \vert \mathcal{E}^{(j)}(\vec p, z)  \big \vert \leq \frac{\mu}{16} \rho_{j+1}, \text{   } \text{for} \text{   }z \in D\Big (z^{(j)}(\vec p), \frac{2}{3}r_{j+1}  \Big ).
\end{equation}
\end{itemize}
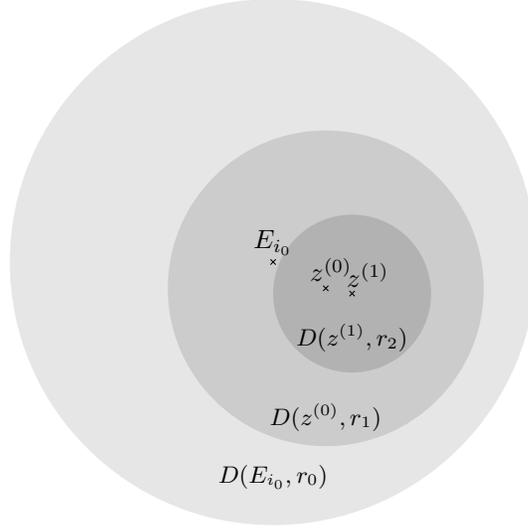
\begin{figure}[H]
\begin{center}
\begin{tikzpicture}[scale=0.7]

      \fill[ color=black!10] (2,0) circle (5) ;
       \fill[ color=black!20] (3,-0.5) circle (3) ;
        \fill[ color=black!30] (3.5,-0.6) circle (1.5) ;

    \draw[-] (1.95,-0.05)--(2.05,0.05);
    \draw[-] (1.95,0.05)--(2.05,0-.05);

    \draw[-] (2.95,-0.55)--(3.05,-0.45);
    \draw[-] (2.95,-0.45)--(3.05,-0.55);

        \draw[-] (3.45,-0.65)--(3.55,-0.55);
    \draw[-] (3.45,-0.55)--(3.55,0-.65);

    \draw (2,-.06) node[above] {\small$ E_{i_0}$};
      \draw (3.1,-.56) node[above] {\small$ z^{(0)}$};
        \draw (3.8,-.66) node[above] {\small$ z^{(1)}$};

      \draw (2,-4.5) node[above] {\footnotesize $ D(E_{i_0},r_0)$};
           \draw (3,-3.4) node[above] {\footnotesize $ D(z^{(0)},r_1)$};
             \draw (3.5,-1.9) node[above] {\footnotesize $ D(z^{(1)},r_2)$};

\end{tikzpicture}

\end{center}
\caption{\small For fixed $\vec{p} \in U_{\theta}[\vec{p}^*]$, the sets $D(z^{(j)}(\vec{p}),r_{j+1})$ are shrinking  super-exponentially fast with $j$ and, for every $j \in \mathbb{N}_0$, $D(z^{(j)}(\vec{p}),r_{j+1}) \subset D(z^{(j-1)}(\vec{p}),r_j)$.} 
\end{figure} 


\noindent By Eqs. \eqref{Hw} and \eqref{Wmnj},
\begin{equation}\label{frej}
H^{(j)}(\vec p, z) =  W^{(j)}_{0,0}(\vec p, z) + \mathcal{E}^{(j)}(\vec{p},z)+ W^{(j)}_{\geq 1}(\vec p, z).
\end{equation}
As a function of $H_f$ and $\vec{P}_f$, the operator 
$ W^{(j)}_{0,0}(\vec p, z) = w^{(j)}_{0,0}(\vec{p},z,H_f, \vec{P}_f)$ 
is defined by functional calculus and satisfies \eqref{LowerBound}. Given $w^{(j)}_{0,0}$, the spectrum of 
\begin{equation}\label{free-part-j}
 W^{(j)}_{0,0}(\vec p, z) + \mathcal{E}^{(j)}(\vec{p},z)
\end{equation}
can be determined explicitly. This operator is therefore considered to be the unperturbed Hamiltonian (the operator $T$ in Definition \ref{feshbach}) in the next application of the Feshbach-Schur map.  Eq. \eqref{frej} shows that the operator $ H^{(j)}(\vec p, z) $ is the sum of the unperturbed Hamiltonian, 
$T = W^{(j)}_{0,0}(\vec p, z) + \mathcal{E}^{(j)}(\vec{p},z) $, and a perturbation given by
$W =  W^{(j)}_{\geq 1}(\vec p, z)$ whose norm tends to zero, as $j$ tends to $\infty$.  We will actually prove that, for every $j \in \mathbb{N}_0$,
\begin{equation}\label{sej}
\|  W^{(j)}_{\geq 1}(\vec p, z)  \| \leq \mathbf{C}^j \rho_j^2,
\end{equation}
for some constant $\mathbf{C}>1$. Recalling formula \eqref{FB} for the Feshbach-Schur map, we find that the bounds \eqref{LowerBound}, \eqref{fore}, and  \eqref{sej},  enable us to construct $F_{\chi_{  \rho_{j+1} }(H_f)}[H^{(j)}(\vec p, z),W^{(j)}_{0,0}(\vec{p},z) + \mathcal{E}^{(j)}(\vec{p},z)]$ with the help of a convergent Neumann expansion in powers of
the perturbation $W^{(j)}_{\geq 1}(\vec{p}, z)$. Thanks to  \eqref{LowerBound}, \eqref{fore} and \eqref{sej}, and using ``iso-spectrality'', the sequence $\lbrace H^{(j)}(\vec{p},z)\rbrace$ of effective Hamiltonians enables us to locate the spectrum of the deformed Hamiltonian $H_{\theta}(\vec{p})$, 
$\Im\theta > 0$, near the energy $E_{i_0}$ with ever higher precision as the resonance energy 
$z^{(\infty)}(\vec{p})$ is approached. The complex number $ z^{(\infty)}(\vec p) $ is defined in the next subsection as the limit of $ z^{(j)}(\vec p)$ as $j$ tends to infinity. There we show that it is an eigenvalue of $H_{\theta}(\vec p)$.

\textit{It is a characteristic feature of multi-scale renormalization, as well as of KAM theory, that a problem of singular perturbation theory involving an infinite range of scales is decomposed into a sequence of infinitely many regular perturbation problems, one for every finite range of scales, solved iteratively, with the splitting of an effective Hamiltonian into an unperturbed part and a perturbation chosen {\bf{anew}}, in every step, $j$, of the iterative perturbative analysis. These are key features enabling one to successfully cope with problems of singular perturbation theory. They will become manifest in the analysis presented in this paper.}

\subsubsection{Construction of Eigenvalues and Analyticity in $\vec{p}$, $\theta$ and $\lambda_{0}$} \label{cea} 
In this section we sketch the main ideas of our construction of the ground-state-($i_0=1$) and resonance-($i_0>1$) 
energy $z^{(\infty)}_{i_0}(\vec p)$ of $H_\theta(\vec p)$ (for some fixed $i_0 \geq 1$) and of the proof that $z^{(\infty)}(\vec{p})=z^{(\infty)}_{i_0}(\vec p)$ is an eigenvalue of $H_\theta(\vec p)$, for 
$ \vec p \in U_{\theta}[ \vec p^*] $, $\Im\theta<\frac{\pi}{4}$ large enough, and $\lambda_0 \ge 0$ small enough. Note that, for the ground-state, i.e., for $i_0 = 1$, we can choose $\theta$ to vanish, and $z^{(\infty)}(\vec{p})$ is shown to be a simple eigenvalue of the self-adjoint operator $H(\vec{p})$, for $\vec{p}\in\mathbb{R}^{3}$, with $\vert\vec{p}\vert < 1$, and $\lambda_0$ positive and small enough. As a function of 
$\vec{p}$ this is the renormalized dispersion law of the atom.

We start our considerations by observing that the sequence of approximate resonance energies 
$ ( z^{(j)}(\vec p))_{ j \in \mathbb{N}_0}$ is Cauchy, as follows from Eq. \eqref{fore}. 
It is not difficult to show (see Section \ref{conv}) that it actually is a Cauchy sequence of analytic functions of the momentum $\vec{p}$, for $\vec p \in U_{\theta} [\vec p^*]$. Analyticity in $\theta$, for $\Im\theta < \frac{\pi}{4}$ large enough, and in $\lambda_0$, for $\vert\lambda_0\vert$ small enough, can be shown by very similar arguments, which we skip here. We then define   
\begin{equation}\label{ero.5}
 z^{(\infty)}(\vec p) : = \lim_{j \to \infty} z^{(j)}(\vec p)  =   \bigcap_{j \in \mathbb{N}_0} D\big ( z^{(j-1)}(\vec p), r_j\big ), 
\end{equation}  
which is analytic in $ \vec p \in U_{\theta}[ \vec p^*] $. The complex number $ z^{(\infty)}(\vec p) $
is an eigenvalue of $H_\theta(\vec p)$; it is the resonance energy that we are looking for. 
It is convenient to extend the operator $H^{(j)}(\vec p, z^{(\infty)}(\vec p))$, for $j \in \mathbb{N}_0$, to an operator defined on the entire Fock space 
$\cH_f$ by defining it to vanish on the orthogonal complement of the subspace $\cH^{(j)}= \mathrm{Ran}(\mathds{1}_{H_f \leq \rho_{j}})$. We continue to use the same symbol, $ H^{(j)}(\vec p, z^{(\infty)}(\vec p)) $, for this extension. Similarly, we extend the other operators in \eqref{frej} to operators acting on the entire Fock space. 

We then show that 
\begin{equation}\label{zero.4} 
\lim_{j \to \infty }  H^{(j)}(\vec p, z^{(\infty)}(\vec p)) = 0 = H^{(\infty)}( \vec p, z^{(\infty)}(\vec p)). 
\end{equation}
In the proof of \eqref{zero.4}, we use \eqref{w00eq0}.

With some further effort, using iso-spectrality, one then shows that 
\begin{equation}\label{iso2infty} 
z^{(\infty)}(\vec p) \text{   } \text{is an eigenvalue of}\text{   } H_{\theta}(\vec{p}), 
\end{equation}
for $\Im\theta < \frac{\pi}{4}$ large enough. 
Analyticity of  $z^{(\infty)}(\vec p)$ in $\theta$ then implies that this quantity is actually \textit{independent} of $\theta$ (and this is the reason why the index $\theta$ for $z^{(\infty)}(\vec{p})$  is  omitted).

Next, we sketch the proofs of \eqref{zero.4} and of \eqref{iso2infty}. Using \eqref{fore}, \eqref{frej} and \eqref{sej}, we see that 
\begin{equation}\label{zero.1}
\lim_{j \to \infty } \|  H^{(j)}(\vec p, z^{(\infty)}(\vec p))  - W^{(j)}_{0, 0}(\vec p, z^{(\infty)}(\vec p))   \| = 0. 
\end{equation}
As explained in Section \ref{secic}, see \eqref{sim} and \eqref{hilbertj}, 
\begin{equation}\label{zero.2}
W^{(j)}_{0, 0}(\vec p, z^{(\infty)}(\vec p)) = w^{(j)}_{0,0}(\vec p, z^{(\infty)}(\vec p), H_f, \vec P_f) \mathds{1}_{H_f \leq \rho_j}.  
\end{equation}
The derivatives of $ w^{(j)}_{0,0}(\vec p, z^{(\infty)}(\vec p), r, \vec l) $ in the variables $ r $ and $\vec l$ are uniformly bounded, for $ r \in [0, \rho_j] $, $|\vec l|\leq r $ (and 
$\vert  w^{(j)}_{0,0}(\vec p, z^{(\infty)}(\vec p), r, \vec l) \vert \geq \epsilon\cdot r$, for some constant $\epsilon>0$ independent of $j$), for all $j \in \mathbb{N}_0$. These properties and the normalization condition \eqref{w00eq0} imply that 
\begin{equation}\label{zero.3}
\lim_{j \to \infty } W^{(j)}_{0, 0}(\vec p, z^{(\infty)}(\vec p)) = 0, 
\end{equation}
which, together with \eqref{zero.1}, implies \eqref{zero.4}. 

By \eqref{hilbertj}, 
$$\underset{j \in \mathbb{N}_0}{\bigcap} \mathcal{H}^{(j)} = \lbrace \mathbb{C}\Omega\rbrace.$$ 
Eq. \eqref{zero.4} then shows that the vacuum $\Omega$ is an eigenvector of 
$H^{(\infty)}( \vec p, z^{(\infty)}(\vec p))$ with eigenvalue $0$.
To prove \eqref{iso2infty} we apply part (ii) of Theorem \ref{feshbach-theorem} iteratively, after each application of a Feshbach map. \\
We define 
\begin{equation} \label{Qchij}
Q_{\chi_{\rho_j }} : = Q_{\chi_{\rho_j }}\Big( H^{(j-1)}(\vec p, z^{(\infty)}(\vec p)),  
W^{(j-1)}_{0, 0}(\vec p, z^{(\infty)}(\vec p)) +  {\mathcal E}^{(j-1)}_{0, 0}(\vec p, z^{(\infty)}(\vec p)) \Big ), 
\end{equation}
$j \in \mathbb{N}$,
where the operator $Q$ has been defined in Theorem \ref{feshbach-theorem}. One can then prove that 
\begin{equation}
\psi : = \lim_{j \to \infty }  Q_{\chi_{\rho_1 }}  \cdots Q_{\chi_{\rho_j }} (\Omega)
\end{equation}
exists. Using that $H^{(\infty)}( \vec p, z^{(\infty)}(\vec p))\Omega = 0$, we are able to show that
$$H^{(0)}(\vec{p},z^{(\infty)}(\vec{p}))\psi = 0.$$  
Then, using Theorem \ref{feshbach-theorem} once more, we conclude that [see \eqref{psi0}-\eqref{chi0}]
\begin{equation}\label{ulti}
 Q_{\bm{\chi}_{i_0}}\Big (H_\theta (\vec p) -  z^{(\infty)}(\vec p), H_{\theta, 0} (\vec p) -  z^{(\infty)}(\vec p) \Big)(\psi_{i_0} \otimes \psi)
\end{equation}  
is an eigenvector of $ H_\theta (\vec p)$ with eigenvalue $z^{(\infty)}(\vec p) $.

\subsubsection*{Acknowledgement} We are grateful to T.Chen and A.Pizzo for stimulating discussions on problems related to those studied in this paper. We are especially indebted to V. Bach and I. M. Sigal for numerous illuminating discussions and collaboration on problems closely related to those analyzed and the mathematical methods used in the present paper. The research of J. Fa. is supported by ANR grant ANR-12-JS01-0008-01. The research of B.S. is supported in part by the Region Lorraine.

\section{Parameters of the Problem, Notations}
In this section, we present a  list  of  all the parameters appearing in the analysis of the spectral problems solved in this paper. In Subsection \ref{Par22},  we  introduce the main symbols and notations used in subsequent sections. 

\label{not}
\subsection{System- and algorithmic parameters}
The quantities $\lambda_0$ (coupling constant), $\delta_0$ (spacing between energies of excited states of the atom),  $N$ (number of internal energy levels of the atom), $\sigma_{\Lambda}$ (ultraviolet cut-off imposed on the quantized electric field),  and  $\mu$ (bound on the momentum of the atom) are parameters characteristic of the physical system under investigation. They are henceforth called {\bf{system parameters}}. All our estimates depend on the choice of these parameters, and our main results only hold if suitable restrictions on the values of these parameters are imposed.
Other parameters appearing in our analysis are related to the mathematical methods applied to establish our main results, in particular to the algorithm (inductive construction) used to derive the main estimates needed in our proofs. We call them \textbf{algorithmic parameters}. Among these parameters are the dilatation parameter, $\vartheta$, appearing in the complex deformation of the basic Hamiltonian used to locate the resonance energies, and the scale parameter $\rho_0$, as well as the parameter $\varepsilon$ appearing in the definition of the Feshbach maps; see Eq. \eqref{sequence}. These (auxiliary) parameters are chosen so as to ensure (and ``optimize'') the convergence of the inductive construction outlined above. Constraints on the choice of the parameters  $\vartheta$ and $\rho_0$ are discussed in Sections \ref{firstde} and \ref{itere}, respectively. In the rest of this text, we call  {\bf{problem parameter}} any system-or algorithmic parameter.

\subsection{Notations relative to creation/annihilation operators and integrals} \label{Par22}
We introduce the notations
\begin{align*}
&\underline{k}^{(m)} := (\underline{k}_1,\dots,\underline{k}_m)\in \underline{\mathbb{R}}^{3m}, \quad \underline{\tilde{k}} \phantom{}^{(n)} := (\underline{\tilde{k}}_1,\dots,\underline{\tilde{k}}_n)\in \underline{\mathbb{R}}^{3n} , \\
&\underline{K}^{(m,n)} := ( \underline{k}^{(m)} , \underline{\tilde{k}} \phantom{}^{(n)} ) , \quad \mathrm{d} \underline{K}^{(m,n)} := \prod_{i=1}^m \mathrm{d} \underline{k}_i  \prod_{j=1}^n \mathrm{d} \underline{\tilde{k}}_j   , \\
& | \underline{K}^{(m,n)} | := | \underline{k}^{(m)} | \, | \underline{\tilde{k}} \phantom{}^{(n)} | , \quad | \underline{k}^{(m)} | := \prod_{i=1}^m | \vec{k}_i | , \quad |Ê\underline{\tilde{k}} \phantom{}^{(n)} | := \prod_{j=1}^n | \vec{\tilde{k}}_j | , \\
&a^*( \underline{k}^{(m)} ) := \prod_{ i = 1 }^m a_{\lambda_i}^*( \vec{k}_i ) , \quad a( \underline{\tilde{k}} \phantom{}^{(n)} ) := \prod_{ j = 1 }^n a_{\lambda_j}( \vec{\tilde{k}}_j ) .
\end{align*}
\vspace{2mm}

\noindent
For $\rho \in \mathbb{C}$, we set 
\begin{align*}
\rho \underline{k}^{(m)} := ( \rho \vec{k}_1 , \lambda_1 , \dots , \rho \vec{k}_m , \lambda_m ), \quad \rho \underline{K}^{(m,n)} := ( \rho \underline{k}^{(m)} , \rho \underline{ \tilde{k} } \phantom{}^{(n)} ).
\end{align*}
\vspace{1mm}

\noindent For $\underline{m}:=(m_1,...,m_L)$, $\underline{n}:=(n_1,...,n_L)$, we set
\begin{align*}
\underline{k}^{(m_i)}_i&:=(\underline{k}_{i1},...,\underline{k}_{im_i}) \in  \underline{\mathbb{R}}^{3 m_i}, \qquad \underline{k}^{(\underline{m})}:=(\underline{k}^{(m_1)}_1,...,\underline{k}^{(m_L)}_L)  \in \underline{\mathbb{R}}^{3[ \sum_i m_i] },\\
\underline{K}^{(\underline{m},\underline{n})}&:=(\underline{k}^{(m_1)}_1,...,\underline{k}^{(m_L)}_L, \underline{\tilde{k}}^{(n_1)}_1,...,\underline{\tilde{k}}^{(n_L)}_L) \in \underline{\mathbb{R}}^{3[\sum_i (m_i +n_i)]}.
\end{align*}
\vspace{1mm}

\noindent For $\rho \in \mathbb{R}_+$, and $m,n \in \mathbb{N}$, we introduce 
\begin{align*}
\underline{B}_{\rho}&:= \{  \underline{k} \in \underline{\mathbb{R}}^3 \mid  \vert \vec{k} \vert \leq \rho \},\\[4pt]
\underline{B}_{\rho}^{(m)}&:= \{  (\underline{k}_1,..., \underline{k}_m) \in \underline{\mathbb{R}}^{3m} \mid \sum_{i=1}^{m} \vert \vec{k}_i \vert \leq \rho \},\\[6pt]
\underline{B}_{\rho}^{(m,n)}&:=\underline{B}_{\rho}^{(m)} \times \underline{B}_{\rho}^{(n)}.
\end{align*}
\vspace{0.5mm}

\subsection{Kernels and their domain of definition}
Let $\rho>0$. We set   
\begin{align}
\label{BB}
\mathcal{B}_{\rho } &:= \{ ( r , \vec{l} ) \in [ 0 ,\rho ] \times \mathbb{R}^3 , | \vec{l} | \le r \}.
\end{align}
Let $w_{m,n}$ be a function  
$$w_{m,n}: \mathcal{B}_{\rho }   \times \underline{B}_{\rho }^{(m,n)}    \to  \mathbb{C}. $$

\noindent  We introduce  
\begin{equation} \label{1/2}
\|w_{m,n} \|_{\frac{1}{2}} := \sup_{(r,\vec{l}) \in \mathcal{B}_{\rho }}  \underset{ \underline{K}^{(m,n)}\in \underline{B}_{\rho }^{(m,n)}}{\text{ ess sup}}  \frac{ \vert w_{m,n}\left(r,\vec{l},  \underline{K}^{(m,n)}\right) \vert }{ \vert  \underline{k}^{(m)} \vert ^{1/2} \vert \underline{\tilde{k}}^{(n)} \vert^{1/2}}.
\end{equation}
\vspace{1mm}

\noindent   The choice of the norm   $\| \cdot \|_{\frac{1}{2} }$ is motivated by  the infrared behavior of the interaction Hamiltonian $H_{I}$, which behaves like $ \vec{k} \mapsto \vert \vec{k} \vert^{1/2}$  for small values of $\vert \vec{k} \vert$.  Lemma \ref{kerop} below establishes the link between the norm of the operator $W_{m,n}$ and the norm $\|w_{m,n} \|_{\frac{1}{2} }$ of its associated kernel. Finally, if $w_{0,0}: \mathcal{B}_{\rho } \rightarrow \mathbb{C}$ is essentially bounded,  we set
\begin{equation}
\|w_{0,0} \|_{\infty} = \underset{ (r,\vec{l}) \in \mathcal{B}_{\rho} }{ \text{ ess sup} }  \, \vert w_{0,0}\left(r,\vec{l}\right) \vert. 
\end{equation}
\vspace{1mm}

\subsection{Notations relative to estimates}
Many numerical constants appear in our estimates. Keeping track of all these constants would be very cumbersome and is not necessary for mathematical rigor.  Let $a,b >0$.  We write
\begin{equation}
a=\mathcal{O}(b)
\end{equation}
if there is a numerical constant $C>0$ independent of the system and algorithmic parameters  such that $a \le Cb$.  

The shorthand
 \begin{equation}
\text{``}\text{for all } a \ll b , \dots \text{''}
 \end{equation}
means that ``there exists a (possibly very small, but) positive numerical constant $C$ independent of the system and algorithmic  parameters such that, for all $a \le C b$, \dots''
\vspace{2mm}

\section{The first decimation step}
\label{firstde}
Here we present details of the results described in Section \ref{1.3.1}. We use the notations introduced there.
 
In Subsection \ref{technical}, we  state two standard lemmas that we use repeatedly in our analysis.  The proofs  are  postponed to Appendix \ref{AppendixB} for the  reader's convenience.

In Subsection \ref{feshp}, we prove under suitable assumptions that the pair $(H_{\theta}(\vec{p})-z,H_{\theta,0}(\vec{p}) -z)$ is a Feshbach-Schur pair  associated to the generalized projection $\bm{\chi}_{i_{0}}= P_{i_0} \otimes \chi_{\rho_0}(H_f)$ defined in  \eqref{chi0}. This result holds for  all $(\vec{p},z)$  in  the open set  $\mathcal{U}_{\rho_0}[E_{i_{0}}]$, where
\label{premiere}
 \begin{equation}
 \label{U-11}
   \mathcal{U}_{\rho_0}[E_{i_{0}}]:= U_{\theta}[\vec{p}^*] \times D(E_{i_{0}},r_0); \end{equation}

\noindent see \eqref{UP} and \eqref{sequence}.
We remind the reader that the operator $H^{(0)}(\vec{p},z)$ is the partial trace over the internal degrees of freedom of the restriction of the  Feshbach operator  $F_{\bm{\chi}_{i_0}}(H_{\theta}(\vec{p})-z,H_{\theta,0}(\vec{p}) -z)$ to $\mathds{1}_{H_f \leq \rho_0}(\mathcal{H}_f)$; see \eqref{Wo}-\eqref{sim}. In Subsection \ref{par2}, we show that    $H^{(0)}(\vec{p},z)$ can be rewritten as a convergent series of Wick monomials and that it is analytic in $(\vec{p},z)$ on the open set $\mathcal{U}_{\rho_0}[E_{i_{0}}]$.   Details of the proofs are postponed to Appendix \ref{AppA} and  \ref{Analyfirst}.  In Lemma  \ref{zeros}, we prove that  there exists  a unique element $z^{(0)}(\vec{p}) \in   D(E_{i_{0}},r_0)$, for each $\vec{p} \in U_{\theta}[\vec{p}]$,  such that $\mathcal{E}^{(0)} (\vec{p},z^{(0)}(\vec{p}))=0$. The properties of the kernels $w_{m,n}^{(0)}$ and the function $\mathcal{E}^{(0)}$   established in  Lemmas \ref{borne0} and \ref{zeros} are the basis of the inductive construction described in  Section \ref{itere}.

\subsection{Feshbach-Schur Pair}
\subsubsection{Two Lemmas}
\label{technical}  
We begin with a lemma showing that the norm of the Wick monomials is controlled by the norm of their associated kernels. The proof is standard and deferred to Appendix \ref{AppendixB}.
  \begin{lemma}
\label{kerop}
Let $\rho  > 0$. Let $w_{m,n}$ be a function  $w_{m,n}: \mathcal{B}_{\rho }   \times \underline{B}_{\rho }^{(m,n)}    \to  \mathbb{C}$ with  $\|w_{m,n} \|_{\frac{1}{2} } < \infty$,  and let $W_{m,n} $ be the  Wick monomial  on $\mathds{1}_{H_f \leq \rho} \mathcal{H}_f$,  defined in the sense of  quadratic forms  by
\begin{equation*}
W_{m,n} :=    \mathds{1}_{H_f \leq \rho}  \left( \int_{\underline{ B}_{\rho }^{(m,n)}} d \underline{K}^{(m,n)} a^{*}(\underline{k}^{(m)}) w_{m,n}\left(H_f,\vec{P_f},  \underline{K}^{(m,n)}\right) a(\underline{\tilde{k}}^{(n)}) 
 \right)    \mathds{1}_{H_f \leq \rho }. 
\end{equation*}
  Then
\begin{equation}\label{eq:a1}
\|W_{m,n} \| \leq  (8 \pi)^{\frac{m+n}{2}} \rho^{2(m+n)} \|w_{m,n} \|_{\frac{1}{2} }.
\end{equation}
\end{lemma}
\vspace{2mm}

The next lemma will be used  in the remainder of this section. Again, its proof is deferred to Appendix \ref{AppendixB}. We remind the reader  that $\sigma_{\Lambda}$ is the ultraviolet cut-off parameter  that appears in  the interacting Hamiltonian $H_{I}$ and that $\delta_0$ denotes the minimal distance between two distinct eigenvalues of $H_{is}$.
  \begin{lemma} \label{tech1}
  \text{ }
\begin{list}{\labelitemi}{\leftmargin=1em}
\item  Let $0 < \rho  < 1$. For all $\theta = i \vartheta$, $0 < \vartheta < \pi / 4$, we have that
  \begin{align}
  \label{HI1}
  \big \| (H_f + \rho )^{-1/2} H_{I,\theta}  (H_f + \rho )^{-1/2} \big \| & =  \mathcal{O} \left(  \frac{\sigma_{\Lambda}^{3/2}}{ \rho ^{1/2}}\right). 
  \end{align}
  \vspace{2mm}
  
 \item Let $0 < \rho_0 < \min ( 1 , \delta_0 )$. For all $\theta = i \vartheta$, $0 < \vartheta < \pi / 4$, and $(\vec{p},z) \in \mathcal{U}_{\rho_0}[E_{i_{0}}]$, the operator $[H_{\theta,0}(\vec{p})-z]_{| \mathrm{Ran}(\overline{\bm{\chi}}_{i_{0}})}$  is bounded invertible and satisfies the estimates
\begin{align}
\label{Htheta0}
\big \| [H_{\theta,0}(\vec{p})-z]^{-1}_{  | \mathrm{Ran}(\overline{\bm{\chi}}_{i_{0}})}  \big \|& = \mathcal{O}\left( \frac{1}{ \mu \rho_0 \sin(\vartheta)} \right),\\[4pt]
\label{Htheta1}
\big \|   \left[ (H_{\theta,0} (\vec{p})-z)^{-\frac{1}{2}}(H_f +  \rho_0)(H_{\theta,0} (\vec{p})-z)^{-\frac{1}{2}} \right]_{| \mathrm{Ran}(\overline{\bm{\chi}}_{i_{0}})} \big \| &= \mathcal{O}\left( \frac{1}{ \mu   \sin(\vartheta)} \right). 
\end{align}
\end{list}
  \end{lemma}
\vspace{1mm}

\subsubsection{ $(H_{\theta}(\vec{p})-z,H_{\theta,0}(\vec{p})-z)$ is a Feshbach-Schur pair} \label{feshp}
We now show  that the pair $(H_{\theta}(\vec{p})-z,H_{\theta,0}(\vec{p})-z)$ is a Feshbach-Schur pair, provided that the coupling constant $\lambda_0$ is small enough and that the scale parameter $\rho_0$ satisfies $\rho_0 \gg \lambda_0^2 ( \mu \sin \vartheta )^{-2}$.
\label{par1}
\begin{lemma}
\label{fespair}
There exists $\lambda_c > 0$ such that, for all $0 \le \lambda_0 \le \lambda_c$, $\theta = i \vartheta$ satisfying $0 < \vartheta < \pi / 4$ and $\sigma_\Lambda^{-3/2} \mu \sin \vartheta \gg \lambda_0$, $(\vec{p},z) \in \mathcal{U}_{\rho_0}[E_{i_{0}}]$, and $\rho_0$ such that $\lambda_0^2 \sigma_\Lambda^{3} (\mu \sin \vartheta )^{-2} \ll \rho_0 < \min ( 1 , \delta_0 )$, the pair $(H_{\theta}(\vec{p})-z,H_{\theta,0}(\vec{p})-z)$ is a Feshbach-Schur pair associated to $\bm{\chi}_{i_{0}}$.
\end{lemma}
\vspace{1mm}

\begin{proof}
  Lemma \ref{tech1} shows that   $[H_{\theta,0}(\vec{p})-z]_{| \mathrm{Ran}(\overline{\bm{\chi}}_{i_{0}})}$  is bounded invertible for all  $(\vec{p},z) \in \mathcal{U}_{\rho_0}[E_{i_{0}}]$. We prove that 
\begin{equation}
H_{\overline{\bm{\chi}}_{i_{0}}}(\vec{p},z):= H_{\theta,0}(\vec{p})-z   + \lambda_0 \overline{\bm{\chi}}_{i_{0}} H_{I,\theta} \overline{\bm{\chi}}_{i_{0}}
\end{equation} 
is bounded invertible on $ \mathrm{Ran}(\overline{\bm{\chi}}_{i_{0}})$. The proof  is standard and relies on Equation \eqref{HI1} in Lemma \ref{tech1}. By \eqref{HI1}, the Neumann series for $[H_{\overline{\bm{\chi}}_{i_{0}}}(\vec{p},z)]^{-1}_{| \mathrm{Ran}(\overline{\bm{\chi}}_{i_{0}})}$ is estimated as
\begin{eqnarray}  \label{neumanb}
\big \| [H_{\overline{\bm{\chi}}_{i_{0}}}(\vec{p},z)]^{-1}_{| \mathrm{Ran}(\overline{\bm{\chi}}_{i_{0}})} \big \|  
& \leq & \rho_{0}^{-1}  \sum_{n=0}^{\infty} \big [ C \sigma_{\Lambda}^{3/2} \lambda_0 \rho_{0}^{-1/2} \big ]^{n} \times \nonumber\\
&  & \Big \|  \left[ (H_{\theta,0} (\vec{p})-z)^{-\frac{1}{2}}(H_f +  \rho_0)(H_{\theta,0} (\vec{p})-z)^{-\frac{1}{2}}  \right]_{| \mathrm{Ran}(\overline{\bm{\chi}}_{i_{0}})}\Big \|^{n+1} ,
\end{eqnarray}
for some numerical constant $C > 0$. Using \eqref{Htheta1}, we see that the Neumann series converges uniformly in $(\vec{p},z) \in \mathcal{U}_{\rho_0}[E_{i_{0}}]$ provided that $\lambda_0 \ll \sigma_\Lambda^{-3/2} \rho_0^{1/2} \mu \sin \vartheta$. Moreover, 
\begin{equation}
\label{neumanb}
\big \| [H_{ \overline{\bm{\chi}}_{i_{0}}}(\vec{p},z)]^{-1}_{| \mathrm{Ran}(\overline{\bm{\chi}}_{i_{0}})} \big \| = \mathcal{O}\left(\frac{ 1}{\mu \rho_0 \sin(\vartheta)} \right).
\end{equation}
Since in addition $H_{I,\theta} \bm{\chi}_{i_0}$ and $\bm{\chi}_{i_0} H_{I,\theta}$ extend to bounded operators on $\mathcal{H}_{ \vec{p} }$, it follows that $(H_{\theta}(\vec{p})-z,H_{\theta,0}(\vec{p})-z)$ is a Feshbach-Schur pair associated to $\bm{\chi}_{i_{0}}$.
\end{proof}

\subsection{Wick-ordering  and analyticity of $H^{(0)}(\vec{p},z)$}
\label{par2}
We assume that the parameters $\rho_0$, $\lambda_0$ and $\theta$ satisfy the  hypotheses of Lemma \ref{fespair}, so that the smooth Feshbach-Schur map associated to $\bm{\chi}_{i_{0}}$ can be applied to the pair $(H_{\theta}(\vec{p})-z,H_{\theta,0}(\vec{p})-z)$ for all $(\vec{p},z) \in \mathcal{U}_{\rho_0}[E_{i_{0}}]$.  Let  $H^{(0)}(\vec{p},z)$ be defined as in \eqref{Wo}-\eqref{sim}. More precisely, $H^{(0)}(\vec{p},z)$ is the  bounded operator on $ \mathcal{H}^{(0)}=\mathds{1}_{H_f \leq \rho_0} \mathcal{H}_f$  associated with the bounded quadratic form defined by
\begin{equation}
\langle \psi \vert  H^{(0)}(\vec{p},z)  \phi \rangle  =  \langle \psi_{i_0}  \otimes \psi  \vert  F_{\bm{\chi}_{i_{0}}}(H_{\theta}(\vec{p})-z,H_{\theta,0}(\vec{p})-z)    \psi_{i_0} \otimes \phi \rangle,
\end{equation}
for all $\psi,\phi \in \mathcal{H}^{(0)}$, where $ \psi_{i_0} $ is a normalized eigenvector associated to the eigenvalue $E_{i_{0}}$ of $H_{is}$. Here we omit the argument $\theta$ to simplify notations. Lemma \ref{borne0} below shows that  $H^{(0)}(\vec{p},z)$ can be rewritten as a convergent series of Wick monomials on $\mathcal{H}^{(0)}$; see \eqref{sim}. The convergence is uniform   on the open set $\mathcal{U}_{\rho_0}[E_{i_{0}}]$.  The main  tool used in  the proof is the pull-though formula
\begin{equation} \label{pull}
a(\underline{k}) g(H_f,\vec{P}_f)=g(H_f + \vert \vec{k} \vert, \vec{P}_f + \vec{k}) a(\underline{k}),
\end{equation}
which  holds for any measurable function $g: \mathbb{R}^4 \rightarrow \mathbb{C}$, and which enables us to normal order the  creation and annihilation operators that appear in $H^{(0)}(\vec{p},z)$. Lemma \ref{borne0} also  shows that $H^{(0)}(\vec{p},z)$ can be made arbitrary close (in norm) to the  operator
\begin{equation}
 \left( e^{- \theta} H_f -  e^{- \theta}\vec{p} \cdot \vec{P}_f + E_{i_{0}} -z \right)_{\mid \mathcal{H}^{(0)}}
\end{equation}
by an appropriate tuning of the  coupling constant $\lambda_0$,   and  that the map $(\vec{p},z) \mapsto H^{(0)}(\vec{p},z) \in \mathcal{L}( \mathcal{H}^{(0)})$ is analytic on $\mathcal{U}_{\rho_0}[E_{i_{0}}]$. The proof of Lemma \ref{borne0} is postponed to Appendix \ref{AppendixC}.

\begin{lemma}
\label{borne0}
Let $\gamma > 0$. There exists $\lambda_c( \gamma ) > 0$ such that, for all $0 \leq \lambda_0  \le  \lambda_c ( \gamma ) $ and $\theta$ and $\rho_0$ as in Lemma \ref{fespair}, $H^{(0)} $ can be rewritten as a uniformly convergent series of Wick monomials on $\mathcal{U}_{\rho_0}[E_{i_0}]$,
\begin{equation}
H^{(0)}(\vec{p},z)=H[\underline{w}^{(0)}(\vec{p},z), \mathcal{E}^{(0)}(\vec{p},z)]= \sum_{m+n \geq 0} W_{m,n}^{(0)}(\vec{p},z)  + \mathcal{E}^{(0)}(\vec{p},z).
\end{equation}
The associated kernels  $$w_{m,n}^{(0)}: \mathcal{U}_{\rho_0}[E_{i_{0}}] \times \mathcal{B}_{\rho_0} \times \underline{B}_{\rho_0}^{(m,n)} \rightarrow \mathbb{C}$$
and the function $\mathcal{E}^{(0)}:  \mathcal{U}_{\rho_0}[E_{i_{0}}] \rightarrow \mathbb{C}$ satisfy the following properties:

\begin{list}{\labelitemi}{\leftmargin=1em}
\item $w^{(0)}_{0,0}(\vec{p},z,\cdot,\cdot)$ is $\mathcal{C}^1$ on $\mathcal{B}_{\rho_0}$ and  $w_{0,0}^{(0)}(\vec{p},z,0,\vec{0})=0$  for all  $ (\vec{p},z) \in \mathcal{U}_{\rho_0}[E_{i_{0}}]$,
\vspace{2mm}

\item  $w^{(0)}_{m,n}(\vec{p},z, \cdot, \cdot, \underline{K}^{(m,n)})$, $m+n \geq 1$,  are $\mathcal{C}^1$ on $\mathcal{B}_{\rho_0}$ for almost every $\underline{K}^{(m,n)} \in \underline{B}_{\rho_0}^{(m,n)}$ and  every $ (\vec{p},z) \in \mathcal{U}_{\rho_0}[E_{i_{0}}]$,
\vspace{2mm}

\item For all $(\vec{p},z) \in \mathcal{U}_{\rho_0}[E_{i_{0}}]$, 
\begin{align}
\label{Bb00}
\| w^{(0)}_{m,n} (\vec{p},z) \|_{\frac{1}{2}}& \le \gamma \mu   \sin(\vartheta)   \rho_{0}^{- \frac12 (m+n) + 1} ,\\[4pt]
\| \partial_{\#} w^{(0)}_{m,n}(\vec{p},z) \|_{\frac{1}{2} }& \le \gamma \rho_{0}^{- \frac12 (m+n)} ,
\end{align}
for all $m+n \ge 1$, where $\partial_\#$ stands for $\partial_r$ or $\partial_{l_j}$, and
\begin{align}
\vert (E_{i_{0}} -z) - \mathcal{E}^{(0)} (\vec{p},z) \vert &\le \gamma \mu \rho_0 \sin(\vartheta),\\[4pt]
\label{Bb000}
\| \partial_r w_{0,0}^{(0)}(\vec{p},z) -e^{- \theta} \|_{\infty} + \sum_{q=1}^{3}  \|  \partial_{l_q} w_{0,0}^{(0)}(\vec{p},z) +  p_q  e^{ - \theta} \|_{\infty} & \le \gamma + \sqrt{3 } \rho_0 . \qquad 
\end{align}
\end{list}
Moreover,  the bounded operator-valued function $(\vec{p},z) \mapsto H^{(0)}(\vec{p},z) \in \mathcal{L}(\mathcal{H}^{(0)})$ is analytic on $\mathcal{U}_{\rho_0}[E_{i_{0}}]$.
\end{lemma}
\vspace{3mm}

 Since $(\vec{p},z) \mapsto H^{(0)}(\vec{p},z) \in \mathcal{L}(\mathcal{H}^{(0)})$ is analytic on the open set $\mathcal{U}_{\rho_0}[E_{i_{0}}]$,  the map $(\vec{p},z) \mapsto \mathcal{E}^{(0)}(\vec{p},z) = \langle \Omega \vert  H^{(0)}(\vec{p},z)   \Omega \rangle \in \mathbb{C}$ is also analytic. Our next lemma establishes, for each $\vec{p} \in U_{\theta}[\vec{p}^*]$, the existence of a unique element $z^{(0)}(\vec{p}) \in   D(E_{i_{0}},r_0)$,  such that $\mathcal{E}^{(0)} (\vec{p},z^{(0)}(\vec{p}))=0$. Here we recall that $r_0 = \rho_0 \mu \sin ( \vartheta ) / 32$. 
\vspace{1mm}

\begin{lemma}
\label{zeros}
Let $0 < \gamma \ll 1 $ and suppose that $\lambda_0$, $\rho_0$, $\theta = i \vartheta$ are fixed as in Lemma \ref{fespair}. Let  $\vec{p} \in U_{\theta}[\vec{p}^* ]$. The holomorphic function $z \mapsto   \mathcal{E}^{(0)}(\vec{p},z) \in \mathbb{C}$ possesses a unique zero $z^{(0)}(\vec{p}) \in  D(E_{i_{0}},r_0)$. Furthermore, for any $\eta>0$ such that $r_0 \eta +\gamma \rho_0 \mu \sin(\vartheta) <r_0$,  $D(z^{(0)}(\vec{p}),r_0 \eta  ) \subset D(E_{i_{0}},r_0)$, and
\begin{align}
\vert \mathcal{E}^{(0)}(\vec{p},z)\vert &\le  2 \gamma \rho_0 \mu \sin ( \vartheta ) + r_0 \eta,\\
 \vert z^{(0)}(\vec{p}) - E_{i_0}\vert &\le   \gamma \rho_0 \mu \sin ( \vartheta ),
\end{align}
for all $z \in  D(z^{(0)}(\vec{p}),r_0  \eta )$.
\end{lemma}
\vspace{2mm}

\begin{proof}
Since $\gamma \ll 1$, we have that $r_0 > \gamma \rho_0 \mu \sin ( \vartheta )$. We apply Rouch\'e's theorem to the functions  $z \mapsto \mathcal{E}^{(0)}(\vec{p},z)$ and $z \mapsto E_{i_{0}}-z$ on $D(E_{i_{0}},r)$ with $\gamma \rho_0 \mu \sin ( \vartheta ) < r < r_0$. For any $z \in \partial D(E_{i_{0}},r)$, we have that $\vert E_{i_{0}}-z \vert = r$, and hence
\begin{equation*}
\vert \mathcal{E}^{(0)}(\vec{p},z) - (E_{i_{0}}-z) \vert \le \gamma \rho_0 \mu \sin ( \vartheta ) < \vert E_{i_{0}}-z \vert ,
\end{equation*}
for any  $r > \gamma \rho_0 \mu \sin( \vartheta )$. As $r$ can be chosen arbitrarily close to $r_0$, we deduce that $z \mapsto \mathcal{E}^{(0)}(\vec{p},z)$ possesses a unique zero, $z^{(0)}(\vec{p})$, in $D(E_{i_{0}},r_0)$. Let  $\eta>0$ such that $r_0 \eta +\gamma \rho_0 \mu \sin(\vartheta) <r_0$. The triangle inequality  implies that
\begin{equation*}
\vert \mathcal{E}^{(0)}(\vec{p},z) \vert  \le \gamma \rho_0 \mu \sin( \vartheta ) + \vert z - z^{(0)}(\vec{p}) \vert +  \vert  z^{(0)}(\vec{p}) -E_{i_{0}} \vert \le 2 \gamma \rho_0 \mu \sin( \vartheta ) + r_0 \eta ,
\end{equation*}
for all $z \in  D(z^{(0)}(\vec{p}),r_0 \eta  ) \subset D(E_{i_{0}},r_0)$.
\end{proof}

\vspace{0,5cm}

\section{The inductive construction}
\label{itere}
 As described in Subsection \ref{secic}, we propose to inductively construct a sequence of effective Hamiltonians, $H^{(j)}(\vec{p},z)$, $j=0,1,...$, with the property that  $H^{(j)}(\vec{p},z^{(\infty)}(\vec{p}))$  is not injective if and only if $z^{(\infty)}(\vec{p})$ is an eigenvalue of $H_{\theta}(\vec{p})$. We use the notations introduced in Section \ref{intro}, and we  now present the details of our inductive construction. In particular, one of our purposes in this section is to prove bounds on the perturbation $W_{\geq 1}^{(j)}(\vec{p},z)$; see \eqref{sej}.  We remind the reader that our  inductive construction can be summarized  by describing the induction step, from $j$ to $j+1$:
\vspace{1mm}

\begin{list}{\labelitemi}{\leftmargin=1em}
\item[(i)] In passing from $j$ to $j+1$, our starting point is the effective Hamiltonian  
$H^{(j)}(\vec{p},z)$ constructed in the previous induction step, which is an operator defined on the space 
$\mathcal{H}^{(j)}=\mathds{1}_{H_f \leq \rho_j} \mathcal{H}_f$, provided $(\vec{p},z)$ is constrained to belong to a certain open subset $\mathcal{U}_{\rho_j}[z^{(j-1)}]$ of $\mathbb{C}^{3}\times \mathbb{C}$. For each 
$\vec{p} \in U_{\theta}[\vec{p}^*]$, the admissible values of $z$ are then taken to lie inside a small disk centered at the zero,  $z^{(j)}(\vec{p})$, of the function $$z \mapsto \mathcal{E}^{(j)}(\vec{p},z)= \langle \Omega \vert H^{(j)}(\vec{p},z)  \Omega \rangle.$$ 
This will define an open set $\mathcal{U}_{\rho_{j+1}}[z^{(j)}] \subset  \mathcal{U}_{\rho_j}[z^{(j-1)}]$.

\item[(ii)] We apply the Feshbach-Schur map to the Feshbach pair $(H^{(j)}(\vec{p},z),W_{0,0}^{(j)}(\vec{p},z)+\mathcal{E}^{(j)}(\vec{p},z))$ associated to ${\chi}_{\rho_{j+1}}(H_f)$, for all $(\vec{p},z)$ in  $\mathcal{U}_{\rho_{j+1}}[z^{(j)}]$. We then re-Wick order the resulting operator (all creation operators moved to the left of all annihilation operators, using the pull-through formula). This yields a new effective Hamiltonian, 
$H^{(j+1)} (\vec{p},z)$, at step $j+1$, which will be shown to be well-defined on  $\mathcal{H}^{(j+1)}=\mathds{1}_{H_f \leq \rho_{j+1}} \mathcal{H}_f$, provided $(\vec{p},z) \in \mathcal{U}_{\rho_{j+1}}[z^{(j)}]$.  
\end{list}
\vspace{1mm}

\subsection{Inductive properties of the kernels -- from an energy scale $\rho$ to the energy scale $  \rho^{2-\varepsilon}$} \label{nton+1}

We first consider an effective Hamiltonian, given as a sum of Wick monomials, at an energy scale $\rho$, with $0 < \rho  <1$. By an application of the smooth Feshbach-Schur map, we obtain a new effective Hamiltonian at an energy scale $\tilde{\rho}$, with
\begin{equation} \label{eq:def_tilderho}
\tilde \rho := \rho^{2 - \varepsilon } ,  \qquad 0<\varepsilon<1,
\end{equation}
which has certain properties allowing us to iterate the construction. For a kernel $w_{m,n}(\vec{p},z)$ defined on  a subset $\mathcal{S}$ of $\mathcal{B}_{\rho }   \times \underline{B}_{\rho }^{(m,n)} $, $\|w_{m,n}(\vec{p},z)\|_{1/2}$ is the norm of $w_{m,n}(\vec{p},z)$ as  defined in \eqref{1/2}  with the supremum  taken over  the subset $\mathcal{S}$.

For $f: \mathbb{C}^3 \rightarrow \mathbb{C}$ a continuous function and $\rho>0$, we define 
 \begin{equation}
\label{UU}
\mathcal{U}_{\rho}[f] := \left \{ ( \vec{p} , z ) \in U_{\theta}[\vec{p}^*] \times \mathbb{C} \ | \ z \in  D\left(f(\vec{p}), r_\rho \right) \right \}, 
\end{equation}
where $D\left(f(\vec{p}), r_\rho \right)$ is the complex open disk centered at $f(\vec{p})$ and with radius 
$$
r_{\rho} := \frac{\rho \mu \sin(\vartheta)}{32}.
$$
For $( \vec{p} , z ) \in \mathcal{U}_\rho[f]$, we consider the operator
\begin{equation} \label{eq:defH_n1}
H(\vec{p},z)=H[\underline{w} (\vec{p},z), \mathcal{E}(\vec{p},z) ]= \sum_{m+n \geq 0} W_{m,n} (\vec{p},z)  + \mathcal{E}(\vec{p},z) ,
\end{equation}
on $\mathds{1}_{H_f \le \rho} \mathcal{H}_f$, associated with a sequence of kernels  $$w_{m,n} : \mathcal{U}_\rho[f] \times \mathcal{B}_{\rho} \times \underline{B}_{\rho}^{(m,n)} \rightarrow \mathbb{C}$$
and a function $\mathcal{E} :  \mathcal{U}_\rho[f] \rightarrow \mathbb{C}$. We assume that there exists a constant $\mathbf{D}>1$ such that 
\vspace{2mm}
\begin{enumerate}[(a)]
\item $w_{0,0}(\vec{p},z,\cdot,\cdot)$ is $\mathcal{C}^1$ on $\mathcal{B}_{\rho}$, and $w_{0,0}(\vec{p},z,0,\vec{0})=0$,  for all  $ (\vec{p},z) \in \mathcal{U}_\rho[f]$;
\vspace{2mm}

\item  the kernels $w_{m,n} (\vec{p},z, \cdot, \cdot, \underline{K}^{(m,n)})$, $m+n \geq 1$,  are $\mathcal{C}^1$ on $\mathcal{B}_{\rho}$, for almost every $\underline{K}^{(m,n)} \in \underline{B}_{\rho}^{(m,n)}$ and  every $ (\vec{p},z) \in \mathcal{U}_\rho[f]$. Moreover, $w_{m,n}(\vec{p},z, \cdot, \cdot, \underline{K}^{(m,n)})$ is symmetric in $\underline{k}^{(m)}$ and $\underline{\tilde{k}}^{(n)}$;
\vspace{2mm}

\item for all $(\vec{p},z) \in \mathcal{U}_{\rho}[f]$, 
\begin{align}
\label{Bb00_n1}
\| w_{m,n} (\vec{p},z) \|_{\frac{1}{2}}& \le \mathbf{D}^{m+n} \rho^{- (m+n) + 1} \mu \sin(\vartheta) ,\\[12pt]
\label{Bb00_n2}
\| \partial_{\#} w_{m,n}(\vec{p},z) \|_{\frac{1}{2}}& \le    \mathbf{D}^{m+n}    \rho^{- (m+n) },
\end{align}
for all $m+n \ge 1$, where $\partial_\#$ stands for $\partial_r$ or $\partial_{l_j}$;

\item the maps $\mathcal{E}  : \mathcal{U}_{\rho}[f]  \rightarrow \mathbb{C}$ and $(\vec{p},z) \mapsto H( \vec{p} , z )   \in \mathcal{L}(\mathds{1}_{H_f \leq \rho}\mathcal{H}_f)$ are analytic on  $\mathcal{U}_{\rho}[f]$;

\item for all $\vec{p} \in   U_{\theta}[\vec{p}^*] $, the holomorphic function $z \mapsto   \mathcal{E} (\vec{p},z) \in \mathbb{C}$ possesses a unique zero $\tilde f (\vec{p}) \in  D( f( \vec{p} ) , 2 r_\rho / 3 )$, where $\tilde f$ is analytic in $\vec{p} \in U_{\theta}[\vec{p}^*] $; with $\tilde \rho$ defined by \eqref{eq:def_tilderho}, we then have that $$ \mathcal{U}_{\tilde \rho}[ \tilde f ] \subset  \mathcal{U}_{\rho}[f],$$ 
\begin{equation} \vert \tilde{f}(\vec{p})- f(\vec{p})\vert <\frac{r_\rho}{2}, \label{4.47_n1} 
\end{equation}
 for all $\vec{p} \in U_{\theta}[\vec{p}^*]$, and 
\begin{equation}\label{eq:eps_n1}
\vert \mathcal{E}(\vec{p},z)\vert \le \frac{ \mu \tilde \rho }{ 16 } ,
\end{equation}
for all $( \vec{p} , z ) \in \mathcal{U}_{\tilde \rho}[ \tilde f ]$.
\end{enumerate}

 \begin{lemma}
\label{borne0_n1} 
Let $\rho, \varepsilon \in (0,1)$  and let $\mathbf{D}>1$ be as in \eqref{Bb00_n1} and  \eqref{Bb00_n2}   and  such that 
$\rho^{ \varepsilon} \ll  \mathbf{D}^{-1}$ and $\rho^{ 1-\varepsilon} \ll 1$. Let $f : \mathbb{C}^3 \to \mathbb{C}$ be a continuous function, and let $H( \vec{p} , z )$ be the operator given in \eqref{eq:defH_n1}. 
For $( \vec{p} , z ) \in \mathcal{U}_\rho[f]$, this operator is assumed to satisfy properties $(a)$--$(e)$, above. In addition, we assume that
\begin{align} \label{4.9}
\| \partial_r w_{0,0}(\vec{p},z) -e^{- \theta} \|_{\infty} + \sum_{q=1}^{3}  \|  \partial_{l_q} w_{0,0}(\vec{p},z) +  p_q  e^{ - \theta} \|_{\infty} & \le \frac{ \mu }{ 4 } , \qquad  \forall (\vec{p},z) \in \mathcal{U}_{\rho}[f] ,
\end{align}
and that
\begin{equation} \label{partialz_n1}
\vert \partial_z \mathcal{E}(\vec{p},z) +1 \vert \le \frac{1}{4}  , \qquad \forall z \in \overline{D} (f(\vec{p}),\frac{2}{3} r_{ \rho}) ,
\end{equation}
where $\overline{D} (f(\vec{p}),\frac{2}{3} r_{ \rho})$ denotes the closed disk with radius 
$2r_\rho / 3$ centered at $f(\vec{p})$. \\
Then, for $\tilde \rho = \rho^{2-\varepsilon}$, the effective Hamiltonian   \normalfont
\begin{equation*}
\tilde H(\vec{p},z):= F_{\chi_{ \tilde \rho }(H_f)} [H (\vec{p},z), W_{0,0}(\vec{p},z) + \mathcal{E}(\vec{p},z)]_{ \mathds{1}_{H_f \leq  \tilde \rho} [\mathcal{H}_f]} ,
\end{equation*}
\itshape is well-defined for all $(\vec{p},z) \in \mathcal{U}_{ \tilde \rho }[ \tilde f ]$, and there exists a sequence of kernels 
$$
\tilde w_{m,n} : \mathcal{U}_{ \tilde \rho }[ \tilde f ] \times \mathcal{B}_{ \tilde \rho } \times \underline{B}_{ \tilde \rho }^{(m,n)} \rightarrow \mathbb{C} ,
$$
and a function $\tilde{\mathcal{E}} : \mathcal{U}_{ \tilde \rho }[ \tilde f ] \rightarrow \mathbb{C}$, such that  
\begin{equation}\label{normalform}
\tilde H(\vec{p},z) = H[\underline{\tilde w}(\vec{p},z), \tilde{\mathcal{E}}(\vec{p},z) ]= \sum_{m+n \geq 0} \tilde{W}_{m,n} (\vec{p},z) + \tilde{\mathcal{E}}(\vec{p},z) ,
\end{equation}
for all $( \vec{p} , z ) \in U_{ \tilde \rho }[ \tilde f ]$. The maps $\tilde{w}_{m,n}$ and $\tilde{\mathcal{E}}$ have properties $(a)$--$(e)$, above, with $\rho$ replaced by $\tilde \rho$, $f $ by $\tilde{f}$, $\mathbf{D}$ replaced by $ \mathbf{DC}$, for some constant $\mathbf{C}$ independent of the 'problem parameters' and $\mathbf{D}$, and  $\tilde{f}$ replaced by $\tilde{\tilde{f}} $.  Here $\tilde{\tilde{f}}( \vec{p} )$ denotes the unique zero of the map $z \mapsto \tilde{\mathcal{E}} (\vec{p},z) \in \mathbb{C}$ in $D( \tilde f( \vec{p} ) , 2 r_{ \tilde \rho } / 3 )$.\\
Moreover, we have that
\begin{align}
& \| \tilde w_{0,0} (\vec{p},z) + \tilde{ \mathcal{E} } (\vec{p},z) - w_{0,0} (\vec{p},z) -\mathcal{E} (\vec{p},z) \|_{\infty} \leq \mathbf{D}^2 \mathbf{C} \rho^{2+\varepsilon} \mu  \sin^2(\vartheta) , \label{n1_n1} \\
& \| \partial_\# ( \tilde w_{0,0} (\vec{p},z) - w_{0,0} (\vec{p},z) ) \|_{\infty} \leq \mathbf{D}^2 \mathbf{C} \rho^{2\varepsilon} \sin(\vartheta) , \label{n1_n2}
\end{align}
for all $( \vec{p} , z ) \in U_{ \tilde \rho }[ \tilde f ]$.
\end{lemma}
\vspace{1mm}

\begin{remark}
  In our inductive construction, the constant $\mathbf{D}$ in Lemma \ref{borne0_n1} is replaced at  step $j$ by  $\mathbf{C}^{j}$, where $\mathbf{C}$ is the numerical constant appearing in \eqref{n1_n1} and \eqref{n1_n2}. Moreover, $\rho$ is replaced by $\rho_{0}^{(2-\varepsilon)^{j}}$.   The hypothesis $\rho_{0}^{\varepsilon(2-\varepsilon)^{j}} \ll \mathbf{C}^{-j}$ is  fulfilled at any step $j$ if   $\rho_{0}$ is sufficiently  small. Furthermore, Equations \eqref{n1_n1} and \eqref{n1_n2}   imply that  \eqref{4.9} and \eqref{partialz_n1}  hold  true  at any step $j$ for sufficiently small values of $\rho_0$; see Paragraph \ref{induu}.
\end{remark}
\vspace{1mm}

\noindent The proof of Lemma \ref{borne0_n1} occupies three subsections. In Subsection \ref{Fesh1} we show  that  the pair    $$(H (\vec{p},z), W_{0,0} (\vec{p},z)  + \mathcal{E} (\vec{p},z))$$  is a Feshbach-Schur pair associated to $\chi_{\tilde \rho}(H_f)$, for all $(\vec{p},z) \in \mathcal{U}_{ \tilde \rho}[ \tilde f ]$. In Subsection \ref{kiwi}, we apply the Feshbach-Schur map to this Feshbach-Schur pair and re-Wick order the resulting operator  $\tilde{H}(\vec{p},z)$, so as to bring it into the ``normal form'' shown in Eq. \eqref{normalform}. We then verify that the sequence of kernels $\tilde{\underline{w}}$ has properties (a), (b), (c) and satisfy the estimates in \eqref{n1_n1}--\eqref{n1_n2}. Finally, in Subsection \ref{d_and_e}, we prove that properties (d) and (e) hold, too.

\subsubsection{Proof of applicability of the Feshbach-Schur map.} \label{Fesh1}

Let $( \vec{p} , z ) \in \mathcal{U}_{ \tilde \rho }[ \tilde f ]$. We have that
\begin{align*}
\vert w_{0,0} (\vec{p},z,r,\vec{l})  + \mathcal{E} ( \vec{p},z) \vert &\geq  \vert w_{0,0}  (\vec{p},z,r,\vec{l}) \vert -  \vert  \mathcal{E} (\vec{p},z) \vert \\
& \geq \vert  r -  \vec{p} \cdot \vec{l}  \vert - \frac{ \mu }{ 4 }r   - \mu \tilde \rho / 16  \\
& \geq r \left( \mu - \frac{ \mu }{ 4 }  \right) - \mu \tilde \rho / 16 \\
&\ge \frac{ \mu}{2} \tilde {\rho},
\end{align*}
for all $(\vec{p},z) \in \mathcal{U}_{\tilde \rho}[ \tilde f ]$ and $r \ge 3 \tilde \rho /4$. Therefore,   the restriction of  $W_{0,0} (\vec{p},z)  + \mathcal{E}(\vec{p},z)$ to $\text{Ran}(\overline{\chi}_{ \tilde \rho }(H_f))$ is bounded-invertible and
\begin{align*}
&\|    [W_{0,0}(\vec{p},z) + \mathcal{E}(\vec{p},z)]^{-1}  \overline{\chi}_{ \tilde \rho}(H_f) \|  = \mathcal{O} \Big ( \frac{ 1 }{\mu \tilde \rho } \Big ).
\end{align*} 

Next, we prove that  the restriction of  $H_{\overline{\chi}_{\tilde{\rho}}(H_f)}(\vec{p},z)$   to $\text{Ran}(\overline{\chi}_{ \tilde \rho }(H_f))$ is bounded-invertible.  It follows from \eqref{Bb00_n1} and \eqref{eq:a1} that
\begin{equation}
\label{uneq_n1}
\|W_{m,n} (\vec{p},z) \| \leq \rho^{2(m+n)} \|w_{m,n} (\vec{p},z) \|_{\frac{1}{2}} (8 \pi)^{\frac{m+n}{2}} \le \rho^{2(m+n)} (8 \pi)^{\frac{m+n}{2}} \rho^{- (m+n) + 1}  \mathbf{D}^{m+n} \mu.
\end{equation}
Summing \eqref{uneq_n1} over $m+n \geq 1$, we find that
\begin{equation}\label{4.27}
\|W_{\geq 1} (\vec{p},z) \| = \mathcal{O} \left( \mathbf{D}  \rho^{2} \mu \right) ,
\end{equation}
which yields
\begin{align*}
&\|    [W_{0,0}(\vec{p},z) + \mathcal{E} (\vec{p},z)]^{-1}  \overline{\chi}_{\tilde \rho}(H_f) W_{\geq 1}(\vec{p},z) \|  = \mathcal{O} \left( \mathbf{D} \rho^\varepsilon \right ).
\end{align*} 
Since $\rho^{\varepsilon} \ll \mathbf{D}^{-1}$, we deduce that $(H(\vec{p},z), W_{0,0}(\vec{p},z)  + \mathcal{E}(\vec{p},z))$ is a Feshbach-Schur pair associated to $\chi_{ \tilde \rho }(H_f)$, for all  $(\vec{p},z) \in \mathcal{U}_{\tilde \rho} [z( \vec{p} )]$, and that  $\tilde H (\vec{p},z)$ is well-defined.

\subsubsection{Re-Wick ordering  and proof of the inductive properties (a), (b) and (c)} \label{kiwi}
We temporarily omit the argument $(\vec{p},z)$. To shorten our notation, we introduce the operators
\begin{equation}
H_{\overline{\chi}_{\tilde{\rho}}}:= W_{0,0}  + \mathcal{E} +  \overline{\chi}_{ \tilde \rho }(H_f) W_{\geq 1} \overline{\chi}_{ \tilde \rho }(H_f)  
\end{equation}
and 
\begin{equation}
 R( H_f , \vec{P}_f ) := \frac{ \overline{\chi}^{2}_{ \tilde \rho }(H_f)}{W_{0,0} (H_f,\vec{P}_f) + \mathcal{E} } .
\end{equation}
 We have  that 
\begin{equation}
\label{3.24'_n1}
\begin{split}
\tilde H &=(W_{0,0} + \mathcal{E} ) \mathds{1}_{H_f \leq \tilde \rho } + \chi_{ \tilde \rho }(H_f) W_{\geq 1} \chi_{ \tilde \rho }(H_f) + \tilde{V} ,
\end{split}
\end{equation}
where
\begin{equation*}
 \tilde{V} = - \chi_{ \tilde \rho }(H_f) W_{\geq 1} \overline{\chi}_{ \tilde \rho  }(H_f) \left[ H_{\overline{\chi}_{\tilde{\rho}}}\right]^{-1}_{\vert \text{Ran}( \overline{\chi}_{ \tilde \rho }(H_f) ) } \overline{\chi}_{ \tilde \rho }(H_f)  W_{\geq 1} \chi_{ \tilde \rho }(H_f). 
\end{equation*}
For any $L \in \mathbb{N}$ and any $M_1,...,M_L \in \mathbb{N}$, we define  $\underline{M} := (M_1,...,M_L)$. The Neumann expansion for $\tilde{V}$ reads
\begin{equation} 
 \tilde{V} = - \sum_{L=2}^{\infty} (-1)^{L}   \underset{ \underline{M},\underline{N} ; \text{ }M_i + N_i \geq 1 }{\sum} \tilde{V}_{\underline{M},\underline{N}},
\end{equation} 
where
\begin{equation}
\label{op1_n1}
\begin{split}
\tilde{V}_{\underline{M},\underline{N}}:=  \chi_{ \tilde \rho }(H_f)     \prod_{i=1}^{L-1} \left[   W_{M_i,N_i}   R(H_f,\vec{P}_f) \right]   W_{M_L,N_L}  \chi_{ \tilde \rho }(H_f).
\end{split}
\end{equation} 
To normal order $\tilde{V}_{\underline{M},\underline{N}}$, we pick out $m_1/n_1$ creation/annihilation operators from the $M_1/N_1$ creation/annihilation operators available in $ W_{M_1,N_1}$,..., $m_L/n_L$ creation/annihilation operators from the $M_L/N_L$ creation/annihilation operators available in $ W_{M_L,N_L}$, and contract all the remaining annihilation and creation operators. As the monomials $W_{M_i,N_i}$'s are symmetric in $\underline{k}_1,...,\underline{k}_{M_i}$, and in $\underline{\tilde{k}}_1,...,\underline{\tilde{k}}_{N_i}$, there are 
\begin{equation*}
C_{\underline{m},\underline{n} }^{\underline{M}, \underline{N}} :=\prod_{i=1}^{L}  \binom{M_i}{m_i} \binom{N_i}{n_i} 
\end{equation*}
contraction schemes giving rise to the same contribution. We then pull through all the remaining  $m_1 +....+ m_L$ uncontracted creation operators to the left, and the remaining $n_1 +....+ n_L$ uncontracted annihilation operators to the right. This causes a shift in the arguments of the operators $w_{M_i,N_i} (H_f,\vec{P}_f, \underline{K}_{i}^{(M_i,N_i)})$ and $R(H_f,\vec{P}_f)$ via the pull-through formula given in \eqref{pull}.
 The contracted part is expressed as a vacuum expectation value.  $\tilde{V}_{\underline{M},\underline{N}}$ can be rewritten in the form 
\vspace{1mm}
\begin{equation*}
\tilde{V}_{\underline{M},\underline{N}} = \underset{ \underline{m},\underline{n}; \text{ }  m_i \leq M_i, n_i \leq N_i}{\sum}  C_{\underline{m},\underline{n} }^{\underline{M}, \underline{N}} \int d \underline{K}^{(\underline{m},\underline{n})}  \text{ }a^{*}(\underline{k}^{(\underline{m})}) \langle \Omega \vert  V_{\underline{m},\underline{n} }^{\underline{M}, \underline{N}} (r,\vec{l}, \underline{K}^{(\underline{m},\underline{n})})   \Omega \rangle_{r=H_f, \vec{l}=\vec{P}_f} \text{ } a(\underline{\tilde{k}}^{(\underline{n})}), 
\end{equation*}
\normalsize
where  $a^{*}(\underline{k}^{(\underline{m})})=\prod_{i=1}^{L} a^*(\underline{k}_{i}^{(m_i)})$, see  Paragraph \ref{Par22}.  Furthermore,  if $m_i$ (or $n_i$) is equal to zero, $a^*(\underline{k}_{i}^{(m_i)})$ (or $a(\underline{\tilde{k}}_{i}^{(n_i)})$)  is  replaced by  $1$ in the above formula.   A precise expression for the operator $ V_{\underline{m},\underline{n} }^{\underline{M}, \underline{N}} (r,\vec{l}, \underline{K}^{(\underline{m},\underline{n})})$ is given in \cite{FFS}. The reader should notice that this precise expression is not necessary to pursue our analysis.  What matters is that  $ V_{\underline{m},\underline{n} }^{\underline{M}, \underline{N}} (r,\vec{l}, \underline{K}^{(\underline{m},\underline{n})})$  is a product of $L-1$ operators $R(H_f+r,\vec{P}_f+ \vec{l})$ with shifted arguments,  together with the truncated kernels $ {W}_{m_i,n_i}^{M_i,N_i}(r,\vec{l},\underline{k}^{(m_i)}_i,\underline{\tilde{k}}^{(n_i)}_i)$ with shifted arguments, where
\begin{equation}
\label{quadra}
\begin{split}
 {W}_{m_i,n_i}^{M_i,N_i}(r,\vec{l}, & \underline{k}^{(m_i)}_i,\underline{\tilde{k}}^{(n_i)}_i)= \mathds{1}_{H_f \leq \rho} \text{ } \int d \underline{X}^{(M_i-m_i,N_i-n_i)}_i   \text{ } a^{*} (\underline{x}^{(M_i-m_i)}_i)\\
& \qquad  w_{M_i,N_i}  \text{ }  (H_f + r,\vec{P}_f + \vec{l}, \underline{K}_{i}^{(m_i,n_i)}, \underline{X}^{(M_i-m_i,N_i-n_i)}_i)   \text{ } a (\underline{\tilde{x}}^{(N_i-n_i)}_i)  \text{ }  \mathds{1}_{H_f \leq \rho}
\end{split}
\end{equation} 
and
\begin{equation}
\label{forms_n1}
\begin{split}
& V_{\underline{m}, \underline{n} }^{\underline{M}, \underline{N}}  (r,\vec{l}, \underline{K}^{(\underline{m},\underline{n})}) = \chi_{ \tilde \rho }(r + \tilde{r}_0) \prod_{i=1}^{L-1} \Big[W_{m_i,n_i}^{M_i,N_i}(r+r_i,\vec{l}+\vec{l}_i, \underline{k}^{(m_i)}_i,\underline{\tilde{k}}^{(n_i)}_i) \\
& \text{ }R(H_f + r + \tilde{r}_i, \vec{P}_f+\vec{l} + \vec{\tilde{l}}_i) \Big]   W_{m_Ln_L}^{M_L,N_L}(r+r_L,\vec{l}+\vec{l}_L, \underline{k}^{(m_L)}_L,\underline{\tilde{k}}^{(n_L)}_L) \chi_{ \tilde \rho }(r + \tilde{r}_L).
\end{split}
\end{equation}
\vspace{1mm}

\noindent The terms $r_i$'s, $\tilde{r}_i$'s, $\vec{l}_i$'s, $\vec{\tilde{l}}_i$'s are the shifts that come from the pull-through formula, see \eqref{def_mu_L}. Therefore, we deduce that   there exists a sequence of kernels $\tilde w_{M,N} : \mathcal{U}_{\tilde \rho} [\tilde{f} ] \times \mathcal{B}_{ \tilde \rho } \times \underline{B}_{ \tilde \rho }^{(m,n)} \rightarrow \mathbb{C}$ and a function $\tilde{\mathcal{E}} : \mathcal{U}_{\tilde \rho} [\tilde{f}] \rightarrow \mathbb{C}$  such that 
\begin{equation}
\label{dent_n1}
\tilde H (\vec{p},z)= H[\tilde{\underline{w}} (\vec{p},z), \tilde{\mathcal{E}} (\vec{p},z)]=\sum_{M+N \geq 0} \tilde W_{M,N}(\vec{p},z) + \tilde{\mathcal{E}}(\vec{p},z),
\end{equation}
where we have introduced the arguments $(\vec{p},z)$ again in \eqref{dent_n1}.  The associated kernels are given by 
\begin{align}
\tilde w_{M,N} (\vec{p},z,r,\vec{l}, \underline{k}^{(M)},\underline{\tilde{k}}^{(N)})&=\sum_{L=1} (-1)^{L+1} \underset{ \tiny \begin{array}{c} m_1+...+m_L =M \\ n_1+...+n_L =N \end{array}}{\sum}  \underset{ \tiny \begin{array}{c} p_1,...,p_L \\q_1,...,q_L  \\  m _i+n_i+p_i+q_i \geq 1\end{array}}{\sum} C^{\underline{m+p},\underline{n+q}}_{\underline{m},\underline{n}} \notag \\
&  \qquad \langle \Omega \vert  V_{\underline{m},\underline{n} }^{\underline{m+p}, \underline{n+q}} (\vec{p},z,r,\vec{l}, \underline{K}^{(M,N)}) \Omega \rangle_{\text{sym}}, \label{noyau1_n1}
\end{align}
 for $M+N \geq 1$, where $\underline{m+p}=(m_1+p_1,...,m_L+p_L)$, and $f_{\text{sym}}$ denotes the symmetrization of $f$ with respect to the variables $\underline{k}^{(M)}$ and $\underline{\tilde{k}}^{(N)}$.  Notice that we have set $p_i:=M_i-m_i$, $q_i:=N_i-n_i$, $i=1,...,L$.  Also note that we started the sum in \eqref{noyau1_n1} at $L=1$, and not $L=2$, since we included the contribution of the term $\chi_{ \tilde \rho }(H_f) W_{\geq 1}(\vec{p},z) \chi_{ \tilde \rho }(H_f) $ in \eqref{3.24'_n1}. Similarly,  
\begin{equation}
\label{noyau0_n1}
\begin{split} 
\tilde w_{0,0}(\vec{p},z,r,\vec{l}) + & \tilde{\mathcal{E}} (\vec{p},z) = \Big [w_{0,0} (\vec{p},z,r,\vec{l}) + \mathcal{E} (\vec{p},z)  \\
& + \sum_{L=2}^{\infty} (-1)^{L-1} \underset{ \tiny \begin{array}{c} p_1,...,p_L \\q_1,...,q_L  \\  p_i+q_i \geq 1\end{array}}{\sum} \langle \Omega \vert V_{\underline{0},\underline{0}}^{\underline{p},\underline{q}}(r,\vec{l})   \Omega \rangle \Big] \mathds{1}_{r \leq  \tilde \rho }.
\end{split}
\end{equation}
We now bound \eqref{noyau1_n1} and \eqref{noyau0_n1}. Thanks to \eqref{eq:a1},
\begin{equation}
\label{BW_n1}
\|W^{M_i,N_i}_{m_i,n_i}(r,\vec{l},\underline{k}^{(m_i)}_{i}, \underline{\tilde{k}}^{(n_i)}_{i}) \| \leq \vert \underline{k}^{(m_i)}_{i} \vert^{1/2} \text{ }  \vert \underline{\tilde{k}}^{(n_i)}_{i} \vert^{1/2}  \|w_{M_i,N_i} \|_{\frac{1}{2}} \text{ } (\sqrt{8 \pi} \rho^2 )^{M_i+N_i -m_i-n_i}.
\end{equation}
The same bounds  are satisfied by  the partial derivatives with respect to $r$ and $l_q$ with $ \|w_{M_i,N_i} \|_{\frac{1}{2}} $ replaced by $ \| \partial_{ \# } w_{M_i,N_i} \|_{\frac{1}{2}} $. Since

\begin{equation*}
\vert w_{0,0} (\vec{p},z,r,\vec{l})  + \mathcal{E} (\vec{p},z) \vert  \geq  \frac{ \mu}{2} \tilde \rho , 
\end{equation*}
 there exists a numerical  constant $C >0$ such  that
\begin{align}
\label{BFF}
\underset{(r,\vec{l}), r \ge \frac{3}{4}  \tilde \rho ,\vert \vec{l} \vert \leq r  } {\sup} \vert R(\vec{p},z,r,\vec{l}) \vert &\leq \frac{C}{ \tilde \rho \mu} , \\
 \underset{(r,\vec{l}), r \ge \frac{3}{4} \tilde \rho ,\vert \vec{l} \vert \leq r } {\sup} \vert  \partial_{\sharp} R(\vec{p},z,r,\vec{l}) \vert  & \leq  \frac{C}{ \tilde \rho^{2}   \mu^2}.
\end{align}
Suppressing the argument $(\vec{p},z,r,\vec{l}, \underline{K}^{(\underline{m},\underline{n})})$ on the left side to shorten our  formulas, it follows immediately from  \eqref{BW_n1} and \eqref{forms_n1} that
\begin{equation*}
\|   V_{\underline{m},\underline{n} }^{\underline{m+p}, \underline{n+q}} \| \leq \frac{C^L}{( \mu \tilde \rho )^{L-1}}  \prod_{i=1}^{L} \left[   \vert \underline{k}^{(m_i)}_i \vert^{1/2} \vert \tilde{\underline{k}}^{(n_i)}_i \vert^{1/2}    \| w_{m_i+p_i,n_i+q_i} (\vec{p},z) \|_{\frac{1}{2}} (\sqrt{8 \pi} \rho^2 )^{p_i+q_i} \right],
\end{equation*} 
and,
\begin{equation*}
\begin{split}
\|   \partial_{r } V_{\underline{m},\underline{n} }^{\underline{m+p}, \underline{n+q}}  & \| \leq L  \frac{C^L}{(  \mu \tilde \rho )^{L-1}}  \sup_{j} \prod_{i=1}^{L}  \Big[   \vert \underline{k}^{(m_i)}_i \vert^{1/2} \vert \tilde{\underline{k}}^{(n_i)}_i \vert^{1/2}  \Big( (1-\delta_{ij}) \| w_{m_i+p_i,n_i+q_i} (\vec{p},z) \|_{\frac{1}{2} }  \\
& + \delta_{ij}   \| \partial_{r}  w_{m_i+p_i,n_i+q_i} (\vec{p},z) \|_{\frac{1}{2}} \Big) (\sqrt{8 \pi}\rho^2)^{p_i+q_i} \Big]\\
&+ (L+1) \frac{C^L}{( \mu \tilde \rho  )^{L}}  \prod_{i=1}^{L} \left[   \vert \underline{k}^{(m_i)}_i \vert^{1/2} \vert \tilde{\underline{k}}^{(n_i)}_i \vert^{1/2}  \| w_{m_i+p_i,n_i+q_i} (\vec{p},z) \|_{\frac{1}{2}} (\sqrt{8 \pi}\rho^2)^{p_i+q_i} \right]
\end{split}
\end{equation*}

\noindent  \noindent Using  the inequality $$ \binom{m_i+p_i}{m_i}  \leq 2^{m_i+p_i},$$ we deduce that
 \begin{equation*}
\begin{split}
\| \tilde w_{M,N} (\vec{p},z) \|_{\frac{1}{2} } \leq  \sum_{L=1}^{\infty}   & \frac{ 4^{M+N} C^L}{ (\mu \tilde \rho )^{L-1}}  \underset{ \underset{m_i+n_i+p_i +q_i \geq 1 }{\underline{m}, \underline{n}, \underline{p}, \underline{q}} }{\sum} \prod_{i=1}^{L} \Big( \|w_{m_i+p_i,n_i+q_i}(\vec{p},z) \|_{\frac{1}{2}}  \\
& (2 \sqrt{8 \pi} \rho^2 )^{p_i+q_i}  \left( \frac{1}{2} \right)^{m_i+n_i}\Big).
\end{split}
\end{equation*}
\normalsize
Inequality \eqref{Bb00_n1} for the norms of the kernels implies that
\begin{align*}
 \|w_{m_i+p_i,n_i+q_i}(\vec{p},z) \|_{\frac{1}{2}}  & \leq \rho^{- (m_i+n_i+p_i+q_i) + 1} \mathbf{D}^{m_i+n_i+p_i+q_i} \mu  \sin(\vartheta)  \\
 & \le  \big( \frac{1}{32 C} \tilde \rho \big)^{- ( m_i+n_i+p_i+q_i ) + 1 }    \mathbf{D}^{m_i+n_i+p_i+q_i}  \mu  \sin(\vartheta),
 \end{align*}
where we have used that $m_i+n_i+p_i+q_i \geq 1$. It then follows that 
 \begin{align*}
\| \tilde w_{M,N} (\vec{p},z) \|_{\frac{1}{2} } &\leq  \Big( \frac{128 C  \mathbf{D}}  {\tilde \rho } \Big)^{M+N}  \sum_{L=1}^{\infty}   \frac{  (   \mu   \sin(\vartheta)   \tilde \rho /32 )^{L}}{ ( \mu \tilde \rho )^{L-1}}  \left( \sum_{p+q \geq 0}  \Big( \frac{ 64 C  \sqrt{8 \pi} \rho^{2}  \mathbf{D}}{ \tilde \rho }\Big)^{p+q}  \sum_{m+n \geq 0} \frac{1}{2^{m+n}}\right)^L\\
&\leq   \Big( \frac{128 C  \mathbf{D}}  {\tilde \rho } \Big)^{M+N}   \mu \tilde \rho   \sin(\vartheta) ,
 \end{align*}
where we have used that $\rho^{ \varepsilon} <  ( 128 C \sqrt{8 \pi}  \mathbf{D})^{-1}$, in order to sum over $p+q$  on the right-hand side. A similar bound is satisfied by the partial derivatives with respect to $r$ and $l_q$, $q=1,...,3$, with $\mu \tilde \rho \sin(\vartheta)$ replaced by $1$.  This shows that $\tilde{w}_{M,N}$ satisfies property $(c)$, with $\mathbf{D}$ replaced by $\mathbf{CD}$ provided that $\mathbf{C}$ is chosen large enough.  It is important to remark that  the constant $\mathbf{C}>0$ can be chosen to be independent of the 'problem parameters' and  of  the constant $\mathbf{D}$. This property of $\mathbf{C}$ is needed to implement our inductive construction.
 
Likewise, 
\begin{align*}
\| \tilde w_{0,0} (\vec{p},z) + & \tilde{ \mathcal{E} } (\vec{p},z) - w_{0,0} (\vec{p},z) -\mathcal{E} (\vec{p},z) \|_{\infty}  \\[4pt]
&\leq \sum_{L=2}^{\infty}  \frac{ (  C   \sin(\vartheta)   \rho \mu  )^{L}}{ ( \tilde \rho \mu  )^{L-1}}  \left( \sum_{p+q \geq 1} \left[ 2   \mathbf{D} \sqrt{8 \pi} \rho  \right]^{p+q} \right)^{L}\\[4pt]
& \leq \tilde{C} \mathbf{D}^2 \rho^4 \tilde \rho^{-1} \mu  \sin^2(\vartheta)= \tilde{C} \mathbf{D}^2 \rho^{2+\varepsilon} \mu \sin^2(\vartheta),
\end{align*}
and a similar bound is satisfied for $ \partial_\# ( \tilde w_{0,0} (\vec{p},z) - w_{0,0} (\vec{p},z) )$, which proves \eqref{n1_n1}--\eqref{n1_n2}, provided that $\mathbf{C}$ is chosen large enough.

\subsubsection{Properties (d) and (e)}\label{d_and_e}

We first observe that the map $(\vec{p},z) \mapsto \tilde{H}(\vec{p},z)$ is analytic on $\mathcal{U}_{\tilde{\rho}}[\tilde{f}]$. This follows from the facts that  the smooth Feshbach-Schur map preserves analyticity in $(\vec{p},z)$ (because the Neumann series converges uniformly on $\mathcal{U}_{\tilde{\rho}}[\tilde{f}]$) and that the maps $(\vec{p},z) \mapsto W_{\geq 1 }(\vec{p},z)  \in \mathcal{L}(\mathcal{H}_f)$ and  $(\vec{p},z) \mapsto W_{0,0 }(\vec{p},z)  + \mathcal{E}(\vec{p},z)$ are analytic on $\mathcal{U}_{\tilde{\rho}}[\tilde{f}]$ (because, by assumption, $ (\vec{p},z) \mapsto H(\vec{p},z)$  is analytic on the open set  $ \mathcal{U}_{\rho}[f] \supset\mathcal{U}_{\tilde{\rho}}[\tilde{f}]$). We refer the reader to \cite{FFS} for more details. Therefore $(\vec{p},z) \mapsto \tilde{\mathcal{E}} (\vec{p},z)  = \langle \Omega \vert  \tilde{H}(\vec{p},z)  \Omega \rangle$ is analytic in $(\vec{p},z)$, too.
  
We now prove property (e). Let $\vec{p} \in U_{\theta}[\vec{p}^*]$ and $z = \tilde{f}(\vec{p}) + \beta r_{\tilde{\rho}}  \in D( \tilde{f}(\vec{p}),r_{\tilde{\rho}})$, with $0 \leq \vert  \beta  \vert \leq 2/3$. The triangle inequality, the inequality \eqref{4.47_n1}, and the hypothesis $\rho^{1-\varepsilon} \ll 1$ imply that $z \in  \overline{D}(f(\vec{p}), 2r _{\rho}/3)$.  We consider the  circular  contour $\mathcal{C}$ centered at $\tilde{f}(\vec{p})$ and with radius $r_{\mathcal{C}}= 3 r_{\tilde{\rho}}/4$. We have that  $ \{ \vec{p}\} \times \mathcal{C}  \subset \mathcal{U}_{\tilde{\rho}}[\tilde{f}] \subset   \mathcal{U}_{\rho}[f] $, and, by Cauchy's formula,
\begin{equation*}
\begin{split}
\vert \partial_z \tilde{\mathcal{E}}(\vec{p},z) - \partial_z \mathcal{E}(\vec{p},z) \vert& \leq \frac{1}{2\pi} \Big \vert  \int_{\mathcal{C}} dz' \frac{\tilde{\mathcal{E}}(\vec{p},z') -\mathcal{E}(\vec{p},z') }{ (z-z')^2} \Big \vert \\
& \leq \frac{3 }{ 4 (3/4-\vert\beta\vert)^2 } \frac{ \mathbf{C} \mathbf{D}^{2}  \mu \sin^2(\vartheta)  \rho^{2+\varepsilon}}{r_{\tilde{\rho}} } ,
\end{split}
\end{equation*}
where we have used \eqref{n1_n1} and the fact that $w_{0,0}(\vec{p},z,0,\vec{0}) = 0 = \tilde w_{0,0}(\vec{p},z,0,\vec{0})$, see property (a). It follows that 
\begin{equation} \label{lmim}
\vert \partial_z \tilde{\mathcal{E}}(\vec{p},z) - \partial_z \mathcal{E}(\vec{p},z) \vert  \leq    \frac{96 }{ 4 (3/4-\vert\beta\vert)^2 }    \mathbf{C}    \mathbf{D}^{2} \rho^{2\varepsilon}.
\end{equation}
Since
\begin{align*}
\vert \partial_z \mathcal{E} (\vec{p},z) +1 \vert & \le 1/4,
\end{align*}
 for any $z \in \overline{D}(\tilde{f}(\vec{p}),\frac{2}{3} r_{\tilde{\rho}})$,  we deduce that 
 \begin{equation} \label{4.44_n1}
\vert \partial_z \tilde{\mathcal{E}}(\vec{p},z) + 1 \vert \le \frac{1}{2} ,
\end{equation}
for any $z \in \overline{D}(\tilde{f}(\vec{p}),\frac{2}{3} r_{\tilde{\rho}})$, where we use the hypothesis that 
$\rho^{\varepsilon} \ll \mathbf{D}^{-1}$.
We estimate $\vert   \tilde{\mathcal{E}}(\vec{p},z) -\tilde{f}(\vec{p}) +z \vert$ on the circle $\partial D(\tilde{f}(\vec{p}),\frac{2}{3}r_{\tilde{\rho}})$. Let $z \in  \overline{D}(\tilde{f}(\vec{p}),\frac{2}{3} r_{\tilde{\rho}})$.  We obtain from \eqref{4.44_n1} that
\begin{align}
\vert \tilde{\mathcal{E}}(\vec{p},z) - \tilde{f}(\vec{p})+z \vert& \le \vert \tilde{\mathcal{E}}(\vec{p},z) - \tilde{\mathcal{E}}(\vec{p},\tilde{f}(\vec{p})) +z- \tilde{f}(\vec{p})  \vert + | \tilde{\mathcal{E}}(\vec{p},\tilde{f}(\vec{p})) | \notag \\
& \le \frac{1}{2} \vert z - \tilde{f}(\vec{p}) \vert + \vert \tilde{\mathcal{E}}(\vec{p},\tilde{f}(\vec{p})) - \mathcal{E} (\vec{p},\tilde{f}(\vec{p})) \vert \notag \\
& \le  \frac{1}{2} \vert z - \tilde{f}(\vec{p}) \vert +  \mathbf{C}   \mathbf{D}^{2} \mu \sin^2(\vartheta) \rho^{2+ \varepsilon} , \label{eq:estim_n1}
\end{align}
where we have used the fact that $\mathcal{E} (\vec{p},\tilde{f}(\vec{p})) = 0$, along with \eqref{n1_n1}. If $\rho$ is sufficiently small then
\begin{equation*}
\mathbf{C}  \mathbf{D}^{2}   \sin^2(\vartheta) \rho^{2 + \varepsilon} \mu  < \frac{r_{\tilde{\rho}}}{3} ,
 \end{equation*} 
and hence
\begin{equation}
\vert \tilde{\mathcal{E}}(\vec{p},z) - \tilde{f}(\vec{p})+z \vert <  \vert z-  \tilde{f}(\vec{p}) \vert,
\end{equation}
for any  $z \in \partial D(\tilde{f}(\vec{p}),\frac{2}{3} r_{\tilde{\rho}})$.  Rouch\'e's theorem then implies that 
$z \mapsto  \tilde{\mathcal{E}}(\vec{p},z)$ has a unique zero, $\tilde{\tilde{f}}(\vec{p})$, inside the disk of radius $2r_{\tilde{\rho}}/3$. 

To prove \eqref{4.47_n1}, we observe that \eqref{eq:estim_n1}, with $z = \tilde{\tilde{f}}( \vec{p} )$, yields
\begin{align*}
\frac{1}{2}  \vert \tilde{ \tilde{f} }(\vec{p})  -\tilde{f}(\vec{p}) \vert  & \le  \mathbf{C} \mathbf{D}^{2}  \mu \sin^2(\vartheta) \rho^{2+ \varepsilon} = 32   \mathbf{C} \mathbf{D}^{2}   \sin( \vartheta ) \rho^{2\varepsilon} r_{ \tilde \rho } .
\end{align*}
The hypothesis that $\rho^{\varepsilon} \ll \mathbf{D}^{-1}$ implies that
\begin{equation}
  \vert \tilde{\tilde{f}}(\vec{p})  -\tilde{f}(\vec{p}) \vert \le \frac{r_{\tilde{\rho}}}{2}.
\end{equation}

Next, we verify that $\tilde{\mathcal{E}}$  satisfies \eqref{eq:eps_n1} with $\tilde{\rho}$ replaced by  $\tilde{\tilde{\rho}} := \tilde{\rho}^{2-\varepsilon}$. The triangle inequality directly implies that  $\mathcal{U}_{\tilde{\tilde{\rho}}}[\tilde{\tilde{f}}] \subset \mathcal{U}_{\tilde{\rho}}[\tilde{f}]$, and we have that
\begin{align}
\vert  \tilde{\mathcal{E}}(\vec{p},z)  \vert  & \le \vert  \tilde{\mathcal{E}}(\vec{p},z)   + z -\tilde{\tilde f}(\vec{p}) \vert +   \vert z  -  \tilde{ \tilde f}(\vec{p})  \vert \notag \\
& \le \frac{3}{2} \vert z-\tilde{ \tilde f}(\vec{p}) \vert \leq 2 r_{\tilde{\tilde \rho}}   \le \frac{1}{16}  \tilde{ \tilde \rho} \mu , \label{4.47}
\end{align}
for any $z \in  D(\tilde{ \tilde f}(\vec{p}),r_{\tilde{ \tilde \rho}}) \subset \overline{D}(\tilde{f}(\vec{p}),2r_{\tilde{\rho}}/3)$, where \eqref{4.44_n1} has been used in the second inequality.
\vspace{2mm}

To complete our proof, we show that the map $\vec{p} \mapsto \tilde{ \tilde f}(\vec{p})$ is analytic on $U_{\theta}[\vec{p}^*]$ and that the set $\mathcal{U}_{\tilde{ \tilde \rho}}[\tilde{ \tilde f}]$ is open. Let  $ (\vec{p}_0,\tilde{ \tilde f}(\vec{p}_0)) \in \mathcal{U}_{\tilde{ \tilde \rho}}[\tilde{ \tilde f}] \subset \mathcal{U}_{ \tilde \rho}[\tilde{  f}]$. Since $\vert \partial_{z} \tilde{\mathcal{E}}(\vec{p}_0,z) + 1 \vert  <1/2$, for all $z \in  D(\tilde{f}(\vec{p}_0),2 r_{\tilde{\rho}}/3)$, the inverse function theorem implies that the map
\begin{equation}
(\vec{p},z) \mapsto (\vec{p},\tilde{\mathcal{E}}(\vec{p},z))
\end{equation}
is  biholomorphic in a small polydisk $\mathcal{D}_0 \subset  \mathcal{U}_{\tilde{\rho}}[\tilde{f}]$ around  $(\vec{p}_0,\tilde{\tilde f}(\vec{p}_0))$.  We denote  its  inverse by $h$. The image of  $\mathcal{D}_0$, denoted $\mathcal{D}_1$, contains the point $(\vec{p}_0,0)$, because 
$\tilde{\mathcal{E}}(\vec{p}_0,\tilde{\tilde f}(\vec{p}_0))=0$.  As $\mathcal{D}_1$ is open, it  contains $(\vec{p},0)$, for any $\vec{p}$ sufficiently close to $\vec{p}_0$. Therefore, $h  (\vec{p},0)$   coincides  with $(\vec{p},\tilde{ \tilde f} ( \vec{p}))$, for $\vec{p}$ near $\vec{p}_0$, and we deduce that $\vec{p} \mapsto \tilde{\tilde f}(\vec{p})$ is holomorphic in a neighborhood of $\vec{p}_0$. By letting $\vec{p}_0$ vary one sees that this implies that $\vec{p} \mapsto \tilde{\tilde f}(\vec{p})$ is holomorphic on $U_{\theta}[\vec{p}^*]$.  In particular  the set 
\begin{equation}
\label{UU2}
\mathcal{U}_{\tilde{ \tilde \rho}}[\tilde{ \tilde f}]:= \left\{  (\vec{p},z) \in U_{\theta}[\vec{p}^*] \times  \mathbb{C} \mid  z \in D(\tilde{ \tilde f}(\vec{p}), r_{\tilde{ \tilde \rho}}) \right\}. 
\end{equation}
is an open subset of $\mathbb{C}^4$.

\subsection{Inductive construction of the operators $H^{(j)}(\vec{p},z)$}
\label{induu}
Let $0 < \varepsilon< 1$ and let $\rho_0$ be such that $0< \rho_0  \ll  \min(1, \delta_0)  \mu $. Let  $\gamma < \rho_0$, where  $\gamma$ is the constant that appears in the estimates \eqref{Bb00}--\eqref{Bb000}. For $\lambda_0$ and $\theta = i \vartheta$ as in Lemma \ref{borne0},  we deduce from Section \ref{firstde} that the sequence of kernels $ w^{(0)}_{m,n}$ and the function $\mathcal{E}^{(0)}$  corresponding to the operator $H^{(0)}(\vec{p},z)$  satisfy properties  (a)--(e) and \eqref{4.9}--\eqref{partialz_n1}, with $\mathbf{D}=1$, $f(\vec{p})=E_{i_0}$, $\tilde{f}(\vec{p})=z^{(0)}(\vec{p})$, and $\rho=\rho_0$. 
We consider the  sequences $(\rho_{j})_{j \in \mathbb{N}_0}$ and $(r_{j})_{j \in \mathbb{N}_0}$ introduced in \eqref{sequence}.
By repeated application of Lemma \ref{borne0_n1}, we construct a sequence of effective Hamiltonians $H^{(j)}( \vec{p} , z )$ at scale $\rho_j$ by setting
\begin{equation}\label{eq:stepj}
 H^{(j+1)}(\vec{p},z)= F_{\chi_{   \rho_{j+1} }(H_f)} [H^{(j)} (\vec{p},z), W^{(j)}_{0,0}(\vec{p},z) + \mathcal{E}^{(j)}(\vec{p},z)]_{ \mathds{1}_{H_f \leq \rho_{j+1}}[\mathcal{H}_f] } .
\end{equation}
 In order to be able to apply Lemma \ref{borne0_n1} in each induction step,  we only need to show that the properties \eqref{4.9}--\eqref{partialz_n1} hold in each step of the induction and that 
 $\rho_{j}^{\varepsilon} \mathbf{C}^{j} \ll1$, for all $j \in \mathbb{N}_0$. This is accomplished in the following lemma.

  \begin{lemma} \label{induc}
Suppose that  \begin{equation}\label{rhoC}
 \sum_{k=1}^{\infty} \mathbf{C}^{2k-1} \rho_{0}^{  \varepsilon (2-\varepsilon)^{k-1}}  + \sqrt{3} \rho_0 \ll \mu \qquad \text{and} \qquad \rho_{0}^{1-\varepsilon} \ll1.
  \end{equation} 
  Let   $\gamma< \rho_0$ and choose  $\lambda_0$   and $\theta=i \vartheta$ as in Lemma \ref{borne0}. For any $j \in \mathbb{N}_0$, there exists a function $z^{(j-1)}: U_{\theta}[\vec{p}^*] \rightarrow \mathbb{C}$, a sequence of kernels $w_{m,n}^{(j)} : \mathcal{U}_{\rho_j}[z^{(j-1)}] \times \mathcal{B}_{\rho_j} \times \underline{B}_{\rho_j}^{(m,n)} \rightarrow \mathbb{C}$ and  a  function $\mathcal{E}^{(j)}:  \mathcal{U}_{\rho_j}[z^{(j-1)}] \rightarrow \mathbb{C}$ such that \eqref{eq:stepj} holds, and
\vspace{2mm}
\begin{enumerate}[(a)]
\item $w_{0,0}^{(j)}(\vec{p},z,\cdot,\cdot)$ is $\mathcal{C}^1$ on $\mathcal{B}_{\rho_j}$, and  $w_{0,0}^{(j)}(\vec{p},z,0,\vec{0})=0$,  for all  $ (\vec{p},z) \in \mathcal{U}_{\rho_j}[z^{(j-1)}]$.
\vspace{1mm}

\item  $w_{m,n}^{(j)} (\vec{p},z, \cdot, \cdot, \underline{K}^{(m,n)})$, $m+n \geq 1$,  are $\mathcal{C}^1$ on $\mathcal{B}_{\rho_j}$ for almost every $\underline{K}^{(m,n)} \in \underline{B}_{\rho_j}^{(m,n)}$ and  every $ (\vec{p},z) \in \mathcal{U}_{\rho_j}[z^{(j-1)}]$. Moreover, $w_{m,n}^{(j)}(\vec{p},z, \cdot, \cdot, \underline{K}^{(m,n)})$ is symmetric in $\underline{k}^{(m)}$ and $\underline{\tilde{k}}^{(n)}$.
\vspace{1mm}

\item For all $(\vec{p},z) \in \mathcal{U}_{\rho_j}[z^{(j-1)}]$, 
\begin{align}
\label{Bb00_n11}
\| w_{m,n}^{(j)} (\vec{p},z) \|_{\frac{1}{2}}& \le \mathbf{C}^{j(m+n)} \rho_{j}^{- (m+n) + 1} \mu \sin(\vartheta) ,\\[12pt]
\| \partial_{\#} w^{(j)}_{m,n}(\vec{p},z) \|_{\frac{1}{2}}& \le  \mathbf{C}^{j(m+n)}    \rho_{j}^{- (m+n) } ,
\end{align}
for all $m+n \ge 1$, where $\partial_\#$ stands for $\partial_r$ or $\partial_{l_p}$ and $\mathbf{C}$ is the positive constant appearing in Lemma \ref{borne0_n1}.
\vspace{1mm}

\item The maps $\mathcal{E}  :\mathcal{U}_{\rho_j}[z^{(j-1)}] \rightarrow \mathbb{C}$ and $(\vec{p},z) \mapsto H[\underline{w}^{(j)}(\vec{p},z),\mathcal{E}^{(j)}(\vec{p},z)]   \in \mathcal{L}(\mathds{1}_{H_f \leq \rho_j}\mathcal{H}_f)$ are analytic on  $\mathcal{U}_{\rho_j}[z^{(j-1)}]$.
\vspace{1mm}

\item For all $\vec{p} \in   U_{\theta}[\vec{p}^*] $, the holomorphic function $z \mapsto   \mathcal{E}^{(j)} (\vec{p},z) \in \mathbb{C}$ has a unique zero $  z^{(j)}(\vec{p}) \in  D( z^{(j-1)}( \vec{p} ) , 2 r_{\rho_j} / 3 )$.  The function $z^{(j)}$ is analytic on $U_{\theta}[\vec{p}^*]$. Moreover, $ \mathcal{U}_{  \rho_{j+1}}[z^{(j)}] \subset  \mathcal{U}_{\rho_j}[z^{(j-1)}]$, and
\begin{align}
\label{EJJ}
\vert \mathcal{E}^{(j)}(\vec{p},z)\vert &\le \frac{ \mu  \rho_{j+1} }{ 16 } ,\\
 \label{distz}
 \vert z^{(j)}(\vec{p})  -z^{(j-1)}(\vec{p}) \vert   &< r_{ j}/2,
\end{align}
for all $( \vec{p} , z ) \in\mathcal{U}_{\rho_{j+1}}[z^{(j)}]$ and $\vec{p} \in U_{\theta}[\vec{p}^*]$, respectively.
\end{enumerate}
 \end{lemma}

\begin{proof}
The hypothesis that   \begin{equation*}
 \sum_{k=1}^{\infty} \mathbf{C}^{2k-1} \rho_{0}^{  \varepsilon (2-\varepsilon)^{k-1}}  + \sqrt{3} \rho_0 \ll \mu
  \end{equation*}
  implies that  $\mathbf{C}^j \rho_{j}^{ \varepsilon} <c $, for a  positive numerical constant $c$, where $c$ is  small  enough such that Lemma \ref{borne0_n1}  holds   at any step $j$. As mentioned at the beginning of this subsection, the properties (a)--(e)  and the inequalities  \eqref{4.9} and \eqref{partialz_n1}  are   valid are step $j=0$.  Proceeding by induction, we assume that the properties (a)--(e)  and the inequalities  \eqref{4.9} and \eqref{partialz_n1} are valid at any  step $ k \leq j \in \mathbb{N}_0$. If we prove that 
\begin{align} \label{4.43}
\| \partial_r w^{(j)}_{0,0}(\vec{p},z) -e^{- \theta} \|_{\infty} + \sum_{q=1}^{3}  \|  \partial_{l_q} w^{(j)}_{0,0}(\vec{p},z) +  p_q  e^{ - \theta} \|_{\infty} & \le \frac{ \mu }{ 4 } , \quad  \forall (\vec{p},z) \in \mathcal{U}_{\rho_j}[z^{(j-1)}] ,
\end{align}
and \begin{equation} \big \vert \partial_z \mathcal{E}^{(j)}(\vec{p},z) + 1 \big \vert < \frac{1}{4},   \quad  \forall z \in \overline{D}(z^{(j-1)}(\vec{p}),\frac{2}{3} r_{j}), \label{4.44} \end{equation}
then Lemma \ref{borne0_n1} shows that 
$$ H^{(j+1)}(\vec{p},z)= F_{\chi_{   \rho_{j+1} }(H_f)} [H^{(j)} (\vec{p},z), W^{(j)}_{0,0}(\vec{p},z) + \mathcal{E}^{(j)}(\vec{p},z)]_{| \mathcal{H}^{(j+1)}}  $$
 is well-defined on $\mathcal{U}_{\rho_{j+1}}[z^{(j)}]$ and satisfies properties (a)-(e) at step $j+1$, and hence the induction step will be completed. 
 
The bound \eqref{4.43} is a direct consequence of estimate  \eqref{n1_n1} in Lemma \ref{borne0_n1}. Since properties (a)--(e) are valid for any $k \leq j$, we conclude that 
  $ \mathcal{U}_{\rho_j}[z^{(j-1)}]  \subset   \mathcal{U}_{\rho_{j-1}}[z^{(j-2)}] \subset  \text{ ... } \subset     \mathcal{U}_{\rho_{0}}[E_{i_0}] $.  Moreover,  \eqref{n1_n1}--\eqref{n1_n2}  are valid at any step $k =1,...,j$, with $\rho$ replaced by $\rho_{k-1}$ and $\mathbf{D}^2 \mathbf{C}$ by $\mathbf{C}^{2k-1}$. Therefore,
 \begin{align*}
 \| \partial_r w^{(j)}_{0,0}(\vec{p},z) -e^{- \theta} \|_{\infty}& \leq \sum_{k=1}^{j} \| \partial_r w^{(k)}_{0,0}(\vec{p},z) -\partial_r w^{(k-1)}_{0,0}(\vec{p},z) \|_{\infty} +  \| \partial_r w^{(0)}_{0,0}(\vec{p},z) -e^{- \theta} \|_{\infty}\\
 & \leq  \sum_{k=1}^{j}  \mathbf{C}^{2k-1} \rho_{k-1}^{2\varepsilon}  \sin(\vartheta)  + \gamma + \sqrt{3} \rho_0,
 \end{align*}
 for any $(\vec{p},z) \in  \mathcal{U}_{\rho_j}[z^{(j-1)}]$.   Our assumptions on $\rho_0$ and $\gamma$ imply  that  the  sum on the right side  is smaller than $\mu/4$. 
 
The bound \eqref{4.44} is proven by a similar argument:  For any $\vec{p} \in U_{\theta}[\vec{p}^*]$ and any $z \in \overline{D}(z^{(j-1)}(\vec{p}),\frac{2}{3} r_{j})$, we have that
\begin{align*}
 \big \vert \partial_z \mathcal{E}^{(j)}(\vec{p},z) + 1 \big \vert & \leq \sum_{k=1}^{j}  \big \vert \partial_z \mathcal{E}^{(k)}(\vec{p},z) - \partial_z \mathcal{E}^{(k-1)}(\vec{p},z)   \big \vert  +   \big \vert \partial_z \mathcal{E}^{(0)}(\vec{p},z) + 1 \big \vert\\
 & \leq C \left( \sum_{k=1}^{j}   \mathbf{C}^{2k - 1}   \rho_{k-1}^{2\varepsilon}   +    \gamma \right),
 \end{align*}
where $C$ is a positive numerical constant independent of the 'problem parameters'. The last inequality follows from \eqref{lmim} in the proof of Lemma \ref{borne0_n1}.  Also note that we have used the fact that  $D(z^{(j-1)}(\vec{p}),\frac{2}{3} r_{j}) \subset  D(z^{(j-2)}(\vec{p}),\frac{2}{3} r_{j-1})$, which follows from \eqref{distz}. Thus, the right side is smaller than $1/4$ if $\rho_0$ is sufficiently small. This completes the proof of the lemma.
 \end{proof}

 \section{Existence and analyticity of the resonances}
\label{conv}
In  Paragraph \ref{con1},  we prove that the sequence of functions $\vec{p} \mapsto z^{(j)}( \vec{p} )$ constructed in the previous subsection converges uniformly on $U_{\theta}[ \vec{p}^*]$. 
Identifying $\mathds{1}_{\mathbb{C}^{N}} \otimes Q_{\chi_{\rho_j}}$ with $Q_{\chi_{\rho_j}} $ for any $j \geq 1$ (see \eqref{Qchij}),   we show   in Subsection \ref{exi} that the sequence of vector valued functions   
\begin{equation}
\Psi^{(j)}(\vec{p}):=Q_{\bm{\chi}_{i_{0}}}  Q_{\chi_{\rho_1}}  ... Q_{\chi_{\rho_j}}   ( \psi_{i_{0}} \otimes \Omega)
\end{equation}
converges uniformly on $U_{\theta}[\vec{p}^*]$ to a non-vanishing function $\vec{p} \mapsto \Psi^{(\infty)}(\vec{p})$. 
 We remind the reader that $ \psi_{i_{0}}$ is a unit eigenvector of $H_{is}$ associated with the eigenvalue $E_{i_{0}}$ and that the operators $Q_\#$ are defined as in Theorem \ref{feshbach-theorem}  and Eqs. \eqref{Qchij}
 and \eqref{ulti}.

Here, $ \psi_{i_{0}}$ is a unit eigenvector of $H_{is}$ associated with the eigenvalue $E_{i_{0}}$ and the operators $Q_\#$ are defined as in Theorem \ref{feshbach-theorem}. For all $\vec{p} \in U_\theta( \vec{p}^* )$, $\Psi^{(\infty)}(\vec{p})$ is  an eigenvector of $H_{\theta}(\vec{p})$ associated with the eigenvalue $z^{(\infty)}(\vec{p})$.

\subsection{Convergence of $z^{(j)}$ and analyticity of the limit} \label{con1}
\begin{lemma}
Suppose that the parameters $\varepsilon$, $\gamma$, $\rho_0$, $\lambda_0$ and $\theta$ are fixed as in Section \ref{itere}. The sequence of  holomorphic functions $(z^{(j)})$ converges uniformly to an holomorphic function $z^{(\infty)}$  on $U_{\theta}[\vec{p}^*]$.  
\end{lemma}
\vspace{2mm}

\begin{proof}
The estimate
\begin{equation}
\vert z^{(j+1)}(\vec{p})-z^{(j)}(\vec{p}) \vert \le \frac{r_{j+1}}{2} = \frac{ \mu \rho_{j+1} \sin( \vartheta ) }{64} , \qquad   \forall \vec{p} \in U_{\theta}[\vec{p}^*] , \\
\end{equation}
obtained in Lemma \ref{induc}, implies that $(z^{(j)})$ is uniformly Cauchy. Hence $(z^{(j)})$ converges uniformly on $U_\theta( \vec{p}^* )$. Since $(z^{(j)})$ is analytic on $U_\theta( \vec{p}^* )$ for all $j$, by Lemma \ref{induc}, the uniform limit $z^{(\infty)}$ is also analytic.
\end{proof}
\vspace{1mm}

\subsection{Existence and analyticity of  the  eigenvector $ \Psi^{(\infty)}(\vec{p})$} \label{exi}
\begin{lemma}
Suppose that the parameters $\varepsilon$, $\gamma$, $\rho_0$, $\lambda_0$ and $\theta$ are fixed as in Section \ref{itere}. The sequence $(\Psi^{(j)})$ converges uniformly on $U_{\theta}[\vec{p}^*]$. The limit, $\Psi^{(\infty)}$, satisfies $\Psi^{(\infty)}(\vec{p}) \neq 0$ for all $\vec{p} \in U_{\theta}[\vec{p}^*]$. Furthermore, $(H_{\theta}(\vec{p}) -z^{(\infty)}(\vec{p})) \Psi^{(j)}(\vec{p})$ converges  to zero uniformly on $U_{\theta}[\vec{p}^*]$ and
\begin{equation*}
\mathrm{dim} \, \mathrm{Ker} ( H_{\theta}(\vec{p}) -z^{(\infty)}(\vec{p}) ) = 1.
\end{equation*}
\end{lemma}
\vspace{2mm}

\begin{proof}
We use the formula
\begin{equation}
H Q_{\chi}(H,T) = \chi F_{\chi}(H,T)
\end{equation}
for the Feshbach pair $(H,T)$, see \cite{BaChFrSi03_01}.  It implies that
\begin{equation}
\begin{split}
\label{atravers}
[H_{\theta}(\vec{p}) -z^{(\infty)}(\vec{p})] & Q_{\bm{\chi}_{i_{0}}} Q_{\chi_{\rho_1} }  ... Q_{\chi_{\rho_j}}   \\
&= \bm{\chi}_{i_{0}} \chi_{\rho_1}(H_f) ... \chi_{\rho_j}(H_f) F_{\chi_{\rho_j}(H_f)}[H^{(j-1)} ,W^{(j-1)}_{0,0}+ \mathcal{E}^{(j-1)}],
\end{split}
\end{equation}
where we omitted the argument $(\vec{p},z^{(\infty)}(\vec{p}))$ on the right-hand side of \eqref{atravers}. Applying \eqref{atravers} to $ \psi_{i_{0}} \otimes \Omega$, we deduce  that 
\begin{equation}
[H_{\theta}(\vec{p}) -z^{(\infty)}(\vec{p}))] \Psi^{(j)}(\vec{p})= \bm{\chi}_{i_{0}} \chi_{\rho_1}(H_f) ... \chi_{\rho_j}(H_f) \text{ }H^{(j)}(\vec{p},z^{(\infty)}(\vec{p}))   \vert \psi_{i_{0}} \otimes \Omega \rangle .
\end{equation}
We have that 
\begin{equation} \label{56}
H^{(j)}(\vec{p},z^{(\infty)}(\vec{p}))   \vert \psi_{i_{0}} \otimes \Omega \rangle = \sum_{M > 0} W_{M,0}^{(j)} (\vec{p},z^{(\infty)}(\vec{p}))    \vert \psi_{i_{0}} \otimes \Omega \rangle + \mathcal{E}^{(j)} (\vec{p},z^{(\infty)}(\vec{p}))  \vert \psi_{i_{0}} \otimes \Omega \rangle. 
\end{equation}
The bounds  \eqref{eq:a1}, \eqref{Bb00_n11},  and  \eqref{EJJ} imply  that the right  side in \eqref{56} is bounded in norm by some numerical constant multiplied by $\rho_j$, i.e. tends to zero super-exponentially fast  as $j$ tends to infinity.  We  show that $\Psi^{(j)}$ converges  to a  non-vanishing vector-valued function.  We remark that 
\begin{align*}
 &\Psi^{(j+1)}(\vec{p})-  \Psi^{(j)}(\vec{p}) =  Q_{\bm{\chi}_{i_{0}}}  Q_{\chi_{\rho_1}}   ... Q_{\chi_{\rho_j} } ( Q_{\chi_{\rho_{j+1}} } - \chi_{\rho_{j+1}}(H_f) )  \vert \psi_{i_{0}} \otimes \Omega \rangle\\
 &= (Q_{\bm{\chi}_{i_{0}}} - \bm{\chi}_{i_{0}} +\bm{\chi}_{i_{0}})   ... (Q_{\chi_{\rho_j} }- \chi_{\rho_j}(H_f) + \chi_{\rho_j}(H_f))( Q_{\chi_{\rho_{j+1}} } - \chi_{\rho_{j+1}}(H_f) ) \vert \psi_{i_{0}} \otimes \Omega \rangle. 
\end{align*}
Using that $\exp(x) \geq 1 +x$ for any $x \geq 0$, we deduce that
\begin{equation}
\|  \Psi^{(j+1)}(\vec{p})-  \Psi^{(j)}(\vec{p}) \| \leq  e^{\left( \sum_{k=1}^{j}  \| Q_{\chi_{\rho_k} }- \chi_{\rho_k}(H_f) \| + \|Q_{\bm{\chi}_{i_{0}}} - \bm{\chi}_{i_{0}}\| \right)} \|  Q_{\chi_{\rho_{j+1}}} - \chi_{\rho_{j+1}}(H_f) \|.
\end{equation}
Furthermore, 
\begin{align*}
\| Q_{\chi_{\rho_k} }- \chi_{\rho_k}(H_f) \| &=\|  \overline{\chi}_{\rho_k}(H_f) [H_{ \overline{\chi}_{\rho_k}(H_f)}^{(k-1)}]^{-1}_{\mid \text{Ran}(\overline{\chi}_{\rho_k}(H_f))}  \overline{\chi}_{\rho_k}(H_f)  W_{\geq 1}^{(k-1)} \chi_{\rho_k}(H_f)  \|\\
&=  \mathcal{O} \left( \mathbf{C}^{k-1} \rho_{k-1}^{2}  \mu   \right)  \| [H_{ \overline{\chi}_{\rho_k}(H_f)}^{(k-1)}]^{-1}_{\mid \text{Ran}(\overline{\chi}_{\rho_k}(H_f))}  \|  = \mathcal{O} \left( \mathbf{C}^{k-1}    \rho_{k-1}^{\varepsilon}  \right)
\end{align*}
for any $k \geq 1$. The right-hand side converges super-exponentially fast to zero and the sum over $j$  of $\| Q_{\chi_{\rho_j} }- \chi_{\rho_j}(H_f) \| $ is smaller than one if $\rho_0$ is small enough. We deduce that $\|  \Psi^{(j+1)}(\vec{p})-  \Psi^{(j)}(\vec{p}) \|$ converges super-exponentially fast to zero.  This implies that $(\Psi^{(j)} )$ is uniformly Cauchy and hence converges uniformly. Since each $\Psi^{(j)}$ is analytic in $\vec{p}$, the limit $ \Psi^{(\infty)}$ is also analytic on $U_{\theta}[\vec{p}^*]$. Furthermore, a similar argument as above shows that
$$\| \Psi^{(j)}(\vec{p}) -  \Omega   \| < 1$$
for all $\vec{p}$. Therefore,  $ \Psi^{(\infty)}(\vec{p}) \neq 0$ for all $\vec{p}$, which implies that
\begin{equation*}
\mathrm{dim} \, \mathrm{Ker} ( H_{\theta}(\vec{p}) -z^{(\infty)}(\vec{p}) ) \ge 1.
\end{equation*}
Using the ``iso-spectrality'' of the Feshbach-Schur map together with the fact that any vector $\Psi \in \mathcal{H}_f$  satisfies the $\underset{j  \rightarrow \infty}{\lim}  \chi_{{\rho_j}(H_f)} \Psi = \langle \Omega \vert \Psi \rangle \Omega$,  we verify that
\begin{equation*}
\mathrm{dim} \, \mathrm{Ker} ( H_{\theta}(\vec{p}) -z^{(\infty)}(\vec{p}) ) \le 1.
\end{equation*}
Our argumentation is similar to \cite{HaHe12_01}.   Let $n \in \mathbb{N}$.  We  assume that  $z^{(\infty)}(\vec{p})$ is degenerate.  We can find a non-zero vector $ \tilde{\Psi} \in \mathcal{H}_f$ such that  $P_{\Omega}  \tilde{\Psi} =0$ and $H^{(n)}(\vec{p},z^{\infty}(\vec{p})) \chi_{\rho_n}(H_f) \tilde{\Psi}=0$. We take $n$ large enough such that $\chi_{\rho_n}(H_f) \chi_{\rho_j}(H_f) = \chi_{\rho_j}(H_f) $ for all $j >n$.  We  use that 
\begin{align*}
\mathds{1}_{H_f  \leq \rho_n}  &\geq \chi_{\rho_n} (H_f),\\
  \mathds{1}_{H_f  \leq  \rho_n} &=  P_{\Omega}+  \sum_{j=n+1}^{\infty}  \mathds{1}_{\frac{3}{4} \rho_{j+1} < H_f  \leq \frac{3}{4}  \rho_j}  +  \mathds{1}_{\frac{3}{4} \rho_{n+1} < H_f  \leq  \rho_n}, 
  \end{align*}
to estimate the norm of $ \chi_{\rho_n} (H_f) \tilde{\Psi}$. We get: 
\begin{align*}
\|  \chi_{\rho_n} (H_f) \tilde{\Psi}\|^2 &\leq  \sum_{j=n+1}^{\infty}  \| \mathds{1}_{  \frac{3}{4}  \rho_{j+1} < H_f  \leq \frac{3}{4}  \rho_j}  \chi_{\rho_n}(H_f)  \tilde{\Psi} \|^2 + \|   \mathds{1}_{  \frac{3}{4}  \rho_{n+1} < H_f  \leq   \rho_n}   \chi_{\rho_n} (H_f) \tilde{\Psi}\|^2 \\
& \leq  \frac{4}{ \mu^2} \sum_{j=n}^{\infty}   \rho_{j+1} ^{-2} \big  \| (W^{(j)}_{0,0}(\vec{p},z^{\infty}(\vec{p})) + \mathcal{E}^{(j)}(\vec{p},z^{\infty}(\vec{p}))) \mathds{1}_{ \frac{3}{4}\rho_{j+1}< H_f  \leq \rho_j}  \chi_{\rho_j}(H_f)  \tilde{\Psi}  \big \|^2.
\end{align*} 
 To go from the first line to the second line, we have used the estimate
$$ \vert w^{(j)}_{0,0} (\vec{p},z^{\infty}(\vec{p}),r,\vec{l})  + \mathcal{E}^{(j)} ( \vec{p},z^{\infty}(\vec{p})) \vert \ge \frac{ \mu}{2} \rho_{j+1}, \qquad  \forall r \ge \frac{3}{4}  \rho_{j+1}, \vert \vec{l} \vert \leq r,$$
proven at the beginning of Paragraph \ref{Fesh1}, and the  equality  $$   \mathds{1}_{  \frac{3}{4}  \rho_{j+1} < H_f  \leq   \frac{3}{4} \rho_j}   \chi_{\rho_n} (H_f) =      \mathds{1}_{  \frac{3}{4}  \rho_{j+1} < H_f  \leq    \frac{3}{4} \rho_j}   \chi_{\rho_j} (H_f). $$  Since  $ \mathds{1}_{ 3\rho_{j+1} /4< H_f  \leq \rho_j}$ commutes with $W^{(j)}_{0,0}(\vec{p},z^{\infty}(\vec{p})) + \mathcal{E}^{(j)}(\vec{p},z^{\infty}(\vec{p}))$, and  since $$H^{(j)}(\vec{p},z^{\infty}(\vec{p})) \chi_{\rho_j}(H_f) \tilde{\Psi}=0,$$ we deduce that 
\begin{align*}
\|  \chi_{\rho_n} (H_f) \tilde{\Psi}\|^2  & \leq   \frac{4}{ \mu^2} \|  \chi_{\rho_n} (H_f) \tilde{\Psi}\|^2     \sum_{j=n}^{\infty}  \rho_{j+1}^{-2}  \big \| W_{\geq 1}^{(j)}(\vec{p},z^{\infty}(\vec{p})) \big \|^{2}.
\end{align*}
We have shown in \eqref{4.27} that there exists a numerical constant $C$ independent of $j$ and the problem parameters, such that 
\begin{equation*}
\|W^{(j)}_{\geq 1} (\vec{p},z) \| \leq C  \mathbf{C}^{j}  \rho_{j}^{2} \mu,
\end{equation*}
for all $j \in \mathbb{N}$. Since $\sum_{j} \mathbf{C}^{2j}  \rho_{j}^{4} / \rho_{j+1}^{2}$ converges,  $\|  \chi_{\rho_n} (H_f) \tilde{\Psi}\|$ must be zero for large values of $n$. This contradicts Theorem \ref{feshbach-theorem}, and, therefore, $z^{\infty}(\vec{p})$ cannot be degenerate.

\end{proof}

\section{ Imaginary part of the resonances ($z^{(\infty)}(\vec p)$)}
\label{fermi}

In this section, assuming that Fermi's Golden Rule holds, we prove that the imaginary part of $z^{(\infty)}(\vec p)$ is strictly negative for small enough values of the coupling constant $\lambda_0$. More precisely, using the isospectrality of the Feshbach-Schur map (see Definition \ref{feshbach} and Theorem \ref{feshbach-theorem}), we verify that the operator $H_\theta(\vec p) - z$ is invertible for any $z \in \mathbb{C}$ such that $\Im z \ge - \mathrm{c}_0 \lambda_0^2$, where $\mathrm{c}_0$ is a positive constant.

\subsection{Computing the leading part of the Feshbach-Schur Map} \label{sleading}
We recall from Lemma \ref{fespair} that if the parameters $\lambda_0$, $\rho_0$ and $\theta = i \vartheta$ satisfy the conditions $\lambda_0^2 \sigma_\Lambda^{3} (\mu \sin \vartheta )^{-2} \ll \rho_0 < \min ( 1 , \delta_0 )$, then $(H_{\theta}(\vec{p})-z,H_{\theta,0}(\vec{p})-z)$ is a Feshbach-Schur pair associated to $\bm{\chi}_{i_{0}}$ for any $(\vec{p},z) \in \mathcal{U}_{\rho_0}[E_{i_{0}}]$. The corresponding Feshbach-Schur operator is given by
\begin{align}
&F_{\bm{\chi}_{i_{0}}}(H_{\theta}(\vec{p}) - z , H_{\theta,0}(\vec{p}) - z ) = H_{\theta,0}(\vec p) - z + \lambda_0 \bm{\chi}_{i_{0}} H_{I,\theta} \bm{\chi}_{i_{0}}  \notag \\
&\qquad - \lambda_{0}^{2} \bm{\chi}_{i_{0}} H_{I,\theta}  \bm{\overline{\chi}}_{i_{0}}  [H_{\bm{\overline{\chi}}_{i_{0}} }(\vec{p},z)]^{-1}_{| \mathrm{Ran}(\bm{\overline{\chi}}_{i_{0}})}  \bm{\overline{\chi}}_{i_{0}} H_{I,\theta}  \bm{\chi}_{i_{0}} . \label{fesh}
\end{align}
We recall from Theorem \ref{feshbach-theorem} the iso-spectral property:
\begin{align}
& H_\theta(\vec p) - z \ \text{is bounded invertible} \notag \\
&  \Longleftrightarrow F_{\bm{\chi}_{i_{0}}}(H_{\theta}(\vec{p}) - z , H_{\theta,0}(\vec{p}) - z )_{\mid P_{i_0} \otimes \mathds{1}_{H_f \leq \rho_0}[\mathcal{H}_{\vec{p}}]} \ \text{is bounded invertible}. \label{isos1}
\end{align}

The estimation of the imaginary part of $z^{(\infty)}(\vec p)$ relies on the analysis of  $F_{\bm{\chi}_{i_{0}}}(H_{\theta}(\vec{p}) - z , H_{\theta,0}(\vec{p}) - z )$. We set
\begin{equation}\label{w01}
w^{\theta}_{0,1}(\underline k) : = i e^{-2 \theta } \Lambda(e^{ -\theta} \vec k)  |\vec k|^{1/2} \vec 
\epsilon(\underline k)\cdot \vec d, \hspace{.5cm}  w^{\theta}_{1,0}(\underline k) = - w^{\theta}_{0,1}(\underline k),  
\end{equation}
and
\begin{align}\label{Zs}
Z^{od}(\vec p) : = & \int d \underline k  P_{i_0} w^{\theta}_{0,1}(\underline k) \overline{P}_{i_0}
\big (  H_{is} - E_{i_0}  + e^{-\theta} |\vec k| + e^{- 2 \theta} \frac{\vec k^2}{2}  -
e^{-\theta} \vec p\cdot \vec k \big )^{-1}\overline{P}_{i_0}w^{\theta}_{1,0}(\underline k) P_{i_0}, \\ \notag
Z^{d}(\vec p) : = & \int d \underline k P_{i_0} w^{\theta}_{0,1}(\underline k) 
P_{i_0}  \big ( e^{-\theta} |\vec k| + e^{- 2 \theta} \frac{\vec k^2}{2}  - e^{-\theta} \vec p\cdot \vec k \big )^{-1} w^{\theta}_{1,0}(\underline k) P_{i_0} ,
\end{align}
where, recall, $P_{i_0}$ is the orthogonal projection onto the one-dimensional eigenspace associated to the eigenvalue $E_{i_0}$ of $H_{is}$.

In the next lemma we identify the leading order term of $F_{\bm{\chi}_{i_{0}}}(H_{\theta}(\vec{p}) - z , H_{\theta,0}(\vec{p}) - z )$ in terms of powers of $\lambda_0$.

\begin{lemma}\label{leading-order}
Under the conditions of Lemma \ref{fespair}, there is a bounded operator $ {\rm Rem}$ such that
\begin{align}\label{leadingeq}
F_{ \bm{\chi}_{i_{0}} }(H_{\theta}(\vec{p}) - z , H_{\theta,0}(\vec{p}) - z ) = & \big [E_{i_0} - z  + e^{-\theta} H_f + e^{- 2 \theta} \frac{ \vec P_f^2}{2}  - e^{-\theta} \vec p\cdot \vec P_f \big] \\ \notag & - \lambda_0^2 Z^{d}(\vec p) \otimes \chi^2_{\rho_0}( H_f ) - 
\lambda_0^2 Z^{od}(\vec p) \otimes \chi^2_{\rho_0}( H_f )  + {\rm Rem},
\end{align}
and 
\begin{equation}\label{rem}
\| {\rm Rem} \| \leq C \lambda_0^{2} \Big ( \frac{ \sigma_\Lambda^{9/2} }{ \mu^2 \sin(\vartheta)^2 \min(1, \delta_0^2) } \Big ( \frac{\lambda_0}{{\rho_0}^{1/2}} + 
\frac{{\rho_0}^2}{\lambda_0} \Big ) + {\rho_0} \Big ) ,
\end{equation}
where $C$ is a positive constant.
\end{lemma}

\noindent \emph{Proof.} 
The proof follows the 
lines of Lemma 3.16 of \cite{BFS-1999}, where similar results are shown for a different model (see also \cite{BaFrSi98_01}, where all details are included).    
\qed

\begin{remark} \label{r-leading}
Lemma \ref{leading-order} gives the leading order contribution of 
$F_{ \bm{\chi}_{i_{0}} }(H_{\theta}(\vec{p}) - z , H_{\theta,0}(\vec{p}) - z )$ provided that we choose, for instance,
\begin{equation} \label{67}
\rho_0 = \lambda_0^{4/5} .
\end{equation}
The condition $\lambda_0^2 \sigma_\Lambda^{3} (\mu \sin \vartheta )^{-2} \ll \rho_0$ is then satisfied if we require that $\lambda_0^{6/5} \ll \sigma_\Lambda^{-3} (\mu \sin \vartheta )^{2}$.  Equations \eqref{leadingeq} and \eqref{rem} give
\begin{align}
F_{\bm{\chi}_{i_{0}}}(H_{\theta}(\vec{p}) - z , & H_{\theta,0}(\vec{p}) - z ) =  \big [E_{i_0} - z  + e^{-\theta} H_f + e^{- 2 \theta} \frac{ \vec P_f^2}{2}  - e^{-\theta} \vec p\cdot \vec P_f \big] - \lambda_0^2 Z^{d}(\vec p)\otimes \chi^2_{\rho_0}( H_f )  \notag \\ & -  
\lambda_0^2 Z^{od}(\vec p) \otimes \chi^2_{\rho_0}( H_f ) 
+ \lambda_0^{2 + 3/5} \mathcal{O}\Big(  \frac{\sigma_\Lambda^{9/2} }{ \sin(\vartheta)^2 \mu^2
\min(1, \delta_0^2)}\Big).
\end{align}
The remainder term is small compared to $\lambda_0^2$ provided that 
$$
\lambda_0^{3/5} \ll \sigma_\Lambda^{-9/2} ( \mu \sin \vartheta \min(1, \delta_0))^{2}.
$$
\end{remark}

\subsection{The Imaginary Part of $z^{(\infty)}(\vec p)$} \label{imazinfty}
In this section we estimate the imaginary part of $z^{(\infty)}(\vec p)$, assuming here that $\vec p$ has real entries.  

We study the leading order term of  $F_{\bm{\chi}_{i_{0}}}(H_{\theta}(\vec{p}) - z , H_{\theta,0}(\vec{p}) - z )_{\mid P_{i_0} \otimes \mathds{1}_{H_f \leq \rho_0}[\mathcal{H}_{\vec{p}}]}$, that we denote by  $H_{L}(\vec p) - z$ (see \eqref{Hlead} below). It is a normal operator whose spectrum is explicitly computable: 
As $ Z^{od}(\vec p) $ and $Z^{d}(\vec p)$ are rank-one operators, we can write $Z^{od}(\vec p) =  z^{od}(\vec p)P_{i_0} $ and $Z^{d}(\vec p) =  z^{d}(\vec p)P_{i_0} $, for some complex numbers $ z^{od}(\vec p) $ and 
$ z^{d}(\vec p) $. Then we can write  $H_{L}(\vec p)$   as the sum of 
$ (E_{i_0} - \lambda_0^2 z^{d}(\vec p) - \lambda_0^2 z^{od}(\vec p))P_{i_0}  $ plus an operator that is a function of $H_f $ and $\vec P_f$. The spectrum of the latter operator lies in the lower half plane, which can be easily shown from geometrical considerations, using the spectral theorem. Using analyticity in $\theta$ we show that $z^{d}(\vec p)$ is real, which implies that the imaginary part of the spectrum of $H_{L}(\vec p)$ is below $ - \lambda_0^2 \Im z^{od}(\vec p) $. Proving that $ - \Im z^{od}(\vec p) < 0$ is, thus, essential to show that $    \Im z^{(\infty)}(\vec p) < 0 $. This is where the Fermi Golden Rule is used, see Subsection \ref{ana-Zod}. 

Once we have proven that the imaginary part of the spectrum of $H_{L}(\vec p)$ is below $ - \lambda_0^2 \Im z^{od}(\vec p) $, which is negative, 
we conclude by a perturbative argument, using a Neumann series expansion, that
$ H_{L}(\vec p) - z + {\rm Rem} $ is invertible if $\Im z$ is larger than a (strictly) negative number (for small $\lambda_0$). This assertion and the iso-spectrality of the Feshbach-Schur map then imply that $H_\theta(\vec p) - z$ is invertible for such $z$'s, from which we conclude in Subsection \ref{princ-estm} that $\Im z^{(\infty)}(\vec p)$ is (strictly) smaller than zero.

\subsubsection{Analysis of $Z^{od}(\vec p)$} \label{ana-Zod}
\begin{proposition}
Let $\vec p\in U_{\theta}[\vec{p}^*]\cap\mathbb{R}^3$.
We define (see \eqref{psi0})
\begin{equation}
z^{od}(\vec p) : = \langle \psi_{i_0}  | \; Z^{od}(\vec p)  \psi_{i_0}  \rangle.
\end{equation}
The following holds true:
\begin{align}\label{prop}
\Im z^{od}(\vec p) =   \pi \sum_{j < i_0} \int_{\underline{\mathbb{R}}^3} d \underline{k} & \big | \sum_{s \in \{1, 2, 3\}}(d_s)_{N-j + 1,N- i_0 +1} \epsilon_s(\underline{k})  \big |^2 \\ &  \notag |\vec k| | \Lambda (\vec k)|^2 \delta\big(E_j - E_{i_0} + |\vec k| 
-  \vec p \cdot \vec k + \frac{\vec k^2}{2} \big). 
\end{align}
\end{proposition}
\noindent \emph{Proof.}
We denote by 
\begin{equation} \widehat{k} : = \frac{1}{|k|}\vec k, \hspace{.5 cm} 
r \underline k : =  (r \vec k, \lambda), \hspace{.5cm}  
\widehat{\underline k} : =   (\widehat{k}, \lambda), \hspace{.5cm} 
  \forall \underline k  \in \underline{\mathbb{R}}^3,  \: \underline k \ne (0,\lambda), \: \forall r \geq  0.
\end{equation} 
Let
\begin{align}\label{fod}
f_{\widehat k}^{od}(z) : = &  z^3 e^{-z^2/ \sigma_{\Lambda}^2}    \\ \notag 
& \cdot \langle \psi_{i_0}  |  \vec \epsilon(\widehat{\underline{k}})\cdot \vec d \:  \overline{P}_{i_0}
\big (  H_{is} - E_{i_0}  +z  +  \frac{z^2}{2}  -
z  \vec p\cdot \widehat{\underline{k}} \big )^{-1}   
\overline{P}_{i_0} \vec \epsilon(\widehat{\underline{k}})\cdot \vec d \:  \psi_{i_0}  \rangle .
\end{align}
Using spherical coordinates we obtain that (see \eqref{w01} and \eqref{Zs})  
\begin{equation}\label{zodfod}
z^{od}(\vec p) = \int d  \widehat{\underline k}\int_{0 }^{\infty}  dr e^{-\theta}  f_{\widehat k}^{od}(e^{-\theta}r),
\end{equation}
where $ \int   d  \widehat{\underline k} $ denotes the surface integral over the $2$ dimensional sphere in 
$\underline{\mathbb{R}}^3  $.  

Let $\gamma_\theta : [0, \infty) \to \mathbb{C}$ be the path defined by the formula
\begin{equation}\label{gammatheta}
\gamma_{\theta}(r) : =  e^{-\theta } r . 
\end{equation}
We denote furthermore, for every $R > 0$, by $\gamma_{\theta, R} $ the restriction of $\gamma_\theta$ to the interval 
$[0, R]$ and by $\tilde \gamma_{\theta, R}$ the straight (oriented) line segment with starting point 
$\gamma_{\theta}(R)$ and ending point $\gamma_{\overline{\theta}}(R)$.

Eq. \eqref{zodfod} implies that we can view the integral with respect to $r$ as a complex integral with respect to
$\gamma_\theta$. It follows furthermore that 
\begin{equation}\label{imzod}
\Im z^{od}(\vec p) = \frac{1}{2i} \int d  \widehat{\underline k}
\Big[ \int_{\gamma_{\theta}} f_{\widehat k}^{od}( z ) dz -   \int_{\gamma_{\overline{\theta}}} f_{\widehat k}^{od}( z ) dz \Big].
\end{equation}
The function $ f_{\widehat k}^{od}  $ is meromorphic in the region delimited by the curves $ \gamma_{\theta} $ and  
 $\gamma_{\overline{\theta}}$ containing the positive part of the real axis. The poles  of $ f_{\widehat k}^{od}  $
 in this region are the positive real numbers $r$ such that  $E_{i_0}  - r -  \frac{r^2}{2}  + 
 r  \vec p\cdot \widehat{\underline{k}} $ belongs $ \sigma (H_{is}) \setminus \{E_{i_0}\}$.
Since $  - r -  \frac{r^2}{2}  + 
 r  \vec p\cdot \widehat{\underline{k}}\, $ is strictly negative and strictly decreasing as a function of $r$ for $r > 0$, there are only $i_0 - 1$ poles and they correspond to the positive real numbers $r^{od}_j$ such that 
\begin{equation}\label{poles}
E_{i_0}  - r^{od}_j -  \frac{(r^{od}_j)^2}{2}  + 
 r^{od}_j  \vec p\cdot \widehat{\underline{k}}= E_j 
\end{equation}
for some $ j <i_0 $. In fact the explicit solutions of Eq. \eqref{poles} are given by the formula  
\begin{align}\label{poles1}
E_j   - E_{i_0} +    \frac{r^2}{2}  + 
 r(1 -  \vec p\cdot \widehat{\underline{k}}) = &
 \frac{1}{2}\Big [r - \Big ( -(1 -  \vec p\cdot \widehat{\underline{k}}) + \sqrt{2(E_{i_0} - E_j) + (1 -  \vec p\cdot \widehat{\underline{k}})^2}    \Big ) \Big ] \\ \notag & 
  \cdot \Big [r - \Big ( -(1 -  \vec p\cdot \widehat{\underline{k}}) - \sqrt{2(E_{i_0} - E_j) + (1 -  \vec p\cdot \widehat{\underline{k}})^2}    \Big ) \Big ].  
\end{align}

\noindent Let $R > 0$ such that the poles are contained in the interior of the (closed) curve $  \gamma_{\theta, R} + \tilde \gamma_{\theta, R} - \gamma_{\overline{\theta}, R} $ (the curve 
  $  \gamma_{\theta, R} $ followed by $\tilde \gamma_{\theta, R}  $ and this last one followed by 
  $ - \gamma_{\overline{\theta}, R} $, which is the curve $  \gamma_{\overline{\theta}, R} $ going in the contrary direction). It follows from the exponential decay of $ f_{\widehat k}^{od}  $  that 
\begin{equation}\label{morera}
\int_{\gamma_{\theta}} f_{\widehat k}^{od}( z ) dz -   \int_{\gamma_{\overline{\theta}}} f_{\widehat k}^{od}( z ) dz
= \int_{  \gamma_{\theta, R} + \tilde \gamma_{\theta, R} - \gamma_{\overline{\theta}, R} } f_{\widehat k}^{od}(z).
\end{equation}  
From the residue theorem and \eqref{imzod}-\eqref{morera} we conclude that   
\begin{equation}\label{res-imzod}
\Im z^{od}(\vec p) = \pi \int d \widehat{\underline{k}} \sum_{ j \in \{1, \cdots, i_0 -1 \}} {\rm Res}(f_{\widehat k}^{od}, r^{od}_j). 
\end{equation} 
We obtain finally \eqref{prop} from \eqref{res-imzod} and from the fact that
\begin{align}\label{res-j}
{\rm Res}(f_{\widehat k}^{od}, r^{od}_j) = & \lim_{z\to r_j^{od}} (z - r_j^{od}) f_{\widehat k}^{od}(z) \\ \notag = & \lim_{z\to r_j^{od}} (z - r_j^{od}) z^3e^{-  z^2/\sigma_{\Lambda}^2}    \\ \notag 
& \cdot \langle \psi_{i_0}  | \vec \epsilon(\widehat{\underline{k}})\cdot \vec d \:  P_{j}
\big (  H_{is} - E_{i_0}  + z +  \frac{z^2}{2}  -
z \vec p\cdot \widehat{\underline{k}} \big )^{-1}   
P_{j} \vec \epsilon(\widehat{\underline{k}})\cdot \vec d \: \psi_{i_0} \rangle \\ \notag 
 \notag = & (r_j^{od})^3 e^{-  (r_j^{od})^2/\sigma_{\Lambda}^2}    \\ \notag 
& \cdot \langle \psi_{i_0}  | \vec  \epsilon(\widehat{\underline{k}})\cdot \vec d \: P_{j}
\frac{1}{\sqrt{2(E_{i_0} - E_j) + (1 -  \vec p\cdot \widehat{\underline{k}})^2}}    
P_{j} \vec \epsilon(\widehat{\underline{k}})\cdot \vec d \: \psi_{i_0} \rangle, 
\end{align} 
where we used \eqref{poles}-\eqref{poles1}. 
\qed

\subsubsection{Estimations of $\Im z^{(\infty)}(\vec p)$}\label{princ-estm}
\begin{theorem}\label{princ}
Suppose that the parameters $\theta$, $\lambda_0$, $\rho_0$ satisfy the conditions of Lemma \ref{fespair}, Eq. \eqref{67},   and $\lambda_0^{3/5} \ll \sigma_\Lambda^{-9/2} ( \mu \sin \vartheta \min(1,\delta_0))^{2}$. There exists a positive constant $C$ independent of the problem parameters such that, for all $\vec p\in U_{\theta}[\vec{p}^*] \cap \mathbb{R}^3$ and $z \in \mathbb{C}$ such that $|z - E_{i_0}| < \rho_0 \mu \sin(\vartheta) / 32$ and
$$
\lambda_0^2  \Big[ C \frac{\sigma_\Lambda^{9/2} }{ \sin(\vartheta)^2 \mu^2
\min(1, \delta_0^2)} \lambda_0^{3/5}  - \Im z^{od}(\vec p) \Big] <  \Im z,
$$
the operator $H_{\theta}(\vec{p})- z$ is (bounded) invertible. In particular, if  $  \Im z^{od}(\vec p) > 0 $ (i.e. the Fermi Golden Rule is satisfied), the imaginary part of $ z^{(\infty)}(\vec p)$ is strictly negative.  

\end{theorem}
\begin{proof}
We define 
\begin{equation}
 z^{d}(\vec p) : = \langle \psi_{i_0}  | \; Z^{d}(\vec p)  \psi_{i_0} \rangle.  
\end{equation}
Applying the procedures of Paragraph \ref{ana-Zod}, it is easy to prove that $z^{d}(\vec p)  $ does not depend on $\theta$.  $z^{d}(\vec p)  $  is therefore real, since for $\theta = 0$,  $z^{d}(\vec p) \in \mathbb{R}$. Using \eqref{leadingeq} and \eqref{rem} we obtain that 
\begin{align}\label{leadingeqprima}
F_{\bm{\chi}_{i_{0}}}(H_{\theta}(\vec{p}) - z , H_{\theta,0}(\vec{p}) - z )_{\mid P_{i_0} \otimes \mathds{1}_{H_f \leq \rho_0}[\mathcal{H}_{\vec{p}}]} = &  (H_{L}(\vec p)- z)_{\mid P_{i_0} \otimes \mathds{1}_{H_f \leq \rho_0}[\mathcal{H}_{\vec{p}}]}    + {\rm Rem},
\end{align}
where
\begin{align}\label{Hlead}
H_{L}(\vec p) : = & E_{i_0} - \big [ \lambda_0^2 \Re z^{od}(\vec p) + \lambda_0^2z^{d} (\vec p) \big] \bm{\chi}_{i_{0}}^2
- i \lambda_0^2\Im z^{od}(\vec p) \bm{\chi}_{i_{0}}^2  \\ \notag & 
 + \big[ e^{-\theta} H_f + e^{- 2 \theta} \frac{ \vec P_f^2}{2}  - e^{-\theta} \vec p\cdot \vec P_f \big] 
\end{align}
and
\begin{equation}\label{remprima}
\| {\rm Rem} \| \leq C \lambda_0^{2 + 3/5} \frac{\sigma_\Lambda^{9/2} }{ \sin(\vartheta)^2 \mu^2 \min(1, \delta_0^2)} . 
\end{equation}
As $| \vec p | < 1$, it follows that for every $ r \geq 0 $ and every $\vec l \in \mathbb{R}^3$ with $|\vec l| \leq r$ ,
\begin{equation}\label{tempo}
r - \vec p \cdot \vec l \geq 0, 
\end{equation}
which implies that 
\begin{equation}\label{tempo.1}
\Im \Big[e^{-\theta} r + e^{- 2 \theta} \frac{\vec l^2}{2} - e^{-\theta} \vec p \cdot \vec l  \Big] \leq 0.  
\end{equation}
Eq.\eqref{tempo.1} and the Spectral Theorem applied to the normal operator $H_{L}(\vec p)$ imply that 
$H_{L}(\vec p) - z$ restricted to the range of $\bm{\chi}_{i_{0}}$ is invertible for $\Im z > - \lambda_0^2 \Im z^{od}(\vec p)  $ and that 
\begin{equation}\label{tempo.2}
\| (H_{L}(\vec p) - z)_{\mid P_{i_0} \otimes \mathds{1}_{H_f \leq \rho_0}[\mathcal{H}_{\vec{p}}]}^{-1}\| \leq \frac{1}{| \Im z + \lambda_0^2 \Im z^{od}(\vec p) |}. 
\end{equation}
A Neumann series expansion together with \eqref{remprima} and \eqref{tempo.2} imply that $F_{\bm{\chi}_{i_{0}}}(H_{\theta}(\vec{p}) - z , H_{\theta,0}(\vec{p}) - z )_{\mid P_{i_0} \otimes \mathds{1}_{H_f \leq \rho_0}[\mathcal{H}_{\vec{p}}]}$ is invertible, for $\Im z > - \lambda_0^2 \Im z^{od}(\vec p)$, whenever 
\begin{equation}\label{tempo.3}
C \lambda_0^{2 + 3/5} \frac{\sigma_\Lambda^{9/2} }{ \sin(\vartheta)^2 \mu^2 \min(1, \delta_0^2)} \frac{1}{| \Im z + \lambda_0^2 \Im z^{od} (\vec p)|} < 1 .
\end{equation}
The conclusions of Theorem \ref{princ} follow from this last assertion and the iso-spectrality of the
Feshbach-Schur map.      
\end{proof}

\appendix 
\vspace{4mm}

\section{Analyticity of Type (A)}
\label{typA}
We recall that a map $\eta \mapsto A_\eta$ from an open connected set $V \subset \mathbb{C}^n$ to the set of (unbounded) operators in a Hilbert space $\mathcal{H}$ is called analytic of type (A) if there is a dense domain $D \subset \mathcal{H}$ such that, for all $\eta \in V$, $A_\eta$ is closed on $D$, and for all $ \psi \in D$, the map $\eta \mapsto A_\eta \psi$ is analytic on $V$.

\begin{lemma} \label{typeA}
Let $U := \{ \vec{p} \in \mathbb{C}^3 \mid  \vert \Re(\vec{p}) \vert +\vert \Im(\vec{p}) \vert  < 1 \}$. The map $(\theta,\vec{p},\lambda_0) \mapsto H_\theta( \vec{p} )$ is analytic of type (A) on $D(0, \pi / 4) \times  U \times \mathbb{C}$.
\end{lemma}
\begin{proof}
Let $\theta \in D( 0 , \pi / 4 )$, $\vec{p} \in \mathbb{C}^3$ with $| \Re \vec{p} | + | \Im \vec{p} | < 1$, and $\lambda_0 \in \mathbb{C}$. For all $\psi \in D( H_f ) \cap D( \vec{P}_f^2 )$, we have that
\begin{align*}
& \big \| \big (Êe^{- 2 \theta} \frac{\vec{P}_f^2}{2} - e^{-\theta}  \vec{p} \cdot \vec{P_f} + e^{-\theta} H_{f} \big ) \psi \big \|^2 \\
& = e^{- 2 \Re \theta} \big \|Ê( H_f - \vec{p} \cdot \vec{P}_f ) \psi \big \|^2 + \frac14 e^{-4 \Re \theta}  \big \| \vec{P}_f^2 \psi \big \|^2 + \Re \big ( e^{- \theta } e^{- 2 \bar \theta } \big \langle ( H_f - \vec{p} \cdot \vec{P}_f ) \psi \vert  \text{ }\vec{P}_f^2 \psi \big \rangle \big ) \\
& \ge e^{- 2 \Re \theta} ( 1 - | \vec{p} | )^2  \| H_f \psi \|^2 + \frac14 e^{-4 \Re \theta}  \big \| \vec{P}_f^2 \psi \big \|^2 + e^{-3 \Re \theta } \Re \big ( e^{ i \Im \theta } \big \langle ( H_f - \vec{p} \cdot \vec{P}_f ) \psi \vert  \text{ }\vec{P}_f^2 \psi \big \rangle \big ) .
\end{align*}
Using that $| \sin \Im \theta | \le \cos \Im \theta$, a direct computation gives
\begin{align*}
\Re \big ( e^{ i \Im \theta } \big \langle ( H_f - \vec{p} \cdot \vec{P}_f ) \psi \vert \text{ }\vec{P}_f^2 u \rangle \big ) &= \cos ( \Im \theta ) \big \langle ( H_f - \Re \vec{p} \cdot \vec{P}_f ) \psi \vert   \text{ } \vec{P}_f^2 u \big \rangle  + \sin ( \Im \theta ) \big \langle \Im \vec{p} \cdot \vec{P}_f \psi \vert  \text{ } \vec{P}_f^2 \psi \big \rangle \\
& \ge \cos ( \Im \theta ) ( 1 - | \Re \vec{p} | - | \Im \vec{p} | ) \big \langle H_f \psi \vert   \text{ } \vec{P}_f^2 \psi \big \rangle \ge 0 ,
\end{align*}
and hence
\begin{align*}
& \big \| \big (Êe^{- 2 \theta} \frac{\vec{P}_f^2}{2} - e^{-\theta}  \vec{p} \cdot \vec{P_f} + e^{-\theta} H_{f} \big ) \psi \big \|^2 \ge e^{- 2 \Re \theta} ( 1 - | \vec{p} | )^2  \| H_f \psi \|^2 + \frac14 e^{-4 \Re \theta}  \big \| \vec{P}_f^2 \psi \big \|^2 .
\end{align*}
This implies that $e^{- 2 \theta} \vec{P}_f^2 / 2 - e^{-\theta}  \vec{p} \cdot \vec{P_f} + e^{-\theta} H_{f} $ is closed on $D( H_f ) \cap D( \vec{P}_f^2 )$. Since $H_{I,\theta}$ is relatively $H_f^{1/2}$-bounded, it is infinitesimally small with respect to $H_f$, and therefore, since in addition $H_{is}$ is bounded, we easily deduce that $H_\theta( \vec{p} )$ is closed on $ D( H_f ) \cap D( \vec{P}_f^2 )$.

Verifying that $(\theta,\vec{p},\lambda_0) \mapsto H_\theta( \vec{p} )$ is analytic on $D( 0 , \pi / 4 ) \times U \times \mathbb{C}$ for all $\psi \in D( H_f ) \cap D( \vec{P}_f^2 )$ is straightforward (Here we need in particular that the ultraviolet cutoff function $\Lambda$ is real analytic).
\end{proof}

\section{Proof of Lemmas \ref{kerop} and \ref{tech1}} \label{AppendixB}
\subsection{ Proof of Lemma \ref{kerop}}
Let $\varphi,\psi \in \mathrm{Ran} ( \mathds{1}_{H_f \leq \rho} \mathcal{H}_f )$.  We have that
\begin{align*}
\vert \langle \psi \vert W_{m,n}  \varphi \rangle  \vert & \leq \|w_{m,n} \|_{\frac{1}{2} }  \int_ {\underline{B}_{\rho }^{(m,n)}}   d \underline{K}^{(m,n)}   \| a(\underline{k}^{(m)})  \psi \| \| a(\underline{\tilde{k}}^{(n)})  \varphi\|  \text{ } \vert  \underline{k}^{(m)} \vert ^{1/2} \vert \underline{\tilde{k}}^{(n)} \vert^{1/2}\\
& \leq  \|w_{m,n} \|_{\frac{1}{2}}  \mathcal{V}_{m}^{1/2} \mathcal{V}_{n}^{1/2} D_{n}(\varphi)^{1/2} D^{1/2}_{m}(\psi),   
\end{align*}
where
\begin{align*}
\mathcal{V}_{m}&:= \int_ {\underline{B}_{\rho}^{(m)}}  d \underline{k}^{(m )},\\
 D_{n}(\varphi) &:=   \int_ {\underline{B}_{\rho}^{(n)}}  d \underline{k}^{(n )}  \vert  \underline{k}^{(n)} \vert  \text{ }  \| a(\underline{k}^{(n)})  \varphi\|^{2}.
\end{align*}
A direct computation gives $\mathcal{V}_m \le (m!)^{-1} (8\pi)^m \rho ^{3m}$ and an easy argument by induction shows that  $ D_{n}(\varphi) \le \|H_{f}^{n/2} \varphi \|^{2} \leq \rho^{n} \|\varphi \|^{2}$. Putting all the bounds together, we find that
\begin{equation}
\vert \langle \psi \vert W_{m,n} \varphi \rangle  \vert  \leq \rho^{2(m+n)} \|w_{m,n} \|_{\frac{1}{2}} \frac{ (8 \pi)^{\frac{m+n}{2}} }{\sqrt{m! n!}} \|\psi \| \|\varphi \| ,
\end{equation}
which implies \eqref{eq:a1}.

\subsection{Proof of Lemma \ref{tech1}}

\vspace{2mm}

\subsubsection{Proof of the estimate \eqref{HI1} }
Let  $j \in \{1,...,3\}$. We introduce
$$f_j(\underline{k})= -i \overline{\Lambda}( e^{-  \theta}  \vec{k}) \vert \vec{k} \vert^{1/2}  \epsilon_j(\underline{k}).$$
It is sufficient to bound  $ a(f_j) (H_{f}+\rho )^{-1/2} $.  We set $\omega(\vec{k})= \vert \vec{k} \vert$. Thanks to the pull-through formula, we have that for any $\psi \in \mathcal{H}_f$, 
\begin{align*}
 \| a(f_j)& (H_{f}+\rho )^{-1/2}  \psi \| \leq \int d\underline{k} \text{ } \|   (H_{f} +\rho + \vert \vec{k} \vert )^{-1/2} a (\underline{k}) \overline{f_j(\underline{k})} \psi \| \\
&\leq     \big \| \frac{ f_j  }{ \omega^{1/2}} \big \|_{L^2(\underline{\mathbb{R}}^3)}  \left(\int d\underline{k} \|   (H_{f} +\rho +\omega )^{-1/2} \omega^{1/2} a(.)\psi \|^2 \right)^{1/2}\\
&\leq  2   \| \Lambda( e^{-  i \vartheta}  \cdot) \|_{L^2(\mathbb{R}^3)}\|\psi \|
\end{align*}
where the last line comes from the equality
\begin{equation*}
\vert \vert (H_{f} +\rho +\omega(k))^{-1/2} \omega^{1/2} a(k)\psi \vert \vert^{2}  = \langle   (H_{f} +\rho )^{-1/2} \psi \vert \text{ } \omega(k) a^{*}(k) a(k)  (H_{f} +\rho)^{-1/2} \psi \rangle.  
\end{equation*}
This proves \eqref{HI1}. 
\vspace{2mm}

\subsubsection{Proof of the estimate \eqref{Htheta0} }
Let  $(\vec{p},z) \in \mathcal{U}_{\rho_0}[E_{i_{0}}]$. We have that
\begin{equation}
[H_{\theta,0}(\vec{p}) - z] _{ | \text{Ran}(\overline{\bm{\chi}}_{i_{0}})}= \sum_{j=1}^{N} P_j \otimes b_j(\vec{p},z,H_f,\vec{P}_f),
\end{equation}
where
\begin{align}
b_{i_{0}}(\vec{p},z,r,\vec{l})&= \left(e^{-2 \theta}  \frac{\vec{l}^2}{2} + e^{- \theta} r - e^{-\theta} \vec{p} \cdot \vec{l}   + E_{i_{0}} -z \right) \mathds{1}_{r  \geq 3 \rho_0/4 }, \label{eq:bj_n1} \\[4pt]
b_j(\vec{p},z,r,\vec{l})&=   e^{-2 \theta} \frac{\vec{l}^2}{2}  + e^{- \theta} r - e^{-\theta} \vec{p} \cdot \vec{l}   + E_j -z   \quad \text{}  (j \neq i_0).  \label{eq:bj_n2}
\end{align}

\noindent Any vector $\varphi \in \mathcal{H}_f$ can be represented as a sequence $(\varphi^{(n)})$ of completely symmetric functions of momenta, $\varphi^{(n)} \in L^{2}_{s}(\underline{\mathbb{R}}^{3n})$.  The operators $b_j(\vec{p},z,H_f,\vec{P}_f)$  are multiplication operators in this representation.  Therefore, we  only need to show that $\vert b_{i_{0}}(\vec{p},z,r,\vec{l}) \vert$ and $\vert b_j(\vec{p},z,r,\vec{l}) \vert$, $j \neq i_0$,  are bounded  below by  strictly positive constants.  This amounts to estimate the distance between $z$ and  the range of $b_{i_0}(\vec{p},0, \cdot , \cdot)$ and $b_{j}(\vec{p},0, \cdot , \cdot)$, $j \neq i_0$; see Figure \ref{fig12} below.  
\begin{figure}[H]
\begin{center}
\begin{tikzpicture}
 
       \draw[->] (0,0)--(14,0.0);

             \fill[ color=black] (3,0) circle (0.05) ;
               \fill[ color=black] (7,0) circle (0.05) ;
                    \fill[ color=black] (11,0) circle (0.05) ;

           \draw[ color=black,dotted] (7,0) circle (0.7) ;

          \fill[color=black!25]  (3,0)to[bend left=0] (5.5,-2)  to[bend left=0] (3.6,-2.)   to[bend left=-10] (3,0);              
         \fill[color=black!25]     (7.5,-1)to[bend left=-15] (8.1,-0.75) to[bend left=0] (9.5,-2)  to[bend left=0] (7.7,-2.)   to[bend left=-5] (7.5,-1);    
            \fill[color=black!25]     (11,0)to[bend left=0] (13.5,-2)  to[bend left=0] (11.6,-2.)   to[bend left=-10] (11,0);    
            
             \draw[-,dashed] (3,0)--(5,-2.0);
              \draw[-,dashed] (7,0)--(9,-2.0);
                \draw[-,dashed] (11,0)--(13,-2.0);
                
           \draw[-,dashed] (3,0)--(3.3,-2.0);
              \draw[-,dashed] (7,0)--(7.3,-2.0);
                \draw[-,dashed] (11,0)--(11.3,-2.0);

    \draw (1,-.06) node[above] {...};
     \draw (3,-.06) node[above] {\small$ E_{i_0-1}$};
         \draw (7,-.06) node[above] {\small$ E_{i_0}$};
                     \draw (11,-.06) node[above] {\small$ E_{i_0 + 1}$};
                       \draw (13,-.06) node[above] { ... };

       \draw[->,color=black]  (3.4,0)to[bend left=20] (3.2,-0.2);
       \draw[->,color=black]  (3.8,0)to[bend left=20] (3.1,-0.5);
    \draw (3.4,0.1) node[below] { \footnotesize $\vartheta$};     
     \draw (3.7,-0.2) node[below] { \footnotesize $ 2 \vartheta$};

       \draw[<->,color=black]  (6.95,-0.05)--(6.7,-0.6);
       \draw (6.7,-0.4) node[above] { \footnotesize $r_0$}; 
     
\end{tikzpicture}
\end{center}
\caption{ \small The spectral parameter $z$ is located inside the disk $D(E_{i_0},r_0)$ of center $E_{i_0}$ and radius $r_0$. $\vert E_{i_0}- E_j \vert \geq \delta_0$ for all $j \neq i_0$. The grey shaded regions contain the range of $b_{i_0}(\vec{p},0, \cdot , \cdot)$ and $b_{j}(\vec{p},0, \cdot , \cdot)$, $j \neq i_0$.} \label{fig12}
\end{figure}
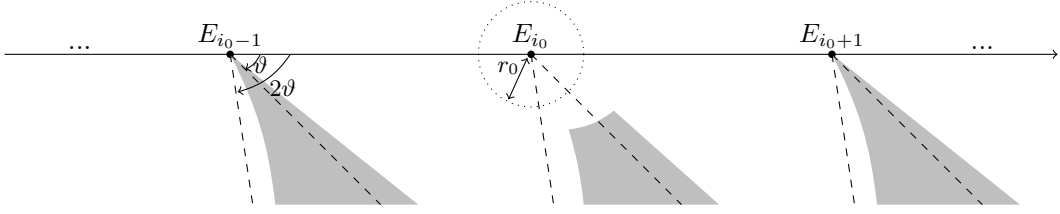
 Let $r \geq  \frac{3 \rho_0}{4}$ and $\vert \vec{l} \vert \leq r$.  Estimating $ \vert b_{i_{0}}(\vec{p},z,r,\vec{l})\vert $ from below by the absolute value of its real part, we obtain
\begin{align*}
\vert b_{i_{0}}(\vec{p},z,r,\vec{l})\vert & \geq \Big \vert \cos(2 \vartheta)  \frac{\vec{l}^2}{2}  + \cos(\vartheta) \left(r-  \Re(\vec{p}) \cdot \vec{l}\right) - \Im(\vec{p}) \cdot \vec{l}  \sin(\vartheta) + E_{i_{0}} - \Re(z) \Big \vert\\[4pt]
& \geq r \cos(\vartheta) \left( 1 -  \vert \vec{p} \vert    -  \vert \Im(\vec{p}) \vert  \tan(\vartheta) \right) - \vert E_{i_{0}} -z \vert , 
\end{align*}
as  $\theta=i \vartheta$, with $0<\vartheta  <\pi/4$.  Since $ \vert \vec{p} \vert   < 1- \mu$, and $\vert \Im \vec{p}  \vert  < \mu  \tan(  \vartheta)/2 < \mu/2$, we finally find that 
  \begin{align*}
\vert b_{i_{0}}(\vec{p},z,r,\vec{l}) \vert & \geq r \cos(\vartheta)  \frac{\mu}{2} -  \frac{\rho_0 \mu }{32}  >\frac{ \rho_0 \mu  }{8}.
\end{align*}

\noindent For $j<i_0$, the difference $E_j-z$  can  be cancelled by  $e^{-2 \theta}  \frac{\vec{l}^2}{2} + e^{- \theta} r - e^{-\theta} \vec{p} \cdot \vec{l}  $ if $\rho_0$ is not small enough, and we have to impose a constraint on $\rho_0$ that depends on the minimal separation between the eigenvalues, $\delta_0$.  Let  $\rho_0 < \delta_0$. We assume that $0<\delta_0<1$. For $\delta_0 \geq 1$, it suffices to replace $\delta_0$ by $1$ in the next estimates.   We split the interval $[0, \infty )$ into two disjoints subintervals $I_1$ and $I_2$, with $I_1=[0, \delta_0/8 ]$ and $I_2 = (\delta_0/8 , \infty )$. We give a lower bound for $b_j$ on both intervals. Let $r \in I_1$. We have that
\begin{equation}
\label{pardel1}
\begin{split}
\vert b_j(\vec{p},z,r,\vec{l}) \vert & \geq  \vert E_{i_{0}}- E_j \vert - \vert z- E_{i_{0}} \vert - r \cos(\vartheta) -  r \left( \vert \vec{p} \vert \cos(\vartheta) + \vert \Im(\vec{p}) \vert \sin(\vartheta) + \frac{r}{2} \right)\\
& > \vert E_{i_{0}}- E_j \vert  - \frac{\delta_0}{2}  \geq \frac{\delta_0}{2}.
\end{split}
\end{equation}

\noindent Let $r \in I_2$. We estimate $| b_j(\vec{p},z,r,\vec{l}) |$ from below by the absolute value of its imaginary part.  An easy calculation shows that
\begin{equation*}
\vert b_j(\vec{p},z,r,\vec{l}) \vert \geq r \sin(\vartheta) \left( 1- \vert \vec{p} \vert  -  \vert  \Im(\vec{p}) \vert \cot(\vartheta) \right)- \frac{\rho_0 \mu \sin(\vartheta)}{32}.
\end{equation*}
Since we have chosen $ \vert  \Im(\vec{p}) \vert  \le \mu \tan(\vartheta)/2$,  it follows that 
\begin{equation}
\label{pardel2}
\vert b_j(\vec{p},z,r,\vec{l}) \vert \ge \frac{\delta_0 \mu \sin(\vartheta)}{32}.
\end{equation}
We  deduce that $[H_{\theta,0}(\vec{p})- z]_{\mid \text{Ran}(\overline{\bm{\chi}}_{i_{0}})}$ is bounded invertible and that   its  inverse satisfies  \eqref{Htheta0}. 
The proof of  \eqref{Htheta1} is similar.
  \vspace{1mm}

\section{Proof of Lemma \ref{borne0}}\label{AppendixC}
\subsection{Proof of the estimates}
\label{AppA}
To prove Lemma \ref{borne0}, we re-Wick order the product of creation and annihilation operators that appear in 
\begin{equation}\label{eq:A1}
P_{i_0} \otimes H^{(0)}(\vec{p},z)=F_{\bm{\chi}_{i_{0}}}(H_{\theta}(\vec{p}) - z,  H_{\theta,0}(\vec{p}) - z)_{\mid P_{i_0} \otimes \mathds{1}_{H_f \leq \rho_0}[\mathcal{H}_{\vec{p}}]}.
\end{equation}
We recall that we want to find a sequence of  kernels $(w_{M,N}^{(0)})$, $M+N \geq 0$, with $w^{(0)}_{M,N}:  \mathcal{U}_{\rho_0}[E_{i_0}] \times \mathcal{B}_{\rho_0} \times \underline{B}_{\rho_0}^{(M,N)} \rightarrow \mathbb{C}$,  and a map $\mathcal{E}^{(0)}:  \mathcal{U}_{\rho_0}[E_{i_0}] \rightarrow \mathbb{C}$, such that
\begin{equation}
H^{(0)}(\vec{p},z)= \sum_{M+N \geq 0} W_{M,N}^{(0)}(\vec{p},z) + \mathcal{E}^{(0)}(\vec{p},z).
\end{equation}
We remind the reader  that the operators $W_{M,N}^{(0)}(\vec{p},z)$ are defined in the sense of quadratic forms by  
\begin{equation*}
W^{(0)}_{M,N}(\vec{p},z) =    \mathds{1}_{H_f \leq \rho_0}  \left( \int_{\underline{ B }_{\rho_0}^{(M,N)}} d \underline{K}^{(M,N)} a^{*}(\underline{k}^{(M)}) w^{(0)}_{M,N}\left(\vec{p},z,H_f,\vec{P_f},  \underline{K}^{(M,N)}\right) a(\underline{\tilde{k}}^{(N)}) 
 \right)    \mathds{1}_{H_f \leq \rho_0},
\end{equation*}
for all $M+N \geq 1$.

 For any bounded  operator $A$ on $\mathbb{C}^N \otimes \mathcal{H}_f$, we  denote by  $ \langle A \rangle_{i_{0}}$ the bounded operator on $\mathcal{H}_f$ associated to the bounded quadratic form
\begin{equation}
(\psi,\phi) \mapsto \langle  \psi_{i_0}  \otimes  \psi \vert  A (  \psi_{i_0}  \otimes  \phi ) \rangle \in \mathbb{C}.
\end{equation}
 The bounded operator $H^{(0)}(\vec{p},z)$ in \eqref{eq:A1} identifies with the following operator (denoted by the same symbol) on $\mathcal{H}_f$:
\begin{equation} \label{C4}
\begin{split} 
H^{(0)}(\vec{p},z)&=  \Big (  \left(H_f - \vec{p} \cdot \vec{P}_f   \right) e^{- \theta} + \frac{\vec{P}_{f}^2 }{2}e^{- 2 \theta}  + E_{i_{0}} - z \Big ) \mathds{1}_{H_f \leq \rho_0} + \lambda_0  \langle \bm{\chi}_{i_{0}} H_{I,\theta} \bm{\chi}_{i_{0}}\rangle_{i_{0}}\\[6pt]
&- \lambda_{0}^{2}   \langle   \bm{\chi}_{i_{0}} H_{I,\theta}  \bm{\overline{\chi}}_{i_{0}}  [H_{\bm{\overline{\chi}}_{i_{0}} }(\vec{p},z)]^{-1}_{| \text{Ran}(\bm{\overline{\chi}}_{i_{0}})}  \bm{\overline{\chi}}_{i_{0}} H_{I,\theta}  \bm{\chi}_{i_{0}} \rangle_{i_{0}} .
\end{split}
\end{equation}
The operators on the first line are already  Wick-ordered. The first operator contributes to $W^{(0)}_{0,0}(\vec{p},z) + \mathcal{E}^{(0)}(\vec{p},z)$,  and  the second to $W^{(1)}_{1,0}(\vec{p},z)$ and $W^{(1)}_{0,1}(\vec{p},z)$. To normal order the operator on the second line, we use the Neumann expansion for $ [H_{\bm{\overline{\chi}}_{i_{0}} }(\vec{p},z)]^{-1}_{|\mathrm{Ran}(\bm{\overline{\chi}}_{i_{0}})}$.  It is easy to show that
\begin{equation}
\langle  \bm{\chi}_{i_{0}} H_{I,\theta}  \bm{\overline{\chi}}_{i_{0}}  [H_{\bm{\overline{\chi}}_{i_{0}} }(\vec{p},z)]^{-1}_{|\mathrm{Ran}(\bm{\overline{\chi}}_{i_{0}})}  \bm{\overline{\chi}}_{i_{0}} H_{I,\theta}  \bm{\chi}_{i_{0}}  \rangle_{i_{0}} = \sum_{L=2}^{\infty} (- \lambda_0)^{L-2}  \tilde{V}_L(\vec{p},z),
  \end{equation}
 where $$  \tilde{V}_{L}(\vec{p},z):=  \langle  \bm{\chi}_{i_{0}}  \left(   H_{I,\theta} \bm{\overline{\chi}}_{i_0}^{2}  [H_{\theta,0}(\vec{p})-z]_{|\mathrm{Ran}(\bm{\overline{\chi}}_{i_{0}})}^{-1}  \right)^{L-1}   H_{I,\theta}     \bm{\chi}_{i_{0}} \rangle_{i_{0}}.$$ We first split $ H_{I,\theta}$ into the sum of two operators $ (H_{I,\theta})^{0,1}$ and $ (H_{I,\theta})^{1,0}$, where
\begin{align}
(H_{I,\theta})^{0,1}&=  i  \int_{\underline{\mathbb{R}}^3} d \underline{k} \, \vert \vec{k} \vert^{1/2}  e^{- 2 \theta}  \Lambda(e^{-\theta} \vec{k})  \vec{\epsilon}(\vec{k}) \cdot \vec{d} \text{ } a(\underline{k}),\\ 
(H_{I,\theta})^{1,0}&= -i  \int_{\underline{\mathbb{R}}^3} d \underline{k} \, \vert \vec{k} \vert^{1/2}  e^{- 2 \theta}  \Lambda(e^{-\theta} \vec{k})  \vec{\epsilon}(\vec{k}) \cdot \vec{d}   \text{ } a^{*}(\underline{k}).
\end{align}
This yields
\begin{equation}\label{avantpool}
  \tilde{V}_{L}(\vec{p},z) = \underset{ \tiny \begin{array}{c} p_1,...,p_{L} \\q_1,...,q_{L}  \\   p_i+q_i =1\end{array}}{\sum}   \langle  \bm{\chi}_{i_{0}}   \prod_{j=1}^{L-1} \left( (H_{I,\theta})^{p_j,q_j}  \bm{\overline{\chi}}_{i_0}^{2}     [H_{\theta,0}(\vec{p})-z]_{|\mathrm{Ran}(\bm{\overline{\chi}}_{i_{0}})}^{-1} \right)  (H_{I, \theta})^{p_{L},q_{L}}    \bm{\chi}_{i_{0}} \rangle_{i_{0}}. 
\end{equation}
Let $j \in \{1,...,L\}$.  The operator-valued distribution $a(\underline{k}_j)$/$a^*(\underline{k}_j)$ that appears in $ (H_{I,\theta})^{p_j,q_j} $ can  either be  contracted with another creation/annihilation operator appearing in  $(H_{I,\theta})^{p_{j'},q_{j'}} $, $j' \neq j$,  or left uncontracted.  In the latter case, we pull it  to the right  of \eqref{avantpool} if it is an annihilation operator, or to the left of  \eqref{avantpool}   if it is a  creation operator. This modifies the operators $\bm{\overline{\chi}}_{i_0}^{2}     [H_{\theta,0}(\vec{p})-z]_{|\mathrm{Ran}(\bm{\overline{\chi}}_{i_{0}})}^{-1}$ via  the pull-though formula given in \eqref{pull}. We introduce the notations
\begin{align*}
(H_{I,\theta})^{0 0}_{01}&= i \vert \vec{k} \vert^{1/2} e^{- 2 \theta} \Lambda(e^{-\theta } \vec{k}) \vec{\epsilon}(\vec{k}) \cdot \vec{d},\\
(H_{I,\theta})^{0 0}_{10}& = - i   \vert \vec{k} \vert^{1/2} e^{- 2 \theta} \Lambda(e^{-\theta } \vec{k}) \vec{\epsilon}(\vec{k}) \cdot \vec{d},\\
(H_{I,\theta})^{1 0}_{00}&= (H_{I,\theta})^{1,0}, \qquad  (H_{I,\theta})^{01}_{00} =(H_{I,\theta})^{0,1}.
\end{align*}
The contracted part  is  expressed as  a  vacuum expectation value,   and 
\begin{equation}
\tilde{V}_{L}(\vec{p},z) =  \underset{\underset{m_i+n_i+p_i+q_i=1}{\underline{m},\underline{n},\underline{p},\underline{q}}}{\sum} \int a^{*}(\underline{k}^{(\underline{m})})  \langle \psi_{i_{0}} \otimes \Omega  \vert  V^{\underline{m+p}, \underline{n+q} }_{\underline{m}, \underline{n}}  (\vec{p},z,H_f,\vec{P}_f, \underline{K}^{(\underline{m},\underline{n})}) ( \psi_{i_{0}} \otimes \Omega) \rangle a(\underline{\tilde{k}}^{(\underline{n})}) d \underline{K}^{(\underline{m}, \underline{n})}
\end{equation}
where $\underline{m} = (m_1,\dots,m_L)$,
 \begin{equation}
\label{forms2}
\begin{split}
& V^{\underline{m+p}, \underline{n+q} }_{\underline{m}, \underline{n}}  (\vec{p},z,r,\vec{l}, \underline{K}^{(\underline{m},\underline{n})}) =\chi_{\rho_{0}}(r+\tilde{r}_0) \prod_{j=1}^{L-1} \Big[ (H_{I,\theta})_{m_j n_j}^{ p_j  q_j}\\
& \qquad \qquad R(\vec{p},z,H_f + r + \tilde{r}_j, \vec{P}_f+\vec{l} + \vec{\tilde{l}}_j) \Big]   (H_{I,\theta})_{m_{L} n_{L}}^{ p_{L}   q_{L}}   \chi_{\rho_{0}}(r + \tilde{r}_{L}) ,
\end{split}
\end{equation}
and 
\begin{equation}
R(\vec{p},z,r, \vec{l}):= \sum_{j \neq i_0}  P_j \otimes [b_j(\vec{p},z,H_f,\vec{P}_f)]^{-1} + P_{i_0} \otimes \overline{\chi}_{\rho_0}^2(H_f) [b_{i_0}(\vec{p},z,H_f,\vec{P}_f)]^{-1};\end{equation}
(see Eqs. (B.2)-(B.4)).
The shifts  $r_i$'s, $\tilde{r}_i$'s, $\vec{l}_i$'s, $\vec{\tilde{l}}_i$'s come from the Pull-through formula and are given by
 \begin{align}
&r_{i} := \sum_{j=1}^{i-1} \Sigma [ \underline{\tilde{k}}_j^{(n_j)} ] +\sum_{j=i+1}^L \Sigma [ \underline{k}_j^{(m_j)} ] , \quad \vec{l}_{i} := \sum_{j=1}^{i-1} \vec{\Sigma} [ \underline{\tilde{k}}_j^{(n_j)} ] +\sum_{j=i+1}^L \vec{\Sigma} [ \underline{k}_j^{(m_j)} ], \notag \\
&\tilde{r}_{i} := \sum_{j=1}^{i} \Sigma [ \underline{\tilde{k}}_j^{(n_j)} ] +\sum_{j=i+1}^L \Sigma [ \underline{k}_j^{(m_j)} ] , \quad \vec{\tilde{l}}_{i} := \sum_{j=1}^{i} \vec{\Sigma} [ \underline{\tilde{k}}_j^{(n_j)} ] +\sum_{j=i+1}^L \vec{\Sigma} [ \underline{k}_j^{(m_j)} ]. \label{def_mu_L}
\end{align}
where $\Sigma [ \underline{k}^{(m)} ] := \sum_{i=1}^m | \vec{k}_i |$, $\Sigma[ \underline{\tilde{k}} \phantom{}^{(n)} ] := \sum_{j=1}^n | \vec{\tilde{k}}_j |$, $\vec{\Sigma} [ \underline{k}^{(m)} ] := \sum_{i=1}^m \vec{k}_i$, $\vec{\Sigma}[ \underline{\tilde{k}} \phantom{}^{(n)} ] := \sum_{j=1}^n \vec{\tilde{k}}_j$.  We  rewrite 
\begin{equation}
- \lambda_{0}^{2} \langle \bm{\chi}_{i_{0}} H_{I,\theta}  \bm{\overline{\chi}}_{i_{0}}  [H_{\bm{\overline{\chi}}_{i_{0}} }(\vec{p},z)]^{-1}_{|\mathrm{Ran}(\bm{\overline{\chi}}_{i_{0}})}  \bm{\overline{\chi}}_{i_{0}} H_{I,\theta}  \bm{\chi}_{i_{0}} \rangle_{i_0} = \sum_{M+N \geq 0} \tilde{W}_{M,N}^{(0)} (\vec{p},z) ,
\end{equation}
 where the kernels $\tilde{w}_{M,N}^{(0)}$ are given by  
\begin{equation}
\label{noynoy}
\begin{split}
\tilde{w}_{M,N}^{(0)} & (\vec{p},z,r,\vec{l}, \underline{k}^{(M)},\underline{\tilde{k}}^{(N)})=\\
&- \sum_{L=2} (-\lambda_0)^{L} \underset{ \tiny \begin{array}{c} \underline{m},  \underline{n}, \underline{p}, \underline{q} \\ m_1+...+m_L =M \\ n_1+...+n_L =N \\ m_i+n_i+p_i+q_i=1\end{array}}{\sum}  \langle \psi_{i_{0}} \otimes  \Omega \vert V^{\underline{m+p}, \underline{n+q} }_{\underline{m}, \underline{n}}    (\vec{p},z,r,\vec{l}, \underline{K}^{(M,N)})   \psi_{i_{0}} \otimes  \Omega \rangle_{\text{sym}},
\end{split}
\end{equation}
for $M + N \geq 1$, and 
\begin{equation}
\label{noynoy0}
\begin{split}
\tilde{w}_{0,0}^{(0)} ( r,\vec{l})&=-\sum_{L=2} (-\lambda_0)^{L} \underset{ \tiny \begin{array}{c} \underline{p},  \underline{q}\\  p_i+q_i=1\end{array}}{\sum}  \langle \psi_{i_{0}} \otimes \Omega \vert V^{\underline{p}, \underline{q} }_{\underline{0}, \underline{0}}    (\vec{p},z,r,\vec{l})   \psi_{i_{0}} \otimes \Omega \rangle
\end{split}
\end{equation}
for $M+N=0$. $f_{\text{sym}}$ denotes the symmetrization of $f$ with respect to the variables  $\underline{k}^{(M)}$ and $\underline{\tilde{k}}^{(N)}$. Using the estimate \eqref{HI1},
\begin{equation}
\|(H_f + \rho_0)^{-1/2}  H_{I,\theta} (H_f + \rho_0)^{-1/2}\|  = \mathcal{O}( \rho_{0}^{-1/2} \sigma_{\Lambda}^{3/2}),
\end{equation}
we obtain the bound 
\begin{equation}
\label{forms3}
\begin{split}
& \| V^{\underline{m+p}, \underline{n+q} }_{\underline{m}, \underline{n}}  (\vec{p},z,r,\vec{l}, \underline{K}^{(\underline{m},\underline{n})})\| \leq \rho_0  \prod_{j=1}^{L}  \Big \|  (H_f + \rho_0)^{-1/2}  (H_{I,\theta})_{m_j n_j}^{p_j q_j}   (H_f + \rho_0)^{-1/2}  \Big \|\\
& \qquad \qquad \sup_j \Big \|   (H_f + \rho_0) R(\vec{p},z,H_f + r + \tilde{r}_j, \vec{P}_f+\vec{l} + \vec{\tilde{l}}_j) \Big \|^{L-1}.  
\end{split}
\end{equation}
Distinguishing the cases $p_j+q_j=1$ and $p_j+q_j=0$ and using Eqs. \eqref{HI1} and  \eqref{Htheta1} of Lemma \ref{tech1}, we deduce that  there exists a positive constant $C$ such that 
\begin{equation}
\label{forms4}
\begin{split}
& \| V^{\underline{m+p}, \underline{n+q} }_{\underline{m}, \underline{n}}  (\vec{p},z,r,\vec{l}, \underline{K}^{(\underline{m},\underline{n})})\| \leq \rho_0   \frac{C^L}{ (\mu \sin(\vartheta))^{L-1}} \prod_{j=1}^{L}  ( \sigma_{\Lambda}^{3/2} \rho_{0}^{-1/2})^{p_j+q_j}  \rho_{0}^{-(m_j+n_j)} \vert \underline{k}_{j}^{(m_j)} \vert^{1/2}  \vert \underline{\tilde{k}}_{j}^{(n_j)} \vert^{1/2}. 
\end{split}
\end{equation}
Since $m_j+n_j+p_j+q_j = 1$, it follows  that  (up to a multiplication of $C$ by a numerical factor)
\begin{equation}
\| \tilde{w}_{M,N}^{(0)} (\vec{p},z) \|_{\frac{1}{2}} \leq \sum_{L=2}^{\infty} \frac{(    C  \sigma_{\Lambda}^{3/2} \lambda_0 \rho_{0}^{-1/2})^{L}} { (\mu \sin(\vartheta))^{L-1}}  \rho_{0}^{- \frac12(M+N)} \rho_0.
\end{equation}
Similarly,
\begin{equation}
\| \tilde{w}_{0,0}^{(0)} (\vec{p},z) \|_{\infty} \leq \sum_{L=2}^{\infty} \frac{(   C  \sigma_{\Lambda}^{3/2} \lambda_0 \rho_{0}^{-1/2})^{L}} { (\mu \sin(\vartheta))^{L-1}}  \rho_0 ,
\end{equation}
and proceeding in the same way for the derivatives $\partial_\# \tilde{w}_{M,N}^{(0)}$, where $\partial_\#$ stands for $\partial_r$ or $\partial_{l_j}$, we conclude that
\begin{align}
\| \tilde{w}_{M,N}^{(0)} (\vec{p},z) \|_{\frac{1}{2}} &= \mathcal{O} \left( \frac{   \lambda_{0}^2  \sigma_{\Lambda}^{3} } {  \mu \sin(\vartheta) \rho_{0}^{\frac12(M+N)}} \right)  , \\[4pt]
\| \partial_{\#} \tilde{w}_{M,N}^{(0)} (\vec{p},z) \|_{\frac{1}{2}} &=  \mathcal{O} \left( \frac{   \lambda_{0}^2   \sigma_{\Lambda}^{3} } {  \mu^2    \sin^2 (\vartheta) \rho_{0}^{ \frac12(M+N) + 1 } } \right)  , 
\end{align}
uniformly with respect to $M+N \ge 1$, and that
\begin{align}
\| \tilde{w}_{0,0}^{(0)} (\vec{p},z) \|_{\infty} &=  \mathcal{O} \left( \frac{   \lambda_{0}^2  \sigma_{\Lambda}^{3}  } { \mu \sin(\vartheta)}\right), \\[4pt]
\|\partial_{\#} \tilde{w}_{0,0}^{(0)} (\vec{p},z) \|_{\infty}&=   \mathcal{O} \left(  \frac{   \lambda_{0}^2  \sigma_{\Lambda}^{3}  } { \mu^2 \sin^2(\vartheta) \rho_0}\right).
\end{align}
Since $\lambda_0 \ge 0$ is chosen such that  $\lambda_0 \ll \mu \sin ( \vartheta ) \rho_0^{1/2} \sigma_{\Lambda}^{3/2}$, this concludes the proof of the estimates of the lemma.

\subsection{Proof of the analyticity of the map  $(\vec{p},z) \mapsto H^{(0)}(\vec{p},z) \in \mathcal{L}(\mathcal{H}^{(0)})$} \label{Analyfirst}
We start from the expression of $H^{(0)}(\vec{p},z)$ given in \eqref{C4},
\begin{equation}
\begin{split}
\label{2.25}
H^{(0)}(\vec{p},z)&=  \Big(  \left(H_f -  \vec{p} \cdot \vec{P}_f   \right) e^{- \theta} +  \frac{\vec{P}_{f}^2}{2} e^{- 2 \theta}  + E_{i_{0}} - z \Big) \mathds{1}_{H_f \leq \rho_0} + \lambda_0 \langle \bm{\chi}_{i_{0}} H_{I,\theta} \bm{\chi}_{i_{0}} \rangle_{i_0}\\[4pt]
&- \lambda_{0}^{2}    \langle  \bm{\chi}_{i_{0}} H_{I,\theta}  \bm{\overline{\chi}}_{i_{0}}  [H_{\bm{\overline{\chi}}_{i_{0}} }(\vec{p},z)]^{-1}_{|\mathrm{Ran}(\bm{\overline{\chi}}_{i_{0}})}  \bm{\overline{\chi}}_{i_{0}} H_{I,\theta}  \bm{\chi}_{i_{0}} \rangle_{i_0} .
\end{split}
\end{equation}
The first term  on the right-hand side of  \eqref{2.25} is analytic on $\mathcal{U}_{\rho_0}[E_{i_{0}}]$. We have seen in the proof of Lemma \ref{fespair} that the Neumann series for $[H_{\overline{\bm{\chi}}_{i_{0}}}(\vec{p},z)]^{-1}_{|\mathrm{Ran}(\overline{\bm{\chi}}_{i_{0}})}$ is uniformly convergent on $\mathcal{U}_{\rho_0}[E_{i_{0}}]$. It is therefore sufficient to  check that the map 
\begin{equation*}
(\vec{p},z) \mapsto \left[ \frac{H_f + \rho_0}{H_{\theta,0} (\vec{p},z)} \right]_{ \overline{\mathrm{Ran}(\overline{\bm{\chi}}_{i_{0}})}} \in \mathcal{L}( \overline{\mathrm{Ran}(\overline{\bm{\chi}}_{i_{0}})})
\end{equation*}
is analytic on $\mathcal{U}_{\rho_0}[E_{i_{0}}]$. Since weak and strong analyticity are equivalent, we only need to show that the maps
\begin{equation}
(\vec{p},z) \mapsto  \big \langle  \psi_j  \otimes \varphi \big \vert \left[ \frac{H_f + \rho_0}{H_{\theta,0} (\vec{p},z)} \right]_{|\overline{\mathrm{Ran}(\overline{\bm{\chi}}_{i_{0}})}}  \text{ }  \psi_j  \otimes   \varphi \big \rangle \in \mathbb{C}
\end{equation}
are analytic for any $\psi_j \otimes \varphi \in \overline{\text{Ran}(\overline{\bm{\chi}}_{i_{0}})}$, where $\psi_j$, $j=1,...,N$,  are unit eigenvectors of  $H_{is}$ associated to the eigenvalues $E_j$. Using, for any $\varphi$ in Fock space, the representation $\varphi= (\varphi^{(n)})$ with $\varphi^{(n)} \in L_{s}^{2}(\underline{\mathbb{R}}^{3n})$, we find that
\begin{align*}
 \langle \psi_j  \otimes  \varphi \big  \vert  &\left[ \frac{H_f + \rho_0}{H_{\theta,0} (\vec{p},z)} \right]_{| \overline{\text{Ran}(\overline{\bm{\chi}}_{i_{0}})}}   \psi_j  \otimes  \varphi \big \rangle   = (1- \delta_{i_0 j})  \frac{\rho_0}{E_j -z} | \varphi^{(0)}|^2 \qquad  \\
 & + \sum_{n \geq 1}  \int_{\underline{\mathbb{R}}^{3n}} d\underline{k}^{(n) } \frac{\vert \vec{k}_1 \vert + \cdots +\vert \vec{k}_n  \vert+ \rho_0  }{b_{j}(   \vec{p},z,\Sigma[ \underline{k}^{(n) }] , \vec{\Sigma}[ \underline{k}^{(n) }])} \vert \varphi^{(n)} (\underline{k}_1,..., \underline{k}_n) \vert^2,
\end{align*}
where $b_j$ is defined in \eqref{eq:bj_n1}--\eqref{eq:bj_n2}, and we have set $ \Sigma[ \underline{k}^{(n) }] =\vert \vec{k}_1 \vert + \cdots +\vert \vec{k}_n  \vert$ and  $\vec{\Sigma}[ \underline{k}^{(n) }] = \vec{k}_1+ \cdots +\vec{k}_n$. The functions $(\vec{p},z) \mapsto (r + \rho_0)/b_j(\vec{p},z,r,\vec{l})$ are analytic on $\mathcal{U}_{\rho_0}[E_{i_{0}}]$ for any fixed $(r,\vec{l}) \in \mathbb{R}_+ \times \mathbb{R}^3$ with $\vert \vec{l} \vert \leq r$, and $| (r + \rho_0)/b_j(\vec{p},z,r,\vec{l}) |$ is uniformly bounded as follows from the proof of Lemma \ref{fespair}. Using Morera's theorem for several complex variables (see \cite{Vlad}), we deduce that 
\begin{equation}
(\vec{p},z) \mapsto  \big \langle  \psi_j  \otimes \varphi \big \vert \left[ \frac{H_f + \rho_0}{H_{\theta,0} (\vec{p},z)} \right]_{|\overline{\text{Ran}(\overline{\bm{\chi}}_{i_{0}})}}     \psi_j  \otimes \varphi \big \rangle \in \mathbb{C}
\end{equation}
is analytic, and hence, that $H^{(0)}(\vec{p},z)$ is analytic in $(\vec{p},z) \in \mathcal{U}_{\rho_0}[E_{i_{0}}]$.

\nocite{*}
\bibliographystyle{plain}
\bibliography{main}

\end{document}